\documentclass[12pt]{amsart}
\usepackage{amsmath,amsfonts,amssymb} \usepackage{color}
\usepackage[all,cmtip]{xy}
\usepackage{graphicx}
 \usepackage{youngtab}
\usepackage{young}

\newcommand{\corr}[1]{\langle {#1} \rangle}

 \newcommand{\bR}{\mathbb{R}}  
 
\newcommand{\bt}{{\bf t}}

  \newcommand{\cA}{\mathcal{A}}
 \newcommand{\cB}{\mathcal{B}}  
 \newcommand{\cG}{\mathcal{G}} 
  \newcommand{\cS}{\mathcal{S}}
 \newcommand{\cV}{\mathcal{V}} 
 \newcommand{\bZ}{\mathbb{Z}}
 \newcommand{\bC}{\mathbb{C}}
 \newcommand{\pd}{\partial}

 \newcommand{\bu}{\mathbf{u}}
\newcommand{\vac}{|0\rangle} \newcommand{\lvac}{\langle 0|}

 \DeclareMathOperator{\Aut}{Aut} \DeclareMathOperator{\res}{res}

 \DeclareMathOperator{\val}{val}

\newcommand{\be}{\begin{equation}}
\newcommand{\ee}{\end{equation}}
\newcommand{\bea}{\begin{eqnarray}}
\newcommand{\eea}{\end{eqnarray}}
\newcommand{\ben}{\begin{eqnarray*}}
\newcommand{\een}{\end{eqnarray*}}

\newcommand{\half}{\frac{1}{2}}

\newtheorem{cor}{Corollary}[section]

\newtheorem{lem}[cor]{Lemma}
 \newtheorem{prop}[cor]{Proposition}
 \newtheorem{thm}[cor]{Theorem}
\theoremstyle{remark}
 \newtheorem{defn}[cor]{Definition}
 
 \newtheorem{rmk}[cor]{Remark}

\definecolor{A}{rgb}{.75,1,.75}

\definecolor{green}{rgb}{0,1,0}
\definecolor{yellow}{rgb}{1,1,0}
\definecolor{orange}{rgb}{1,.7,0}
\definecolor{red}{rgb}{1,0,0}
\definecolor{white}{rgb}{1,1,1}

 \makeindex
\begin{document}
\title
{On Topological 1D Gravity. I}

\author{Jian Zhou}
\address{Department of Mathematical Sciences\\Tsinghua University\\Beijng, 100084, China}
\email{jzhou@math.tsinghua.edu.cn}

\begin{abstract}
In topological 1D gravity,
the genus zero one-point function
combined with the gradient of the action function leads to a spectral curve
and its special deformation.
After  quantization,
the partition function is identified as an element in the bosonic Fock space
uniquely specified by the Virasoro constraints.
\end{abstract}

\maketitle

\section{Introduction}

This is the first part of a series of papers in which we will systematically study the topological 1D gravity,
in the framework of emergent geometry  and quantum deformation theory  of its spectral curve.

Our motivation to study topological 1D gravity was originally to gain some more understanding of
topological 2D gravity.
This seems to be also the motivations of earlier works on this subject 
\cite{Nishigaki-Yoneya1, Nishigaki-Yoneya2, Anderson-Myers-Periwal, Di  Vecchia-Kato-Ohta, Eguchi-Yamada-Yang}.
As it turns out,
not only does topological 1D gravity  share similar properties as topological 2D gravity,
such as their connections with integrable hierarchies and their Virasoro constraints,
but also the derivations of these properties are much simpler. 
Furthermore,
it admits some easy generalizations and some connections with other topics to be reported
in later parts of this series that makes it have some independent interests.
This series is a companion series to a series on related work on topological 2D gravity whose first part is
\cite{Zhou}.

More than twenty years ago,
topological 2D gravity was studied intensively from the point of view of double scaling limits
of large N random matrices \cite{Brezin-Kazakov, Douglas-Shenker, Gross-Migdal}.
A remarkable connection with the intersection theory of moduli spaces of algebraic curves was made by
Witten \cite{Witten}.
In his proof of Witten Conjecture,
Kontsevich \cite{Kontsevich} introduced a different kind of matrix models.
Around the same time, topological 1D gravity arose in the context of double scaling limits of large N $O(N)$ vector models
in the references \cite{Nishigaki-Yoneya1, Anderson-Myers-Periwal, Di  Vecchia-Kato-Ohta, Eguchi-Yamada-Yang}
mentioned above.
They were originally called (branching) polymer models or (branching) chain models.
In \cite{Nishigaki-Yoneya2},
it was proposed that polymer model is equivalent to a topological theory of 1D gravity.
The author gets interested in this theory because it provides another example in which
one can study mirror symmetry from the point of view of emergent geometry and quantum deformation theory.

In \cite{Zhou} we have proposed to develop quantum deformation theory as an approach to mirror symmetry.
By this we mean the genus zero free energy on the big phase space (the space of  coupling constants
of all gravitational descendants) leads to a geometric structure and its special deformations,
and by quantization of this picture one gets constraints the free energy in all genera
which suffice to determine the whole partition function.
We will refer to this as the emergent geometry.
Or more precisely,
by emergent geometry we mean the geometric picture better seen or understood
when one goes to the big phase space.
In \cite{Zhou},
we have shown that the emergent geometry of topological 2D gravity is the quantum deformation of theory of
the Airy curve:
$$y = \half x^2.$$
For related results,
see \cite{Eynard, Bennett-Cochran-Safnuk-Woskoff}.
One of the main results of this paper is that the emergent geometry  of topological 2D gravity
is the quantum deformation theory of the signed Catalan curve:
$$
y = - \frac{1}{\sqrt{2}} z + \frac{\sqrt{2}}{z}.
$$

In this paper we will also study the following coordinate change on the big phase space:
\ben
&& I_0 = \sum_{k=1}^\infty \frac{1}{k}
\sum_{p_1 + \cdots + p_k = k-1} \frac{t_{p_1}}{p_1!} \cdots
\frac{t_{p_k}}{p_k!}, \\
&& I_k= \sum_{n \geq 0} t_{n+k} \frac{I_0^n}{n!}, \;\;\;\; k \geq 1.
\een
These series were introduced in \cite{Itzykson-Zuber} to express the free energy in the context of topological 2D gravity.
By understanding them as new coordinates on the big phase space,
one can gain better understanding of the global nature of the behavior of the theory on the big phase space.
For example, we use two different methods to show that for topological 1D gravity,
\ben
&& F_0 = \sum_{k=0}^\infty  \frac{(-1)^k}{(k+1)!} (I_k+\delta_{k,1}) I_0^{k+1}, \\
&& F_1 = \frac{1}{2} \ln \frac{1}{1- I_1}, \\
&& F_g  =  \sum_{\sum\limits_{j=2}^{2g-1}  \frac{j-1}{2} l_j = g-1}
 \corr{\tau_2^{l_2} \cdots \tau_{2g-1}^{l_{2g-1}}}_g
\prod_{j=2}^{2g-1} \frac{1}{l_j!}\biggl( \frac{I_j}{(1-I_1)^{(j+1)/2}}\biggr)^{l_j}, \;\; g \geq 2.
\een
Partly fulfilling our wish to gain more insights on topological 2D gravity by studying topological 1D gravity,
we will in a sequel to \cite{Zhou} use one of these methods to prove that similar formulas hold for topological 2D gravity:
\ben
F_0^{2D} & = & (1- I_1)^2 \frac{I_0^3}{3!} + (1-I_1)I_2 \frac{I_0^4}{4!}
+ I_2^2 \frac{I_0^5}{5!}  \\
& + & \sum_{n \geq 6} (-1)^{n-1} \biggl[(n-5) (1- I_1)I_{n-2} \\
& - & \half \sum_{j=2}^{n-3}
\biggl( \binom{n-3}{j} - 2\binom{n-3}{j-1} +  \binom{n-3}{j-2} \biggr) I_j I_{n-1-j}\biggr] \frac{I_0^n}{n!}, \\
F_1^{2D} & = & \frac{1}{24} \ln \frac{1}{1- I_1}, \\
F_g^{2D} & = & \sum_{\sum\limits_{j=2}^{2g-1}  \frac{j-1}{3} l_j = g-1}  \corr{\tau_2^{l_2} \cdots \tau_{3g-2}^{l_{3g-2}}}_g
\prod_{j=2}^{3g-2}  \frac{1}{l_j!}\biggl( \frac{I_j}{(1-I_1)^{(2j+1)/3}}\biggr)^{l_j},
\een
where $g> 1$.
In this case the formula for $F_1$ was due to \cite{Dijkgraaf-Witten},
the formula for $F_g$ ($g>1$) was conjectured in \cite{Itzykson-Zuber},
and the formula for $F_0$ seems to be new.
Such formulas indicate that one can study the behavior of free energy of topological 1D and 2D gravity
near $t_1 = 1$.
This is where one gets multicritical phenomenon originally studied in the matrix model or vector model approach
\cite{Nishigaki-Yoneya1, Anderson-Myers-Periwal, Di  Vecchia-Kato-Ohta}..
Note we prove in this paper for topological 1D gravity:
\ben
F = \half \log (1-t_1)
+ \sum_{\substack{g,n \geq 0\\ 2g-2+n > 0}} \sum_{\substack{a_2, \dots, a_n \neq 1\\ \sum a_j = 2g-2+n }}
\frac{\corr{\prod\limits_{j=1}^n \tau_{a_j}}_g}{(1-t_1)^{g-1+n}} \prod\limits_{j=1}^n t_{a_j},
\een
and in topological 2D gravity one have similar formula by dilaton equation,
for example,
\be
F_0(\bt) = \sum_{k=1}^\infty \frac{1}{k(k+1)(k+2)(1-t_1)^k}
\sum_{\substack{p_1 + \cdots + p_{k+2} = k-1\\p_j \neq 1,\;\; j=1,\dots, k+2}} \frac{t_{p_1}}{p_1!} \cdots
\frac{t_{p_{k+2}}}{p_{k+2}!}.
\ee
By looking at such formulas one might get the wrong impression that it is impossible to consider
the theories at $t_1=1$.
In later parts of this series and a sequel to \cite{Zhou1},
we will address these issues.
Another result proved in this paper is the 
following analogue of Kontsevich's main identity \cite{Kontsevich}:
\ben
&& \sum \corr{\tau_0^{m_0} \cdots \tau_n^{m_n} }_g  \prod_{j=0}^n \frac{t_j^{m_j}}{m_j!}
=  \sum_{\Gamma \in \cG^c} \frac{1}{|\Aut(\Gamma)|} \prod_{v\in V(\Gamma)} \lambda^{\val(v)-2} t_{\val(v)-1}.
\een

In this part of the series,
we will understand the topological 1D gravity as the $N=1$ case of matrix models for topological gravity
in the physics literature.
More precisely,
we will work with a formal Gaussian integral, with infinitely many parameters
giving the coupling constants to gravitational descendants.
This is the mean field theory of the topological 1D gravity.
This point of view of connecting topological 1D gravity and 2D gravity
makes it possible to generalize to topological 1D gravity coupled with topological matters,
a topic to be discussed in a later part of this series.
The problem of building the theory on integrations over moduli spaces of some geometric objects
will be addressed also in a subsequent part of this series.

As the companion work in topological 2D gravity \cite{Zhou1, Zhou2},
we are inspired by \cite{ADKMV}.
In this series we elaborate on an example not included in their beautiful work.

Let us now sketch the contents of the rest of this paper.
In \S \ref{sec:Renormalization} we explain how the simple idea of completing the squares,
when used repeatedly,
leads to renormalization of coupling constants in topological 1D gravity.
Furthermore,
this process can be embedded in the formal gradient flow
and the limit point is the single critical point of the action function in the formal setting.
The $I$-coordinates naturally appear as the Taylor coefficients at the critical point.

In \S \ref{sec:Feynman rules for I} we reformulate the formula for $I$-coordinates in  $t$-coordinates
in terms of Feynman rules.
This raises the problem of realizing them by some quantum field theory.
We propose the solution in \S \ref{sec:1D-TG} as the mean field theory of topological 1D gravity.
We define and develop this theory based on formal Gaussian integrals and their properties in this Section.
In particular,
the formulas of the free energy in $I$-coordinates will be proved based on translation invariance of
formal Gaussian integrals.
This is a mathematically rigorous approach to the saddle point method in this setting.

We study the applications of flow equations and polymer equation for topological 1D gravity \cite{Nishigaki-Yoneya1}
in \S \ref{sec:Flow-Polymer}.
They are the analogues of the KdV hierarchy and the string equation in topological 2D gravity respectively.
In this Section we present two more derivations of the formula for $F_0$ in $I$-coordinates.

Some applications of Virasoro constraints in topological 1D gravity \cite{Nishigaki-Yoneya1}
will be presented in \S \ref{sec:Virasoro}.
These include our fourth derivation of formula for $F_0$ in $I$-coordinates and our second
derivation of the formula for $F_g$ ($g \geq 1$).

In \S \ref{sec:W-Constraints} we rederive Virasoro constraints from the point of view of operator algebras.
We present W-constraints and another version of Virasoro constraints for topological 1D gravity.
As it turns out,
this second version of Virasoro constraints is in closer analogy with the Virasoro constraints
for topological 2D gravity and it is the version we need to develop the quantum deformation
in a later Section.

We define and compute two kinds of $n$-point functions in topological gravity in \S \ref{sec:N-point function}.
The computations rely heavily on the flow equations.
We also obtain some recursion relation for $n$-point functions.
We derive some Feynman rules for $n$-point functions in \S \ref{sec:Feynman for N-Point}.
These rules expresses the $n$-point functions in terms of genus zero free energy restricted to the small phase  space.
In topological 2D gravity similar rules were proposed in \cite{Dijkgraaf-Witten}.

In \S \ref{sec:Spectral Curve} we show that the genus zero one-point function
combined with the gradient of the action function leads us to the spectral curve
and its special deformation for topological 1D gravity.
We also establish the  uniqueness of the special deformation.
After  quantizing the special deformation of the spectral curve in \S \ref{sec:Quantum-Deformation-Theory},
we identify the partition function as an element in the bosonic Fock space
uniquely specified by the Virasoro constraints in \S \ref{sec:W-Constraints}.

We summarize our results in the concluding \S \ref{sec:Conclusion}.

\section{Renormalization  of the Action Function}
\label{sec:Renormalization}

In this and the next Sections,
we will study the $I$-coordinates from various points of views.

\subsection{The effective action function of the topological 1D gravity}
It is the following formal power series in $x$ depending on infinitely
many parameters $t_0, \dots, t_n, \dots$:
\be \label{eqn:Action}
S = - \frac{1}{2}x^2 + \sum_{n \geq 1} t_{n-1} \frac{x^n}{n!},
\ee
The coefficients $t_n$'s will be called the {\em coupling constants}.
Since we do not concern
ourselves with the issue of the convergence of the above series,
we will treat the $t_n$'s either as formal variables,
or take a truncation
\be \label{eqn:Truncation}
t_{n+1} = t_{n+2} = \cdots = 0
\ee
for suitable $n$.

\subsection{The dialton shift}

One can rewrite $S$ more uniformly as follows:
\be
S 
= \sum_{n \geq 0} \tilde{t}_n \frac{x^{n+1}}{(n+1)!},
\ee
where  $\tilde{t}_n = t_n - \delta_{n,1}$.
This shift in coordinates is called the {\em dilaton shift}.

\subsection{Space of action functions}

When we regard $t_n$'s as formal variables,
we will consider the space $\cS$ consisting of formal power series of the form
\be
S = - \frac{1}{2}x^2 + \sum_{n \geq -1} T_n \frac{x^{n+1}}{(n+1)!},
\ee
where each $\tilde{T}_n$ is a formal power series in $t_0, t_1, \dots$;
we require furthermore that
\be
T_1|_{t_0=t_1=\cdots = 0} = 0,
\ee
i.e.,
the constant term of $\tilde{T}_1$ is $0$.

When we take the truncation \eqref{eqn:Truncation},
$\tilde{t}_0, \dots, \tilde{t}_n$ give us coordinates on the $n+1$-dimensional
Euclidean space  $\cS_n$
of degree $n+1$ polynomials without constant terms.

\subsection{Renormalization of the coupling constants by completing the squares}

Let us take
$$
S = - \frac{1}{2}(1-t_1) x^2 + \sum_{n \geq 0} t_{n-1} \frac{x^n}{n!}
$$
and apply the procedure of completing the square:
let $\tilde{x} = x-x_1$, where $x_1 = \frac{t_0}{1-t_1}$,
then
\ben
S & = & t_{-1} - \frac{1}{2}(1-t_1) x^2 + t_0 x + \sum_{n \geq 3} t_{n-1} \frac{x^n}{n!} \\
& = & t_{-1} +\half \frac{t_0^2}{1-t_1}
- \frac{1}{2}(1-t_1)\tilde{x} ^2
+ \sum_{n \geq 3} t_{n-1} \frac{(\tilde{x} + \frac{t_0}{1-t_1})^n}{n!} \\
& = & t_{-1} + \half \frac{t_0^2}{1-t_1}
- \frac{1}{2}(1-t_1)\tilde{x} ^2
+ \sum_{n \geq 3} t_{n-1} \sum_{m=0}^n  \frac{\tilde{x}^m}{m!}\frac{1}{(n-m)!}
\biggl( \frac{t_0}{1-t_1}\biggr)^{n-m} \\
& = & \biggl(t_{-1} +  \half \frac{t_0^2}{1-t_1}
+ \sum_{n \geq 3}  \frac{t_{n-1}}{n!}  \biggl( \frac{t_0}{1-t_1}\biggr)^{n} \biggr)
+ \tilde{x} \sum_{n \geq 3}  \frac{t_{n-1}}{(n-1)!}  \biggl( \frac{t_0}{1-t_1}\biggr)^{n-1} \\
&& - \frac{1}{2}\biggl(1-t_1
- \sum_{n \geq 3} t_{n-1}   \frac{1}{(n-2)!}
\biggl( \frac{t_0}{1-t_1}\biggr)^{n-2} \biggr)\tilde{x} ^2 \\
& + & \sum_{m=3}^\infty \frac{\tilde{x}^m}{m!}
\sum_{n \geq m} t_{n-1} \frac{1}{(n-m)!}  \biggl( \frac{t_0}{1-t_1}\biggr)^{n-m} \\
& = & \biggl(t_{-1} +  \half \frac{t_0^2}{1-t_1} + \sum_{n \geq 3}  \frac{t_{n-1}}{n!}
\biggl( \frac{t_0}{1-t_1}\biggr)^{n} \biggr)
+ \tilde{x} \sum_{n \geq 2}  \frac{t_{n}}{n!}  \biggl( \frac{t_0}{1-t_1}\biggr)^{n} \\
& - & \frac{1}{2}\biggl(1
- \sum_{n \geq 0} t_{n+1}   \frac{1}{n!}
\biggl( \frac{t_0}{1-t_1}\biggr)^{n} \biggr)\tilde{x} ^2
+ \sum_{m=3}^\infty \frac{\tilde{x}^m}{m!}
\sum_{n \geq 0} t_{n+m-1} \frac{1}{n!}  \biggl( \frac{t_0}{1-t_1}\biggr)^{n}.
\een
From this computation we define the renormalization transformation $R: \cS \to \cS$:
\be
\begin{split}
& (t_{-1}, t_0, t_1, \dots)
\mapsto (\hat{t}_{-1}, \hat{t}_0, \hat{t}_1, \dots),
\end{split}
\ee
where
\bea
&& \hat{t}_{-1} = t_{-1} +  \half \frac{t_0^2}{1-t_1}
+ \sum_{n \geq 3}  \frac{t_{n-1}}{n!}  \biggl( \frac{t_0}{1-t_1}\biggr)^n, \\
&& \hat{t}_0 =  \sum_{n \geq 2}  \frac{t_{n}}{n!}  \biggl( \frac{t_0}{1-t_1}\biggr)^n, \\
&& \hat{t}_1 = \sum_{n \geq 0} t_{n+1}   \frac{1}{n!}  \biggl( \frac{t_0}{1-t_1}\biggr)^{n}, \\
&& \hat{t}_m = \sum_{n \geq 0} t_{n+m} \frac{1}{n!}  \biggl( \frac{t_0}{1-t_1}\biggr)^n, \;\;\;
m \geq 2.
\eea

\subsection{Geometric interpretation of the renormalization transformation}

The renormalization transformation is related to Newton's algorithm as follows.
Formally,
consider the tangent line to
$$
y = \frac{\pd}{\pd x}S = - (1-t_1) x + t_0  + \sum_{n \geq 2} t_{n} \frac{x^n}{n!}
$$
at $x_0 = 0$.
Because
\ben
&& y|_{x=0} = t_0, \\
&& \frac{\pd}{\pd x} y |_{x=0} = -(1-t_1),
\een
so the tangent line is given by
$$
y = -(1-t_1) x + t_0.
$$
This line intersection the $x$-axis at $x_1 = \frac{t_0}{1-t_1}$.
Then $R(t_{-1}, t_0, t_1, \dots)$ are the Taylor coefficients of $S$ at $x=x_1$.

\subsection{Renormalization flow and the gradient flow}

Consider the gradient flow of $S$:
\be
\frac{\pd x(s)}{\pd s} =  \frac{\pd S}{\pd x}= - x(s) + \sum_{n \geq 0} t_{n} \frac{x(s)^n}{n!}
\ee
Let us first show that it can be formally solved by power series method.
Write
\ben
x(s) = \sum_{n \geq 1} a_n s^n,
\een
then the above equation gives:
\ben
&& a_1 + 2a_2s + 3 a_3 s^2 + \cdots \\
& = & t_0 - \sum_{m \geq 1} a_m s^m
+ \sum_{l \geq 1} \frac{t_l}{l!} (\sum_{n=1}^\infty a_n t^n)^l \\
& = &  t_0 - \sum_{m \geq 1} a_m s^m
+ \sum_{l \geq 1} t_l \sum_{\substack{k_1+ \cdots + k_r=l\\k_1, \dots, k_r \geq 0}}
\frac{a_1^{k_1}}{k_1!} \cdots \frac{a_r^{k_r}}{k_r!}  \cdot s^{\sum_{j=1}^n jk_j}.
\een
Therefore,
\ben
&& a_1 = t_0, \\
&& (m+1) a_{m+1}
= - a_m +  \sum_{\substack{\sum_{j=1}^r jk_j =m\\k_1, \dots, k_r \geq 0}}
\frac{a_1^{k_1}}{k_1!} \cdots \frac{a_r^{k_r}}{k_r!} \cdot t_{\sum_{j=1}^r k_j},
\een
for $m \geq 1$.
Therefore,
one can recursively find $a_m$.
For example,
\ben
&& a_2 = - \frac{1}{2} t_0(1-t_1), \\
&& a_3 = \frac{1}{3!}(t_0(1-t_1)^2+t_0^2t_2), \\
&& a_4 = - \frac{1}{4!}(t_0(1-t_1)^3+4t_0^2t_2(1-t_1) + t_0^3t_3).
\een
To make an analytic analysis,
one can fix some $N > 0$ and take a truncation $t_n = 0$ for $n \geq N$.
Then the gradient flow will take any initial value to one of the critical points of
$S$ or to infinity.
Under suitable conditions,
one can embed the renormalization transformations into the gradient flow and show that
the repeated renormalization transformations may take us to
the critical point of $S$ in the limit.
For example, take $t_n = 0$ for $n \geq 3$,
then the gradient flow equation becomes
\be
\frac{\pd x(s)}{\pd s} = t_0 - (1-t_1) x(s) + \half t_2 x(s)^2.
\ee
It can be solved by
\be
x(s) = -\frac{2t_0\tan (\half \alpha s)}{\alpha-(1-t_1) \tan (\half \alpha s)},
\ee
where
\be
\alpha = (-(1-t_1)^2+2t_0t_2)^{1/2}.
\ee
One can check that when
\be
s_0 = - \frac{2}{\alpha} \arctan\biggl(\frac{\alpha}{1-t_1} \biggr),
\ee
one has
\be
x(s_0) = \frac{t_0}{1-t_1}.
\ee

\subsection{Limit of the repeated renormalization transformation}

By repeating the Newton algorithm,
one gets a sequence $\{x_n = x_0 + \frac{t_0}{1-t_1} + \cdots
+ \frac{t_0^{(n-1)}}{1-t_1^{(n-1)}} \}$,
and $\{R^n(t_{-1},t_0, t_1, \dots) = (t_{-1}^{(n)}, t_0^{(n)}, t_1^{(n)}, \dots)\}$.
Then $\{x_n\}$ converges in the adic topology of formal power series to $x_\infty$
which is the zero of
$$\frac{\pd S}{\pd x}  = 0,$$
i.e.,
$x_{(\infty)}$ is a critical point of $S$,
and $\{\bt^{(n)} =  (t_{-1}^{(n)}, t_0^{(n)}, t_1^{(n)}, \dots)\}$ converges to
$\bt^{(\infty)}$ which is the Taylor coefficients of $S$ at $x=x_\infty$.

\begin{prop}
The limit point $x_\infty$ satisfies the following equation:
\be \label{eqn:Critical}
x_\infty = \sum_{n\geq 0} t_n \frac{x_\infty^n}{n!}.
\ee
\end{prop}

\begin{prop}
The following formula for $x_\infty$ holds:
\be \label{eqn:Xinfinity}
x_\infty = \sum_{k=1}^\infty \frac{1}{k}
\sum_{p_1 + \cdots + p_k = k-1} \frac{t_{p_1}}{p_1!} \cdots
\frac{t_{p_k}}{p_k!}.
\ee
\end{prop}

\begin{proof}
This can be proved by Lagrange inversion formula as follows.
Consider
\be
z = \frac{w}{ t_0 + \sum_{n \geq 1} t_n \frac{w^n}{n!}}.
\ee
This is a series in $w$ with leading term $\frac{w}{t_0}$.
Take the inverse series
\be
w = \sum_{k \geq 1} a_k z^k
\ee
by Lagrange inversion formula:
\ben
a_k & = & \res_{z=0} \frac{w}{z^{k+1}} dz
= - \frac{1}{k} \res_{z=0} w d \frac{1}{z^k} \\
& = &  \frac{1}{k} \res_{w=0}  \frac{1}{z^k} d w
= \frac{1}{k}  \res_{w=0}  \frac{(\sum_{n \geq 0} t_n \frac{w^n}{n!})^k}{w^k} d w \\
& = & \frac{1}{k} \sum_{p_1+ \cdots + p_k=k-1} \frac{t_{p_1}}{p_1!} \cdots \frac{t_{p_k}}{p_k!}.
\een
The proof is completed by setting $z=1$.
\end{proof}

\begin{thm}
The limit $\bt^{(\infty)}$ is given by:
\be \label{eqn:T-Infinity}
\bt^{(\infty)} = (\sum_{n=0}^\infty \frac{t_{n}-\delta_{n,1}}{(n+1)!} I_0^{n+1},
0, I_1-1, I_2, I_3, \dots).
\ee
where for $k \geq 0$,
\be \label{def:Ik}
I_k= \sum_{n \geq 0} t_{n+k} \frac{x_\infty^n}{n!}.
\ee
\end{thm}

\begin{proof}
Note $\bt^{(\infty)}$ are just the Taylor coefficients of $S$ at $x= x_\infty$
up to constant factors:
\be
S = S(x_\infty) + \sum_{n=1}^\infty \frac{1}{n!} \frac{\pd^nS}{\pd x^n}(x_\infty)
\cdot (x-x_\infty).
\ee
By \eqref{eqn:Action},
we have for $k \geq 1$,
\ben
\frac{\pd^kS}{\pd x^k}(x_\infty)
= \sum_{n \geq 0} t_{n+k-1} \frac{x_\infty^n}{n!} -\delta_{k,1} x_\infty -\delta_{k,2}
= I_{k-1}-\delta_{k,1}I_0-\delta_{k,2}.
\een
By \eqref{eqn:Critical},
\be
x_\infty = I_0,
\ee
and  so we have
\ben
S(x_\infty) = \sum_{n=0}^\infty \frac{t_{n}-\delta_{n,1}}{(n+1)!} x_\infty^{n+1}
= \sum_{n=0}^\infty \frac{t_{n}-\delta_{n,1}}{(n+1)!} I_0^{n+1}.
\een
\end{proof}

Recall that our convention is that $t_{-1} = 1$,
therefore,
\be
I_{-1} = \sum_{n =0}^\infty t_n \frac{I_0^{n+1}}{(n+1)!}.
\ee
Therefore,
\be
\sum_{n=0}^\infty \frac{t_{n}-\delta_{n,1}}{(n+1)!} I_0^{n+1} = I_{-1} - \frac{I_0^2}{2}.
\ee
From an algebraic point of view,
the formula for $\bt^{(\infty)}$ in the above Theorem is not perfect
since its first term contains explicitly the indeterminate $t_n$.
We will fix this in the next subsection.

\subsection{A change of coordinates on the space of coupling constants}

By \eqref{def:Ik} and \eqref{eqn:Xinfinity},
one can express $I_n$'s in terms of $t_m$'s
by a triangular relation,
therefore,
one can also  express $t_m$'s in terms of $I_n$'s.
The formula turns out to be very simple as in the next Proposition.

\begin{prop} \label{prop:T-in-I}
One can express $\{ t_k\}_{k=0}^\infty$ in terms of  $\{ I_k\}_{k=0}^\infty$
by the following formula:
\be \label{eqn:T-in-I}
t_k = \sum_{n=0}^\infty \frac{(-1)^n I_0^n}{n!}I_{n+k}.
\ee
\end{prop}

\begin{proof}
We first rewrite \eqref{def:Ik} as follows:
\be
t_k = I_k - \sum_{n \geq 1} t_{n+k} \frac{I_0^n}{n!}.
\ee
Then repeatedly apply this as follows:
\ben
t_0 & = & I_0 - \sum_{n \geq 1} t_n \frac{I_0^n}{n!}
=  I_0 -t_1 I_0  - t_2 \frac{I_0^2}{2!} - t_3 \frac{I_0^3}{3!} - \cdots \\
& = & I_0 - (I_1 - \sum_{n \geq 1} t_{n+1} \frac{I_0^n}{n!}) I_0
 - t_2 \frac{I_0^2}{2!} - t_3 \frac{I_0^3}{3!} - \cdots \\
& = & I_0 - I_1I_0 + \sum_{n \geq 2} t_n \biggl(\frac{1}{(n-1)!} - \frac{1}{n!}\biggr) I_0^n \\
& = & I_0 - I_1I_0 + \biggl(\frac{1}{(2-1)!} - \frac{1}{2!} \biggr) I_2I_0^2 \\
& + & \sum_{n \geq 3} t_n \biggl(-\frac{1}{(n-2)!} \biggl(\frac{1}{(2-1)!}
- \frac{1}{2!} \biggr) +
\biggl(\frac{1}{(n-1)!} - \frac{1}{n!} \biggr) \biggr) I_0^n \\
& = & I_0 - I_1I_0 + \frac{1}{2!} I_2I_0^2 \\
& + & \sum_{n \geq 3} t_n \biggl(-\frac{1}{(n-2)!} \cdot \frac{1}{2!} + \frac{1}{(n-1)!}
- \frac{1}{n!} \biggr) I_0^n.
\een
Repeating this once more,
\ben
t_0 & = & I_0 - I_1I_0 + \frac{1}{2!} I_2I_0^2  - \frac{1}{3!} I_3 I_0^3 \\
& + & \sum_{n \geq 4} t_n \biggl(\frac{1}{(n-3)!} \cdot \frac{1}{3!}
-\frac{1}{(n-2)!} \cdot \frac{1}{2!} + \frac{1}{(n-1)!} - \frac{1}{n!} \biggr) I_0^n.
\een
Now it is clear how the proof for the case of $t_0$ can be completed by mathematical induction.
The proof for the general case of $t_k$ is exactly the same.
\end{proof}

As a corollary to Proposition \ref{prop:T-in-I},
we have:

\begin{thm} \label{thm:T-infinity}
The limit $\bt^{(\infty)}$ is given by:
\be
\bt^{(\infty)} = ( \sum_{k=0}^\infty  \frac{(-1)^k}{(k+1)!} (I_k+\delta_{k,1}) I_0^{k+1},
0, I_1-1, I_2, I_3, \dots).
\ee
\end{thm}

\begin{proof}
Just combine \eqref{eqn:T-Infinity} with \eqref{eqn:T-in-I} as follows:
\ben
\sum_{n=0}^\infty \frac{t_{n}-\delta_{n,1}}{(n+1)!} I_0^{n+1}
& = & \sum_{n=0}^\infty \frac{I_0^{n+1}}{(n+1)!}
\sum_{m=0}^\infty \frac{(-1)^mI_0^m}{m!} I_{m+n}
- \frac{I_0^2}{2} \\
& = & \sum_{k=0}^\infty \sum_{m=0}^k \frac{(-1)^m}{m!(k-m+1)!} I_k I_0^{k+1} - \frac{I_0^2}{2} \\
& = & \sum_{k=0}^\infty  \frac{(-1)^k}{(k+1)!} I_k I_0^{k+1} - \frac{I_0^2}{2} \\
& = & \sum_{k=0}^\infty  \frac{(-1)^k}{(k+1)!} (I_k+\delta_{k,1}) I_0^{k+1}.
\een
\end{proof}

\subsection{Jacobian matrices}

As another straightforward
corollary to Proposition \ref{prop:T-in-I}, we have:

\begin{cor}
The Jocabian matrix of the coordinate change from $\{T_k\}$ to $\{t_k\}$ is given by:
\bea
&& \frac{\pd t_k}{\pd I_0} = \delta_{k,0} - t_{k+1},  \\
&& \frac{\pd t_k}{\pd I_l} =  \frac{(-1)^{l-k}I_0^{l-k}}{(l-k)!} H(l-k), \;\;\; l \geq 1,
\eea
where $H(x)$ is the Heaviside function:
\be
H(x) = \begin{cases}
1, & x \geq 0, \\
0, & x < 0.
\end{cases}
\ee
\end{cor}

\begin{cor}
The vector fields $\{\frac{\pd}{\pd I_l}\}$ can be expressed in terms of
the vector fields $\{\frac{\pd}{\pd t_k}\}$ as follows:
\bea
&& \frac{\pd}{\pd I_0} = \frac{\pd}{\pd t_0} - \sum_{k \geq 0} t_{k+1} \frac{\pd}{\pd t_k},
\label{eqn: diff I0} \\
&& \frac{\pd}{\pd I_l} = \sum_{k=0}^l \frac{(-1)^{l-k}I_0^{l-k}}{(l-k)!} \frac{\pd}{\pd t_k}.
\label{eqn: Diff I1}
\eea
\end{cor}

Now for $k \geq 0$,
\ben
\frac{\pd I_0}{\pd t_k}  = \frac{I_0^k}{k!}
+ \sum_{n \geq 1} \frac{t_n}{(n-1)!} I_0^{n-1} \cdot
\pd_{t_k} I_0= \frac{I_0^k}{k!}+  I_1 \cdot \pd_{t_k} I_0,
\een
and so
\be \label{eqn:Pd-I0-tk}
\frac{\pd I_0}{\pd t_k}  = \frac{1}{1-I_1}\frac{I_0^k}{k!}.
\ee
And recall for $l \geq 1$,
\be
I_l = \sum_{n\geq 0} \frac{t_{n+l}}{n!} I_0^n.
\ee
Therefore,
\be \label{eqn:Pd-Il-Pd-tk}
\frac{\pd I_l}{\pd t_k}
= \frac{I_{l+1}}{1-I_1}\frac{I_0^k}{k!} + \frac{I_0^{k-l}}{(k-l)!} H(k-l),
\ee
where $H(x)$ is the Heaviside step function.
In fact,
\ben
\frac{\pd I_l}{\pd t_k}  & = &
\sum_{n\geq 1} \frac{t_{n+l}}{(n-1)!} I_0^{n-1} \cdot \pd_{t_k} I_0
+ \sum_{n\geq 0} \frac{\delta_{n, k-l}}{n!} I_0^n \\
& = & \frac{I_{l+1}}{1-I_1}\frac{I_0^k}{k!} + \frac{I_0^{k-l}}{(k-l)!} H(k-l),
\een

\begin{prop}
The vector fields $\{\frac{\pd}{\pd t_k}\}$  can be expressed in terms of
the vector fields  $\{\frac{\pd}{\pd I_l}\}$ as follows:
\be \label{eqn:diff-tk}
\frac{\pd}{\pd t_k}
= \frac{1}{1-I_1}\frac{I_0^k}{k!}  \frac{\pd}{\pd I_0} +
\frac{I_0^k}{k!} \sum_{l \geq 1} \frac{I_{l+1}}{1-I_1} \frac{\pd}{\pd I_l}
+ \sum_{1 \leq l \leq k} \frac{I_0^{k-l}}{(k-l)!} \frac{\pd}{\pd I_l}.
\ee
\end{prop}

\begin{proof}
\ben
\frac{\pd}{\pd t_k}
& = & \frac{\pd I_0}{\pd t_k} \frac{\pd}{\pd I_0}
+ \sum_{l \geq 1} \frac{\pd I_l}{\pd t_k} \frac{\pd}{\pd I_l} \\
& = & \frac{1}{1-I_1}\frac{I_0^k}{k!}  \frac{\pd}{\pd I_0}
+ \sum_{l\geq 1} \biggl( \frac{I_{l+1}}{1-I_1}\frac{I_0^k}{k!}
+ \frac{I_0^{k-l}}{(k-l)!} H(k-l) \biggr) \frac{\pd}{\pd I_l} \\
& = & \frac{1}{1-I_1}\frac{I_0^k}{k!}  \frac{\pd}{\pd I_0} +
\frac{I_0^k}{k!} \sum_{l \geq 1} \frac{I_{l+1}}{1-I_1} \frac{\pd}{\pd I_l}
+ \sum_{1 \leq l \leq k} \frac{I_0^{k-l}}{(k-l)!} \frac{\pd}{\pd I_l}.
\een
\end{proof}

For example,
\bea
&& \frac{\pd}{\pd t_0} 
= \frac{1}{1-I_1} \frac{\pd}{\pd I_0}
+ \sum_{l \geq 1} \frac{I_{l+1}}{1-I_1} \frac{\pd}{\pd I_l}, \label{eqn:Pd-t0-in-I} \\
&& \frac{\pd}{\pd t_1} 
=  \frac{I_0}{1-I_1} \frac{\pd}{\pd I_0} + \biggl(\frac{I_2I_0}{1-I_1} + 1\biggr) \frac{\pd}{\pd I_1}
+ \sum_{l \geq 2} \frac{I_{l+1}I_0}{1-I_1} \frac{\pd}{\pd I_l}.
\eea

\subsection{Determinants and  Cramer's rule for the Jacobian matrices}
It is very interesting to relate the two calculations in the preceding subsection.
The Jacabian matrices there turn out provide nice examples of matrices of infinite sizes
for which one can formally define their determinants and compute their inverse matrices by
Cramer's rule.

First of all, the matrix $( \frac{\pd t_k}{\pd I_l})_{k,l \geq 0}$ has the form:
\be
\begin{pmatrix}
1 -t_1 & - t_2 & - t_3 & -t_4 & -t_5 & -t_6 & \cdots \\
-I_0 & 1 & 0 & 0 & 0 & 0 & \cdots \\
\frac{I_0^2}{2!} & - I_0 & 1 & 0 & 0 & 0 & \cdots \\
- \frac{I_0^3}{3!} & \frac{I_0^2}{2!} & - I_0 & 1 & 0 & 0 & \cdots \\
\vdots & \vdots & \vdots & \vdots & \vdots & \vdots &
\end{pmatrix}
\ee
To define/calculate its determinant,
we expand along the first row to get:
\be
\det \big( \frac{\pd t_k}{\pd I_l}\big)_{k,l \geq 0}
= (1-t_1) + \sum_{j=2}^\infty (-1)^j t_j A_j,
\ee
where $A_j$ is the following determinant:
\ben
A_j = \begin{vmatrix}
-I_0 & 1 & 0 & 0 &  \cdots & 0 & 0 & 0 & \cdots \\
\frac{I_0^2}{2!} & - I_0 & 1 & 0  & \cdots & 0 & 0 & 0 & \cdots \\
- \frac{I_0^3}{3!} & \frac{I_0^2}{2!} & - I_0 & 1  & \cdots & 0 & 0 & 0 & \cdots \\
\vdots & \vdots & \vdots & \vdots  &   & \vdots & \vdots & \vdots & \vdots \\
\frac{(-1)^{j-1}I_0^{j-1}}{(j-1)!} & \frac{(-1)^{j-2}I_0^{j-2}}{(j-2)!}
& \cdots & \cdots & \cdots & -I_0 & 0 & 0 & \cdots \\
\frac{(-1)^{j}I_0^{j}}{j!} & \frac{(-1)^{j-1}I_0^{j-1}}{(j-1)!}
& \cdots & \cdots & \cdots  & \frac{I_0^2}{2} & 1 & 0 & \cdots \\
\frac{(-1)^{j+1}I_0^{j+1}}{j!} & \frac{(-1)^{j}I_0^{j}}{j!}
& \cdots & \cdots & \cdots   & -\frac{I_0^3}{3!} & -I_0 & 1 & \cdots \\
\vdots & \vdots & \vdots   & \vdots &   & \vdots & \vdots & & \vdots \\
\end{vmatrix}
\een
Naturally we define:
\be
A_j = \begin{vmatrix}
-I_0 & 1 & 0 & 0 & 0 &  \cdots & 0\\
\frac{I_0^2}{2!} & - I_0 & 1 & 0 & 0 & \cdots & 0 \\
- \frac{I_0^3}{3!} & \frac{I_0^2}{2!} & - I_0 & 1 & 0  & \cdots & 0 \\
\vdots & \vdots & \vdots & \vdots & \vdots &   & \vdots \\
\frac{(-1)^{j-1}I_0^{j-1}}{(j-1)!} & \frac{(-1)^{j-2}I_0^{j-2}}{(j-2)!}
& \cdots & \cdots & \cdots & \cdots & -I_0
\end{vmatrix}
\ee
They can be computed as follows.
By expansion along the first column one gets a recursion relation for $A_j$:
\be
A_j = - (I_0A_{j-1} + \frac{I_0^2}{2!} A_{j-2} + \cdots + \frac{I_0^{j-1}}{(j-1)!} A_1),
\ee
where $A_1 = 1$.
It is clear that
\be
A_j = (-1)^{j-1} \frac{I_0^{j-1}}{(j-1)!}.
\ee
Therefore,
\be
\det \big( \frac{\pd t_k}{\pd I_l}\big)_{k,l \geq 0}
= (1-t_1) - \sum_{j=2}^\infty t_j \frac{I_0^{j-1}}{(j-1)!} = 1 - I_1.
\ee
Similarly,
one can define and compute the $(i,j)$-minors of the Jacobian matrix
$ \big( \frac{\pd t_k}{\pd I_l}\big)_{k,l \geq 0}$.

More interestingly,
even though the matrix $\big( \frac{\pd I_l}{\pd t_k}\big)_{k,l \geq 0}$
does not have a nice simple shape,
nevertheless,
after some simple row operations (multiplying the first row by $\frac{I_0^k}{k!}$ and
subtracting it from the row indexed by $k$),
it takes the following form:
\be
\begin{pmatrix}
\frac{1}{1-I_1} & \frac{I_2}{1-I_1} & \frac{I_3}{1-I_1} & \frac{I_4}{1-I_1}
& \frac{I_5}{1-I_1} & \frac{I_6}{1-I_1} & \cdots \\
0 & 1 & 0 & 0 & 0 & 0 & \cdots \\
0 & I_0 & 1 & 0 & 0 & 0 & \cdots \\
0& \frac{I_0^2}{2!} & I_0 & 1 & 0 & 0 & \cdots \\
\vdots & \vdots & \vdots & \vdots & \vdots & \vdots &
\end{pmatrix}
\ee
Therefore,
one can still formally define and compute the determinant and the inverse matrix
of $\big( \frac{\pd I_l}{\pd t_k}\big)_{k,l \geq 0}$ as above.

\subsection{The manifold of coupling constants and singular behavior of
coordinate transformation}

We introduce a manifold of coupling constants,
with two coordinate patches,
one with $\{t_k\}_{k \geq 0}$ as local coordinates,
the other with $\{I_l\}_{l \geq 0}$ as local coordinates.
The calculation of the determinant of the Jacobian matrix indicates
singular behavior of $I_l$'s as functions of $t_k$'s along the hypersurface
defined by $I_1 =1$.
We have already shown that
\ben
S = \sum_{k=0}^\infty  \frac{(-1)^k}{(k+1)!} (I_k+\delta_{k,1}) I_0^{k+1} ,
\een
and $x = I_0 = x_\infty$ is a critical point of $S$:
$$\frac{\pd S}{\pd x}\biggr|_{x=I_0} = 0.$$
When $I_1 = 1$,
one also has
\be
\frac{\pd^2 S}{\pd x^2}\biggr|_{x=I_0} = 0.
\ee
For each $k > 1$,
a multicritical point of $S$ order $k$ is a point $x = x_c$ where one has
\bea
&& \frac{\pd^j S}{\pd x^j}\biggr|_{x=x_c} = 0, \;\; j = 1, \dots, k-1, \\
&& \frac{\pd^k S}{\pd x^k}\biggr|_{x=x_c} \neq 0.
\eea
Clearly,
$S$ has a multicritical point of order $k>1$ when
\be
t_n = \delta_{n,1} + \delta_{n,k}.
\ee
Write
\be
s_n^{(k)} = t_n - \delta_{n,1} - \delta_{n,k}.
\ee
We rewrite $S$ as follows:
\be
S = - \frac{1}{k!} x^k + \sum_{n \geq 1} s_{n-1}^{(k)} \frac{x^n}{n!}.
\ee
Then the equation for the critical point \eqref{eqn:Critical} can be rewritten as:
\be \label{eqn:Critical-in-S}
\frac{x_\infty^k}{k!} = \sum_{n\geq 0} s_n^{(k)} \frac{x_\infty^n}{n!},
\ee
and the Taylor expansion at $x=x_\infty$ still has the form:
\be
S =\sum_{k=0}^\infty  \frac{(-1)^k}{(k+1)!} (I_k+\delta_{k,1}) I_0^{k+1}
+ \sum_{n=2}^\infty \frac{I_{n-1}-\delta_{n,2} }{n!} (x-I_0)^n .
\ee

\subsection{Another coordinate system on the space of couple constants}
Let us write down  the first few terms of \eqref{eqn:Xinfinity}:
\ben
I_0 = t_0 + t_0t_1 + (\frac{t_0^2}{2}t_2 + t_0t_1^2)
+ (\frac{1}{6}t_0^3t_3 + \frac{3}{2}t_0^2t_1t_2 + t_0t_1^3) + \cdots .
\een
One can see that
\be
\frac{\pd^k I_0}{\pd t_0^k} - \delta_{k,1} = t_k + \cdots,
\ee
and so $\{\frac{\pd^k I_0}{\pd t_0^k} - \delta_{k,1}\}_{k \geq 0}$ can be also used
as a coordinate system on the space of coupling constants.
We now find explicit formula for them as formal power series in $t$-coordinates.
Note
 \be
I_0  = t_0 + \sum_{k=2}^\infty \frac{1}{k} \sum_{m=1}^{k-1} \binom{k}{m} t_0^m
\sum_{\substack{p_1 + \cdots + p_{k-m} = k-1 \\p_1, \dots, p_{k-m} > 0}} \frac{t_{p_1}}{p_1!} \cdots
\frac{t_{p_{k-m}}}{p_{k-m}!}.
\ee
Take $\frac{\pd}{\pd t_0}$ on both sides:
\ben
\frac{\pd I_0}{\pd t_0}
& = & 1 + \sum_{k=2}^\infty \frac{1}{k} \sum_{m=1}^{k-1} \binom{k}{m} m t_0^{m-1}
\sum_{\substack{p_1 + \cdots + p_{k-m} = k-1 \\p_1, \dots, p_{k-m} > 0}} \frac{t_{p_1}}{p_1!} \cdots
\frac{t_{p_{k-m}}}{p_{k-m}!} \\
& = & 1 + \sum_{k=2}^\infty  \sum_{m=1}^{k-1} \binom{k-1}{m-1} t_0^{m-1}
\sum_{\substack{p_1 + \cdots + p_{k-m} = k-1 \\p_1, \dots, p_{k-m} > 0}} \frac{t_{p_1}}{p_1!} \cdots
\frac{t_{p_{k-m}}}{p_{k-m}!} \\
& = & 1 + t_1 +  \sum_{k=1}^\infty  \sum_{m=0}^{k-1} \binom{k}{m} t_0^{m}
\sum_{\substack{p_1 + \cdots + p_{k-m} = k \\p_1, \dots, p_{k-m} > 0}} \frac{t_{p_1}}{p_1!} \cdots
\frac{t_{p_{k-m}}}{p_{k-m}!} \\
& = & 1 + \sum_{k=1}^\infty \sum_{p_1 + \cdots + p_{k} = k } \frac{t_{p_1}}{p_1!} \cdots
\frac{t_{p_{k}}}{p_{k}!}.
\een
Once more,
\ben
\frac{\pd^2 I_0}{\pd t_0^2}
& = & \sum_{k=2}^\infty  \sum_{m=0}^{k-1} \binom{k}{m} m t_0^{m-1}
\sum_{\substack{p_1 + \cdots + p_{k-m} = k \\p_1, \dots, p_{k-m} > 0}} \frac{t_{p_1}}{p_1!} \cdots
\frac{t_{p_{k-m}}}{p_{k-m}!} \\
& = & t_2 + \sum_{k=3}^\infty  k \sum_{m=1}^{k-1} \binom{k-1}{m-1}  t_0^{m-1}
\sum_{\substack{p_1 + \cdots + p_{k-m} = k \\p_1, \dots, p_{k-m} > 0}} \frac{t_{p_1}}{p_1!} \cdots
\frac{t_{p_{k-m}}}{p_{k-m}!} \\
& = & t_2 + \sum_{k=2}^\infty  (k+1) \sum_{m=0}^{k-1} \binom{k}{m}  t_0^{m}
\sum_{\substack{p_1 + \cdots + p_{k-m} = k+1 \\p_1, \dots, p_{k-m} > 0}} \frac{t_{p_1}}{p_1!} \cdots
\frac{t_{p_{k-m}}}{p_{k-m}!} \\
& = & t_2 + \sum_{k=2}^\infty (k+1) \sum_{p_1 + \cdots + p_{k} = k+1 } \frac{t_{p_1}}{p_1!} \cdots
\frac{t_{p_{k}}}{p_{k}!}.
\een
and inductively,
one finds
\be \label{eqn:Pd-l-I0-Pd-t0}
\frac{\pd^{l} I_0}{\pd t_0^{l}}
= \sum_{k=1}^\infty (k+1)\cdots (k+l-1) \sum_{p_1 + \cdots + p_{k} = k+l-1 } \frac{t_{p_1}}{p_1!} \cdots
\frac{t_{p_{k}}}{p_{k}!}.
\ee

\subsection{Derivatives of $I_0$ and the $I$-coordinates}

In this subsection we present explicit formulas for the coordinate changes between
$\{\frac{\pd^k I_0}{\pd t_0^k}-\delta_{k,1}\}_{k\geq 0}$ and $\{I_n\}_{n \geq 0}$.
Recall by \eqref{eqn:Pd-I0-tk} and  \eqref{eqn:Pd-Il-Pd-tk},
we have
\be \label{eqn:Pd-I0-Pd-t0}
\frac{\pd I_0}{\pd t_0}  = \frac{1}{1-I_1},
\ee
and for $l \geq 1$,
\be \label{eqn:Pd-Il-Pd-t0}
\frac{\pd I_l}{\pd t_0} = \frac{I_{l+1}}{1-I_1}.
\ee
Taking derivatives on both sides of \eqref{eqn:Pd-I0-Pd-t0} and using \eqref{eqn:Pd-Il-Pd-t0} repeatedly,
one gets:
\ben
&& \frac{\pd^2I_0}{\pd t_0^2} = \frac{I_2}{(1-I_1)^3}, \\
&& \frac{\pd^3I_0}{\pd t_0^3} = \frac{I_3}{(1-I_1)^4} + \frac{3I_2^2}{(1-I_1)^5}, \\
&& \frac{\pd^4I_0}{\pd t_0^4} = \frac{I_4}{(1-I_1)^5} + \frac{10I_2I_3}{(1-I_1)^6} + \frac{15I_2^3}{(1-I_1)^7}.
\een
In general we have:

\begin{prop} \label{prop:I0-Der-I}
The higher derivatives in $t_0$ can be written in $I$-coordinates as follows:
\be
\frac{\pd^nI_0}{\pd t_0^n} =\sum_{\sum_{j \geq 1} j m_j = n-1} 
\frac{ (\sum_j (j+1)m_j)!}{\prod_j ((j+1)!)^{m_j}m_j!} 
\cdot  \frac{\prod_jI_{j+1}^{m_j}}{(1-I_1)^{\sum_j (j+1)m_j+1}}. 
\ee
\end{prop}

\begin{proof}
Proof by induction using \eqref{eqn:Pd-t0-in-I}.
\end{proof}

Conversely,
one can also express the $I$-coordinates  in terms $\frac{\pd^kF_0}{\pd t_0^k}$.
By \eqref{eqn:Pd-I0-Pd-t0}, we have
\ben
I_1 = 1- \frac{1}{\frac{\pd I_0}{\pd t_0}}.
\een
Take $\frac{\pd}{\pd t_0}$ on both sides and use \eqref{eqn:Pd-Il-Pd-t0} for $l=1$:
\ben
I_2 = (1-I_1) \cdot \frac{\frac{\pd^2I_0}{\pd t_0^2}}{\biggl(\frac{\pd I_0}{\pd t_0}\biggr)^2} 
= \frac{\frac{\pd^2I_0}{\pd t_0^2}}{\biggl(\frac{\pd I_0}{\pd t_0}\biggr)^3}. 
\een
Repeating this process for several times, 
one gets
\ben
&& I_3 = \frac{\frac{\pd^3I_0}{\pd t_0^3}}{\biggl(\frac{\pd I_0}{\pd t_0}\biggr)^4}
- 3 \frac{\biggl(\frac{\pd^2I_0}{\pd t_0^2}\biggr)^2}{\biggl(\frac{\pd I_0}{\pd t_0}\biggr)^5}, \\
&& I_4 = \frac{\frac{\pd^4I_0}{\pd t_0^4}}{\biggl(\frac{\pd I_0}{\pd t_0}\biggr)^5}
- 10 \frac{\frac{\pd^2I_0}{\pd t_0^2}\frac{\pd^3I_0}{\pd t_0^3} }{\biggl(\frac{\pd I_0}{\pd t_0}\biggr)^6}
+ 15 \frac{\biggl(\frac{\pd^2I_0}{\pd t_0^2}\biggr)^3}{\biggl(\frac{\pd I_0}{\pd t_0}\biggr)^7}.
\een
In general we have:

\begin{prop}
The  $I$-coordinates  can be written in the higher derivatives in $t_0$  as follows:
\be \label{eqn:In-Pd-I0}
I_n =-\sum_{\sum_{j \geq 1} j m_j = n-1}
\frac{ (\sum_j (j+1)m_j)!}{\prod_j ((j+1)!)^{m_j}m_j!}
\cdot  \frac{\prod_j \biggl(-\frac{\pd^{j+1} I_0}{\pd t_0^{j+1}}\biggr)^{m_j}}
{\biggl(\frac{\pd I_0}{\pd t_0}\biggr)^{\sum_j (j+1)m_j+1}}.
\ee
\end{prop}

\begin{proof}
Essentially the same as the proof of Proposition \ref{prop:I0-Der-I}.
\end{proof}

Now we combine \eqref{eqn:T-in-Pd-I} with \eqref{eqn:In-Pd-I0} to get:
\be \label{eqn:T-in-Pd-I}
\begin{split}
t_k = & -\sum_{n=0}^\infty \frac{(-1)^n I_0^n}{n!}\sum_{\sum_{j \geq 1} j m_j = n+k-1}
\frac{ (\sum_j (j+1)m_j)!}{\prod_j ((j+1)!)^{m_j}m_j!} \\
\cdot & \frac{\prod_j \biggl(\frac{\pd^{j+1} I_0}{\pd t_0^{j+1}}\biggr)^{m_j}}
{\biggl(-\frac{\pd I_0}{\pd t_0}\biggr)^{\sum_j (j+1)m_j+1}}.
\end{split}
\ee
This formula expresses the $t$-coordinates in terms of derivatives of $I_0$ in $t_0$.
 
\section{Feynman Rules for $I_k$}
\label{sec:Feynman rules for I}

In this Section,
we will first propose Feynman rules for $I_0$,
and based on them, propose Feynman rules for $I_k$ ($k>-1$) and
$I_{-1} - \half I_0^2$.
We will prove these rules in a later Section.

\subsection{Feynman rules for $I_0$}

The following are the first few terms of $I_0=x_\infty$:
\ben
x_\infty & = &  t_0 + t_1t_0 +
\biggl( t_2\frac{t^2_0}{2!} + 2\frac{t^2_1}{2!} t_0 \biggr)
+ \biggl(t_3\frac{t^3_0}{3!} + 3t_1t_2\frac{t^2_0}{2!}
+ 6 \frac{t^3_1}{3!}  t_0 \biggr) \\
& + & \biggl[t_4\frac{t^4_0}{4!}
+ \biggl( 6\frac{t_2^2}{2!} + 4t_1t_3\biggr) \frac{t^3_0}{3!}
+ 12t_2 \frac{t^2_1}{2!} \frac{t^2_0}{2!} + 24 \frac{t^4_1}{4!}  t_0 \biggl] \\
& + &  \biggl[ t_5\frac{t^5_0}{5!} + (5t_1t_4 + 10t_2t_3) \frac{t^4_0}{4!}
+ \biggl( 30t_1 \frac{t_2^2}{2!} + 20t_3\frac{t^2_1}{2!} \biggl) \frac{t^3_0}{3!} \\
&& + 60t_2\frac{t^3_1}{3!} \frac{t^2_0}{2!} + 120 \frac{t^5_1}{5!} t_0 \biggr] \\
& + & \biggl[ t_6\frac{t^6_0}{6!}  +
\biggl( 20 \frac{t^2_3}{2!} + 6t_1t_5 + 15t_2t_4 \biggr) \frac{t^5_0}{5!} \\
&& + \biggl( 90 \frac{t^3_2}{3!} + 30t_4 \frac{t^2_1}{2!} + 60t_1t_2t_3
\biggr)\frac{t^4_0}{4!}  \\
&& +  \biggl( 120t_3 \frac{t^3_1}{3!} + 180 \frac{t^2_1}{2!} \frac{t^2_2}{2!}  \biggr)
\frac{t^3_0}{3!}
+ 360t_2 \frac{t^4_1}{4!} \frac{t^2_0}{2!} + 720  \frac{t^6_1}{6!} t_0
\biggr) + \cdots
\een

By looking at these explicit expressions,
one can formulate the following

\begin{thm} \label{thm:Rooted trees}
The formal power series $I_0$ is given by a sum over rooted trees
\be
x_\infty = \sum_{\text{$\Gamma$ is a rooted tree}} \frac{1}{|\Aut \Gamma|} w_\Gamma,
\ee
where the weight of $\Gamma$ is given by
\be
w_\Gamma = \prod_{v\in V(\Gamma)} w_v \cdot \prod_{e\in E(\Gamma)} w_e,
\ee
with $w_e$ and $w_v$ given by the following Feynman rule:
\bea
&& w(e) = 1, \\
&& w(v) = \begin{cases}
t_{\val(v)-1}, & \text{if $v$ is not the root vertex $\circ$}, \\
1, & \text{if $v$ is the root vertex $\circ$}.
\end{cases}
\eea
\end{thm}
For example,
$$
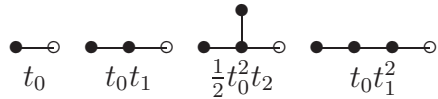
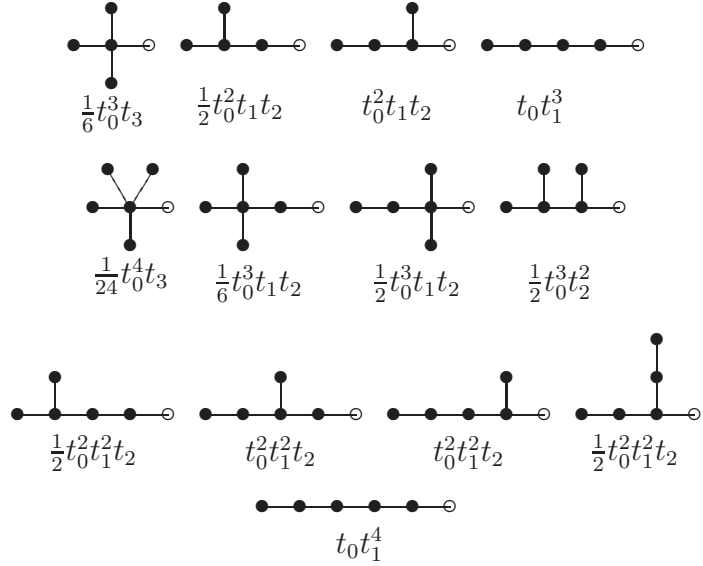
\xy
(0,0); (5,0), **@{-}; (0, 0)*+{\bullet}; (5,0)*+{\circ}; (2.5,-4)*+{t_0};
(10,0); (20,0), **@{-};  (10, 0)*+{\bullet}; (15, 0)*+{\bullet}; (20,0)*+{\circ}; (15,-4)*+{t_0t_1};
(25,0); (35,0), **@{-};  (30,0); (30,5), **@{-};
(25, 0)*+{\bullet}; (30, 0)*+{\bullet}; (30, 5)*+{\bullet}; (35,0)*+{\circ}; (30,-4)*+{\half t_0^2t_2};
(40,0); (55,0), **@{-};  (40, 0)*+{\bullet}; (45, 0)*+{\bullet}; (50, 0)*+{\bullet}; (55,0)*+{\circ};
(47.5,-4)*+{t_0t_1^2};
\endxy
$$

$$
\xy
(-5,0); (5,0), **@{-}; (0,5); (0,-5), **@{-};
(0, 0)*+{\bullet}; (-5, 0)*+{\bullet};(0, 5)*+{\bullet};(0, -5)*+{\bullet}; (5,0)*+{\circ};
(0,-9)*+{\frac{1}{6}t_0^3t_3};
(10,0); (25,0), **@{-};  (10, 0)*+{\bullet}; (15, 0)*+{\bullet}; (20, 0)*+{\bullet}; (25,0)*+{\circ};
(15,0); (15,5), **@{-}; (15,5)*+{\bullet};
(17,-8)*+{\half t_0^2t_1t_2};
(30,0); (45,0), **@{-};  (40,0); (40,5), **@{-};
(30, 0)*+{\bullet}; (35, 0)*+{\bullet}; (40, 0)*+{\bullet};(40, 5)*+{\bullet}; (45,0)*+{\circ};
(38,-8)*+{t_0^2t_1t_2};
(50,0); (70,0), **@{-};  (50, 0)*+{\bullet}; (55, 0)*+{\bullet}; (60, 0)*+{\bullet}; (65, 0)*+{\bullet}; (70,0)*+{\circ};
(57,-8)*+{t_0t_1^3};
\endxy
$$

$$
\xy
(-5,0); (5,0), **@{-}; (-3,5); (0,0), **@{-}; (3,5); (0,0), **@{-}; (0,0); (0,-5), **@{-};
(0, 0)*+{\bullet}; (-5, 0)*+{\bullet};(-3, 5)*+{\bullet}; (3, 5)*+{\bullet}; (0, -5)*+{\bullet}; (5,0)*+{\circ};
(0,-9)*+{\frac{1}{24}t_0^4t_3};
(10,0); (25,0), **@{-};  (10, 0)*+{\bullet}; (15, 0)*+{\bullet}; (20, 0)*+{\bullet}; (25,0)*+{\circ};
(15,-5); (15,5), **@{-}; (15,5)*+{\bullet}; (15,-5)*+{\bullet};
(17,-10)*+{\frac{1}{6} t_0^3t_1t_2};
(30,0); (45,0), **@{-};  (40,-5); (40,5), **@{-};
(30, 0)*+{\bullet}; (35, 0)*+{\bullet}; (40, 0)*+{\bullet};(40, 5)*+{\bullet}; (40, -5)*+{\bullet};(45,0)*+{\circ};
(38,-10)*+{\half t_0^3t_1t_2};
(50,0); (65,0), **@{-}; (55,0); (55,5), **@{-}; (60, 0); (60,5), **@{-};
(50, 0)*+{\bullet}; (55, 0)*+{\bullet}; (60, 0)*+{\bullet};
(55, 5)*+{\bullet}; (60, 5)*+{\bullet}; (65,0)*+{\circ};
(57,-10)*+{\half t_0^3t_2^2};
\endxy
$$

$$
\xy
(0,0); (20,0), **@{-}; (5,0); (5,5), **@{-};
(0, 0)*+{\bullet}; (5, 0)*+{\bullet};(5, 5)*+{\bullet}; (10, 0)*+{\bullet}; (15, 0)*+{\bullet}; (20,0)*+{\circ};
(10,-5)*+{\half t_0^2t_1^2t_2};
(25,0); (45,0), **@{-}; (35,0); (35,5), **@{-};
(25, 0)*+{\bullet}; (30, 0)*+{\bullet};(35, 5)*+{\bullet}; (35, 0)*+{\bullet}; (40, 0)*+{\bullet}; (45,0)*+{\circ};
(35,-5)*+{t_0^2t_1^2t_2};
(50,0); (70,0), **@{-}; (65,0); (65,5), **@{-};
(50, 0)*+{\bullet}; (55, 0)*+{\bullet}; (65, 5)*+{\bullet}; (60, 0)*+{\bullet};
(65, 0)*+{\bullet}; (70,0)*+{\circ};
(60,-5)*+{t_0^2t_1^2t_2};
(75,0); (90,0), **@{-}; (85,0); (85,10), **@{-};
(75, 0)*+{\bullet}; (80, 0)*+{\bullet};(85, 5)*+{\bullet}; (85, 10)*+{\bullet}; (85, 0)*+{\bullet}; (90,0)*+{\circ};
(82,-5)*+{\half t_0^2t_1^2t_2};
\endxy
$$
$$
\xy
(-5,0); (20,0), **@{-};
(0, 0)*+{\bullet}; (5, 0)*+{\bullet};(-5, 0)*+{\bullet}; (10, 0)*+{\bullet}; (15, 0)*+{\bullet}; (20,0)*+{\circ};
(8,-5)*+{t_0t_1^4};
\endxy
$$
The sum of contributions of all such diagrams can be
symbolically denoted by
$\xy
(0,0)*+{I_0}*\cir{}; (10,0), **@{-};    (10,0)*+{\circ};
\endxy$.
We will prove this Theorem later in Section \ref{sec:??}.

\subsection{Feynman rules for $I_k$}

By \eqref{def:Ik} and the above Feynman rules for $I_0$,
one can express each $I_k$ as a sum over some Feynman diagrams.
For example,
\ben
\frac{I_1}{2} & = & \frac{1}{2} t_1 + \frac{1}{2} \sum_{n \geq 1} t_{n+1} \frac{I_0^n}{n!} \\
& = & \frac{1}{2} t_1 + \frac{1}{2} \sum_{n =1}^\infty \frac{t_{n+1}}{n!} \biggl(t_0 + t_0t_1 + (\frac{1}{2} t_0^2 t_2 + t_0t_1^2 )
+ (\frac{1}{6} t_0^3t_3 + \frac{3}{2} t_0^2t_1t_2 + t_0t_1^3) \\
& + & (\frac{1}{24}t_0^4t_4 + \frac{2}{3} t_0^3t_1t_3 + \frac{1}{2} t_0^3t_2^2
+ 3 t_0^2t_1^2t_2 + t_0t_1^4)  + \cdots\biggr)^n \\
& = & \frac{1}{2} t_1 + \half t_0t_2
+ \biggl(\half t_0t_1 t_2 + \frac{1}{4}t_0^2 t_3 \biggr) \\
& + & \biggl( \frac{1}{4}t_0^2t_2^2 + \half t_0t_1^2t_2 + \half t_0^2t_1t_3
+ \frac{1}{12} t_0^3t_4 \biggr) \\
& + & \biggl(\frac{1}{3}t_0^3t_2t_3+\frac{3}{4} t_0^2t_1t_2^2+\frac{1}{2}t_0t_1^3t_2
+ \frac{3}{4}  t_0^2t_1^2 t_3+ \frac{1}{4} t_0^3t_1t_4 + \frac{1}{48} t_0^4t_5 \biggr) \\
& + & \biggl(\half t_0t_1^4t_2 +\frac{3}{2}t_0^2t_1^2t_2^2 + \frac{1}{4} t_0^3t_2^3
+ \frac{4}{3} t_0^3t_1t_2t_3 + t_0^2t_1^3t_3 \\
& + & \frac{7}{48} t_0^4t_2 t_4
+ \frac{1}{12} t_0^4t_3^2+ \frac{1}{2}  t_0^3t_1^2t_4 + \frac{1}{12}t_0^4t_1t_5
+ \frac{1}{240} t_0^5t_6 \biggr)+ \cdots.
\een
The Feynman diagrams for $\frac{I_1}{2!}$ are:
$$
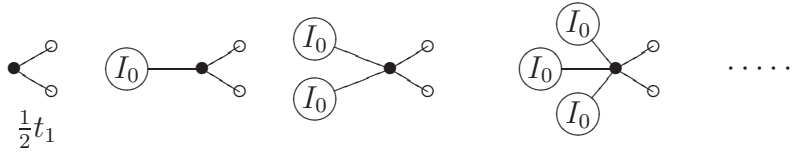
\xy
(-10,3); (-15,0), **@{-}; (-10,-3), **@{-}; (-10,3)*+{\circ}; (-15,0)*+{\bullet}; (-10,-3)*+{\circ};
(-12,-8)*+{\half t_1};
(0,0)*+{I_0}*\cir{}; (10,0), **@{-}; (15,3);  (10,0), **@{-};(15,-3), **@{-};
 (10,0)*+{\bullet}; (15,3)*+{\circ}; (15,-3)*+{\circ};
(25,4)*+{I_0}*\cir{}; (35,0), **@{-}; (25,-4)*+{I_0}*\cir{}, **@{-};
(40,3); (35,0), **@{-}; (40,-3), **@{-};
(40,3)*+{\circ}; (40,-3)*+{\circ};  (35,0)*+{\bullet};
(55,0)*+{I_0}*\cir{}; (65,0), **@{-};  (60,6)*+{I_0}*\cir{}; (65,0), **@{-}; (60,-6)*+{I_0}*\cir{}, **@{-};
(65, 0)*+{\bullet};
(70,3); (65,0), **@{-}; (70,-3), **@{-}; (70,3)*+{\circ}; (70,-3)*+{\circ};
(85, 0)*+{\cdots\cdots};
\endxy
$$
where for example,
$\xy
(0,0)*+{I_0}*\cir{}; (10,0), **@{-}; (15,3);  (10,0), **@{-};(15,-3), **@{-};
 (10,0)*+{\bullet}; (15,3)*+{\circ}; (15,-3)*+{\circ};
 \endxy$
stands for the sum of all the Feynman diagrams of the following form:
$$
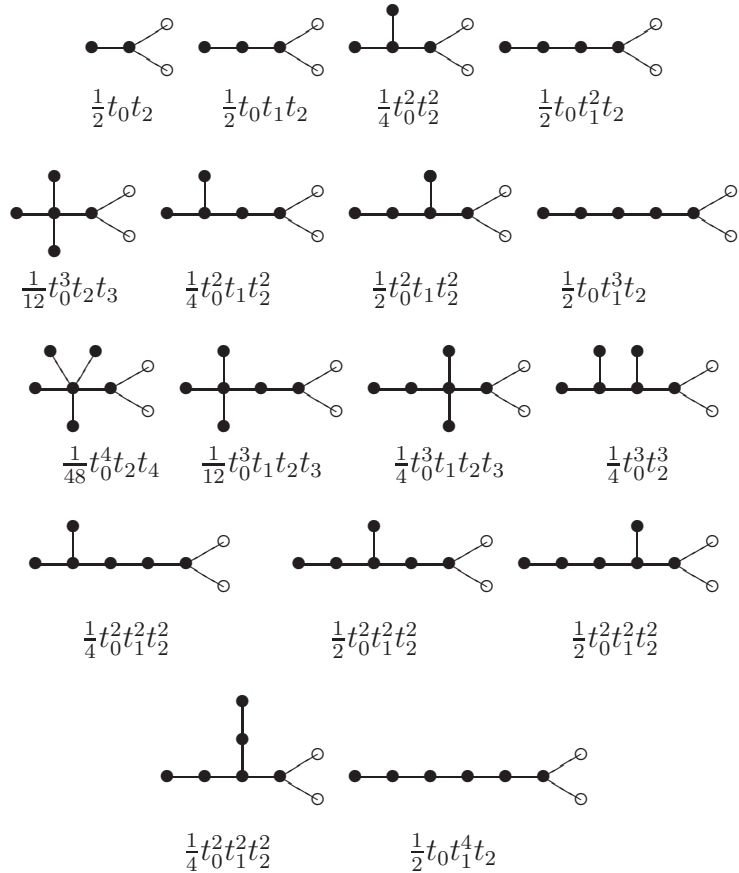
\xy
(0,0); (5,0), **@{-}; (10,3);  (5,0), **@{-};(10,-3), **@{-};
(0, 0)*+{\bullet}; (5,0)*+{\bullet}; (10,3)*+{\circ}; (10,-3)*+{\circ};
(4,-8)*+{\half t_0t_2};
(15,0); (25,0), **@{-}; (30,3); (25,0), **@{-}; (30,-3), **@{-};
(15, 0)*+{\bullet}; (20, 0)*+{\bullet}; (25,0)*+{\bullet}; (30,3)*+{\circ}; (30,-3)*+{\circ};
(23,-8)*+{\half t_0t_1t_2};
(35,0); (45,0), **@{-};  (40,0); (40,5), **@{-}; (50,3); (45,0), **@{-}; (50,-3), **@{-};
(35, 0)*+{\bullet}; (40, 0)*+{\bullet}; (40, 5)*+{\bullet};
(45,0)*+{\bullet}; (50,3)*+{\circ}; (50,-3)*+{\circ};
(42,-8)*+{\frac{1}{4} t_0^2t_2^2};
(55,0); (70,0), **@{-}; (75,3); (70,0), **@{-}; (75,-3), **@{-};
(55, 0)*+{\bullet}; (60, 0)*+{\bullet}; (65, 0)*+{\bullet};
(70,0)*+{\bullet}; (75,3)*+{\circ}; (75,-3)*+{\circ};
(65,-8)*+{\half t_0 t_1^2 t_2};
\endxy
$$

$$
\xy
(-5,0); (5,0), **@{-}; (0,5); (0,-5), **@{-}; (10,3); (5,0), **@{-}; (10,-3), **@{-};
(0, 0)*+{\bullet}; (-5, 0)*+{\bullet};(0, 5)*+{\bullet};(0, -5)*+{\bullet}; (5,0)*+{\bullet};
(10,3)*+{\circ}; (10,-3)*+{\circ};
(2,-10)*+{\frac{1}{12} t_0^3t_2t_3};
(15,0); (30,0), **@{-}; (20,5); (20,0), **@{-}; (35,3); (30, 0), **@{-}; (35,-3), **@{-};
(15, 0)*+{\bullet}; (20, 0)*+{\bullet}; (25, 0)*+{\bullet}; (20, 5)*+{\bullet};
(30,0)*+{\bullet}; (35,3)*+{\circ}; (35,-3)*+{\circ};
(23,-10)*+{\frac{1}{4} t_0^2t_1t_2^2};
(50,0); (50,5), **@{-}; (50,5)*+{\bullet};
(40,0); (55,0), **@{-};  (60,3); (55,0), **@{-}; (60,-3), **@{-};
(40, 0)*+{\bullet}; (45, 0)*+{\bullet}; (50, 0)*+{\bullet}; (50, 5)*+{\bullet};
(55,0)*+{\bullet}; (60, 3)*+{\circ}; (60,-3)*+{\circ};
(48,-10)*+{\frac{1}{2} t_0^2t_1t_2^2};
(65,0); (85,0), **@{-}; (90,3); (85,0), **@{-}; (90,-3), **@{-};
(65, 0)*+{\bullet}; (70, 0)*+{\bullet}; (75, 0)*+{\bullet};
(80, 0)*+{\bullet}; (85, 0)*+{\bullet}; (90,3)*+{\circ}; (90,-3)*+{\circ};
(73,-10)*+{\half t_0t_1^3t_2};
\endxy
$$

$$
\xy
(-5,0); (5,0), **@{-}; (-3,5); (0,0), **@{-}; (3,5); (0,0), **@{-}; (0,0); (0,-5), **@{-};
(10,3); (5,0), **@{-}; (10,-3), **@{-};
(0, 0)*+{\bullet}; (-5, 0)*+{\bullet};(-3, 5)*+{\bullet}; (3, 5)*+{\bullet};
(0, -5)*+{\bullet}; (5,0)*+{\bullet}; (10,3)*+{\circ}; (10,-3)*+{\circ};
(5,-10)*+{\frac{1}{48} t_0^4t_2t_4};
(15,0); (30,0), **@{-}; (35,3); (30,0), **@{-}; (35,-3), **@{-};
(15, 0)*+{\bullet}; (20, 0)*+{\bullet}; (25, 0)*+{\bullet};
(30,0)*+{\bullet}; (20,-5); (20,5), **@{-}; (20,5)*+{\bullet}; (20,-5)*+{\bullet};
(35,3)*+{\circ}; (35,-3)*+{\circ};
(25,-10)*+{\frac{1}{12}t_0^3t_1t_2t_3};
(40,0); (55,0), **@{-};  (50,-5); (50,5), **@{-}; (60,3); (55,0), **@{-}; (60,-3), **@{-};
(40, 0)*+{\bullet}; (45, 0)*+{\bullet}; (50, 0)*+{\bullet};(50, 5)*+{\bullet}; (50, -5)*+{\bullet};
(55,0)*+{\bullet}; (60,3)*+{\circ}; (60,-3)*+{\circ};
(50,-10)*+{\frac{1}{4}t_0^3t_1t_2t_3};
(65,0); (80,0), **@{-}; (70,0); (70,5), **@{-}; (75, 0); (75,5), **@{-};
(85,3); (80,0), **@{-}; (85,-3), **@{-};
(65, 0)*+{\bullet}; (70, 0)*+{\bullet}; (75, 0)*+{\bullet};
(70, 5)*+{\bullet}; (75, 5)*+{\bullet}; (80,0)*+{\bullet};(85,3)*+{\circ}; (85,-3)*+{\circ};
(75,-10)*+{\frac{1}{4} t_0^3t_2^3};
\endxy
$$

$$
\xy
(0,0); (20,0), **@{-}; (5,0); (5,5), **@{-}; (25,3); (20, 0), **@{-}; (25,-3), **@{-};
(0, 0)*+{\bullet}; (5, 0)*+{\bullet};(5, 5)*+{\bullet}; (10, 0)*+{\bullet}; (15, 0)*+{\bullet};
(20,0)*+{\bullet}; (25,3)*+{\circ}; (25,-3)*+{\circ};
(12,-10)*+{\frac{1}{4}t_0^2t_1^2t_2^2};
(35,0); (55,0), **@{-}; (45,0); (45,5), **@{-}; (60,3); (55,0), **@{-}; (60,-3), **@{-};
(35, 0)*+{\bullet}; (40, 0)*+{\bullet};(45, 5)*+{\bullet}; (45, 0)*+{\bullet};
(50, 0)*+{\bullet};  (55,0)*+{\bullet}; (60,3)*+{\circ}; (60,-3)*+{\circ};
(45,-10)*+{\frac{1}{2} t_0^2t_1^2t_2^2};
(65,0); (85,0), **@{-}; (80,0); (80,5), **@{-}; (90,3); (85,0), **@{-}; (90,-3), **@{-};
(65, 0)*+{\bullet}; (70, 0)*+{\bullet};(80, 5)*+{\bullet}; (75, 0)*+{\bullet}; (80, 0)*+{\bullet};
(85,0)*+{\bullet}; (90,3)*+{\circ}; (90,-3)*+{\circ};
(77,-10)*+{\half t_0^2t_1^2t_2^2};
\endxy$$

$$\xy
(0,0); (15,0), **@{-}; (10,0); (10,10), **@{-}; (20,3); (15,0), **@{-}; (20,-3), **@{-};
(0, 0)*+{\bullet}; (5, 0)*+{\bullet};(10, 5)*+{\bullet}; (10, 10)*+{\bullet}; (10, 0)*+{\bullet};
(15,0)*+{\bullet}; (20,3)*+{\circ}; (20,-3)*+{\circ};
(8,-10)*+{\frac{1}{4} t_0^2t_1^2t_2^2};
(25,0); (50,0), **@{-}; (55,3); (50,0), **@{-}; (55,-3), **@{-};
(30, 0)*+{\bullet}; (35, 0)*+{\bullet};(25, 0)*+{\bullet}; (40, 0)*+{\bullet}; (45, 0)*+{\bullet};
(50,0)*+{\bullet}; (55,3)*+{\circ}; (55,-3)*+{\circ};
(38,-10)*+{\half t_0t_1^4t_2};
\endxy
$$
These diagrams can be obtained by grafting the diagrams in $\xy
(0,0)*+{I_0}*\cir{}; (10,0), **@{-};    (10,0)*+{\circ};
\endxy$
to $\xy
(5,0); (10,0), **@{-}; (15,3);  (10,0), **@{-};(15,-3), **@{-};
 (10,0)*+{\bullet}; (15,3)*+{\circ}; (15,-3)*+{\circ}; (5,0)*+{\circ};
 \endxy$.
Similarly for other diagrams.

\begin{defn}
By a tree of type $k$ we mean a rooted tree whose vertices are marked either by $\bullet$ or $\circ$,
the root vertex $v_0$ is marked by $\bullet$,
and there are exactly $k$ vertices marked by $\circ$, all of which are adjacent to $v_0$.
\end{defn}

\begin{thm}
The formal power series $I_k$ is given by a sum over rooted trees of type $k$:
\be
I_k = \sum_{\text{$\Gamma$ is a rooted tree of type $k$}} \frac{1}{|\Aut \Gamma|} w_\Gamma,
\ee
where the weight of $\Gamma$ is given by
\be
w_\Gamma = \prod_{v\in V(\Gamma)} w_v \cdot \prod_{e\in E(\Gamma)} w_e,
\ee
with $w_e$ and $w_v$ given by the following Feynman rule:
\bea
&& w(e) = 1, \\
&& w(v) = \begin{cases}
t_{\val(v)-1}, & \text{if $v$ is not the root vertex $\circ$}, \\
1, & \text{if $v$ is the root vertex $\circ$}.
\end{cases}
\eea
\end{thm}

\begin{proof}
This is clear from the formula
\be
I_k= \sum_{n \geq 0} t_{n+k} \frac{I_0^n}{n!}
\ee
and Theorem \ref{thm:Rooted trees}.
\end{proof}

\subsection{Feynman rules for $I_{-1}$}

\begin{thm}
The formal power series $I_{-1}$ is given by a sum over rooted trees:
\be
I_{-1} = \sum_{\text{$\Gamma$ is a rooted tree}} \frac{1}{|\Aut \Gamma|} w_\Gamma,
\ee
where the weight of $\Gamma$ is given by
\be
w_\Gamma = \prod_{v\in V(\Gamma)} w_v \cdot \prod_{e\in E(\Gamma)} w_e,
\ee
with $w_e$ and $w_v$ given by the following Feynman rule:
\bea
&& w(e) = 1, \\
&& w(v) = t_{\val(v)-1},
\eea
\end{thm}

\begin{proof}
This is clear from the formula
\be
I_{-1} = \sum_{n =0}^\infty t_n \frac{I_0^{n+1}}{(n+1)!}
\ee
and Theorem \ref{thm:Rooted trees}.
\end{proof}

\subsection{Feynman rules for $I_{-1} - \frac{1}{2}I_0^2$}

\begin{prop}
We have
\be \label{eqn:D I-1}
\frac{\pd}{\pd t_0} (I_{-1} - \frac{1}{2}I_0^2) = I_0.
\ee
\end{prop}

\begin{proof}
Recall
$$
I_{-1} - \frac{1}{2}I_0^2
= \sum_{k=0}^\infty  \frac{(-1)^k}{(k+1)!} (I_k+\delta_{k,1}) I_0^{k+1},
$$
and we have
$$
\frac{\pd}{\pd t_0}
= \frac{1}{1-I_1} \frac{\pd}{\pd I_0} + \sum_{l \geq 1} \frac{I_{l+1}}{1-I_1} \frac{\pd}{\pd I_l}.
$$
The proof is completed by a simple calculation.
\end{proof}

As a corollary to this Proposition and Theorem \ref{thm:Rooted trees},
we have:

\begin{thm} \label{thm:Trees}
The formal power series $I_{-1}$ is given by a sum over rooted trees:
\be
I_{-1} - \half I_0^2 = \sum_{\text{$\Gamma$ is a tree}} \frac{1}{|\Aut \Gamma|} w_\Gamma,
\ee
where the weight of $\Gamma$ is given by
\be
w_\Gamma = \prod_{v\in V(\Gamma)} w_v \cdot \prod_{e\in E(\Gamma)} w_e,
\ee
with $w_e$ and $w_v$ given by the following Feynman rules:
\bea
&& w(e) = 1, \\
&& w(v) = t_{\val(v)-1},
\eea
\end{thm}

For example,
$$
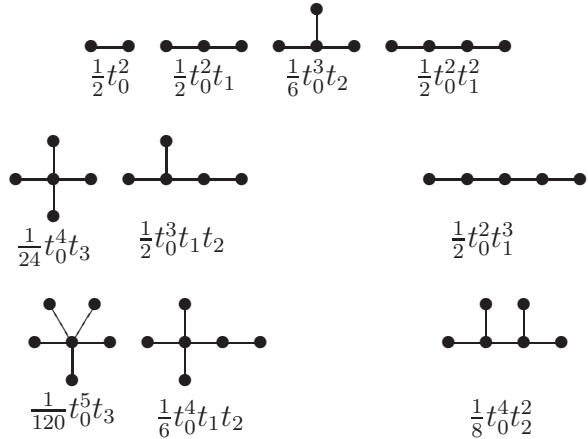
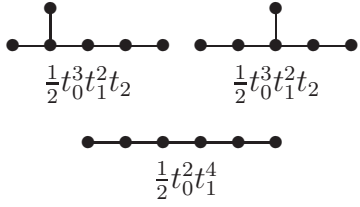
\xy
(0,0); (5,0), **@{-}; (0, 0)*+{\bullet}; (5,0)*+{\bullet}; (2.5,-4)*+{\half t_0^2};
(10,0); (20,0), **@{-};  (10, 0)*+{\bullet}; (15, 0)*+{\bullet}; (20,0)*+{\bullet};
(15,-4)*+{\half t_0^2t_1};
(25,0); (35,0), **@{-};  (30,0); (30,5), **@{-};
(25, 0)*+{\bullet}; (30, 0)*+{\bullet}; (30, 5)*+{\bullet}; (35,0)*+{\bullet};
(30,-4)*+{\frac{1}{6} t_0^3t_2};
(40,0); (55,0), **@{-};  (40, 0)*+{\bullet}; (45, 0)*+{\bullet}; (50, 0)*+{\bullet};
(55,0)*+{\bullet};
(47.5,-4)*+{\half t_0^2t_1^2};
\endxy
$$

$$
\xy
(-5,0); (5,0), **@{-}; (0,5); (0,-5), **@{-};
(0, 0)*+{\bullet}; (-5, 0)*+{\bullet};(0, 5)*+{\bullet};(0, -5)*+{\bullet}; (5,0)*+{\bullet};
(0,-9)*+{\frac{1}{24}t_0^4t_3};
(10,0); (25,0), **@{-};  (10, 0)*+{\bullet}; (15, 0)*+{\bullet}; (20, 0)*+{\bullet};
(15,0); (15,5), **@{-}; (15,5)*+{\bullet}; (25,0)*+{\bullet};
(17,-8)*+{\half t_0^3t_1t_2};
(50,0); (70,0), **@{-};  (50, 0)*+{\bullet}; (55, 0)*+{\bullet}; (60, 0)*+{\bullet};
(65, 0)*+{\bullet}; (70,0)*+{\bullet};
(57,-8)*+{\half t_0^2t_1^3};
\endxy
$$

$$
\xy
(-5,0); (5,0), **@{-}; (-3,5); (0,0), **@{-}; (3,5); (0,0), **@{-}; (0,0); (0,-5), **@{-};
(0, 0)*+{\bullet}; (-5, 0)*+{\bullet};(-3, 5)*+{\bullet}; (3, 5)*+{\bullet};
(0, -5)*+{\bullet}; (5,0)*+{\bullet};
(0,-9)*+{\frac{1}{120}t_0^5t_3};
(10,0); (25,0), **@{-};  (10, 0)*+{\bullet}; (15, 0)*+{\bullet}; (20, 0)*+{\bullet};
(25,0)*+{\bullet};
(15,-5); (15,5), **@{-}; (15,5)*+{\bullet}; (15,-5)*+{\bullet};
(17,-10)*+{\frac{1}{6} t_0^4t_1t_2};
(50,0); (65,0), **@{-}; (55,0); (55,5), **@{-}; (60, 0); (60,5), **@{-};
(50, 0)*+{\bullet}; (55, 0)*+{\bullet}; (60, 0)*+{\bullet};
(55, 5)*+{\bullet}; (60, 5)*+{\bullet}; (65,0)*+{\bullet};
(57,-10)*+{\frac{1}{8} t_0^4t_2^2};
\endxy
$$

$$
\xy
(0,0); (20,0), **@{-}; (5,0); (5,5), **@{-};
(0, 0)*+{\bullet}; (5, 0)*+{\bullet};(5, 5)*+{\bullet}; (10, 0)*+{\bullet};
(15, 0)*+{\bullet}; (20,0)*+{\bullet};
(10,-5)*+{\half t_0^3t_1^2t_2};
(25,0); (45,0), **@{-}; (35,0); (35,5), **@{-};
(25, 0)*+{\bullet}; (30, 0)*+{\bullet};(35, 5)*+{\bullet}; (35, 0)*+{\bullet};
(40, 0)*+{\bullet}; (45,0)*+{\bullet};
(35,-5)*+{\half t_0^3t_1^2t_2};
\endxy
$$
$$
\xy
(-5,0); (20,0), **@{-};
(0, 0)*+{\bullet}; (5, 0)*+{\bullet};(-5, 0)*+{\bullet}; (10, 0)*+{\bullet};
(15, 0)*+{\bullet}; (20,0)*+{\bullet};
(8,-5)*+{\half t_0^2t_1^4};
\endxy
$$
These diagrams give the first few terms of $I_{-1}  - \half I_0^2$:
\ben
I_{-1} - \half I_0^2 & = & \half t_0^2 + \half t_0^2t_1
+ \frac{1}{6} t_0^3t_2 + \half t_0^2t_1^2 \\
& + & \frac{1}{24}t_0^4t_3 + \half t_0^3t_1t_2 + \half t_0^2t_1^3 \\
& + & \frac{1}{120}t_0^5t_3 + \frac{1}{6} t_0^4t_1t_2 + \frac{1}{8} t_0^4t_2^2 \\
& + & \half t_0^3t_1^2t_2 + \half t_0^3t_1^2t_2 + \half t_0^2t_1^4 + \cdots.
\een

\subsection{Feynman rules for $t_k$}

By \eqref{eqn:T-in-I} we have:
\ben
&& \frac{1}{(k+1)!}t_k = \sum_{n=0}^\infty (-1)^n  \frac{1}{n!(k+1)!}I_0^nI_{n+k}.
\een
The right-hand side is a summation over trees with one $\bullet$-vertex of valence $n+k+1$,
on which $n$ edges are attached, connecting to $n$ $\bullet$-vertex of valence $1$,
and $k+1$ edges are attached, connecting to $k+1$ $\circ$-vertex of valence $1$.
One can easily formulate Feynman rules for the right-hand side,
a term $\frac{I_0^n}{n!}I_{n+k}$ corresponds to a rooted tree with $n$,
the rules for the weights are:
\be
w_v = \begin{cases}
-I_0, & \text{if $v$ is a $\bullet$-vertex of valence $1$}, \\
I_{n+k+1}, & \text{if $v$ is the $\bullet$-vertex of valence $n+k1$}, \\
1, & \text{if $v$ is a $\circ$-vertex of valence $1$},
\end{cases}
\ee
\be
w_e = 1.
\ee
The factor $n!(k+1)!$ is exactly the order of the automorphism group of this tree.

\section{Mean Field Theory of the Topological 1D Gravity}

\label{sec:1D-TG}

The Feynman diagrams and Feynman rules in last Section suggests that they come from
some quantum field theory.
In this Section we propose that they arise in the mean field theory of the topological 1D gravity.

\subsection{Gaussian integrals and some properties}

Recall the Gaussian integrals ($a>0$):
\be \label{eqn:Gaussian}
\frac{1}{\sqrt{2\pi}} \int_\bR dx x^n e^{ -\frac{a}{2} x^2}
= \begin{cases}
0, & \text{if $n$ is odd}, \\
\frac{(2m)!}{m!2^ma^{m+1/2}} = \frac{(2m-1)!!}{a^{m+1/2}}, & \text{if $n =2m$}.
\end{cases}
\ee
In this paper,
we will use some properties of the Gaussian integrals summarized in the following:

\begin{prop}
Gaussian integrals have the following properties:
\begin{itemize}
\item[(1)] (Scaling of variable) For $a, b > 0$,
\be \label{eqn:Gaussian-Scaling}
\frac{1}{\sqrt{2\pi}} \int_\bR dx \cdot x^n e^{ -\frac{a}{2} x^2}
 = \biggl( \frac{b}{a} \biggr)^{(n+1)/2} \frac{1}{\sqrt{2\pi}} \int_\bR dx  \cdot x^n e^{ -\frac{b}{2} x^2}.
\ee
\item[(2)] (Translation of variable) For $a > 0$ and any $c\in \bR$,
\be  \label{eqn:Gaussian-Translation}
\frac{1}{\sqrt{2\pi}} \int_\bR dx \cdot x^n e^{ -\frac{a}{2} x^2}
= e^{- \frac{ac^2}{2}} \sum_{j=0}^\infty \frac{1}{\sqrt{2\pi}} \int_\bR dx
(x+c)^n (-ac)^j \frac{x^j}{j!} \cdot e^{ -\frac{a}{2} x^2}.
\ee
\item[(3)] (Separation of the square term) When $a> 0$ and $a+b > 0$,
\be \label{eqn:Gaussian-Separation}
\frac{1}{\sqrt{2\pi}} \int_\bR dx \cdot x^n e^{ -\frac{a+b}{2} x^2}
= \sum_{j \geq 0} \frac{b^j}{2^jj!} \frac{1}{\sqrt{2\pi}} \int_\bR dx \cdot x^{n+2j} e^{ -\frac{a}{2} x^2}.
\ee
\item[(4)] (Integration by parts) For $a> 0$,
\be  \label{eqn:Gaussian-By-Parts}
\frac{1}{\sqrt{2\pi}} \int_\bR dx \cdot \pd_x \biggl(x^n e^{ -\frac{a}{2} x^2}\biggr) = 0.
\ee
\end{itemize}
\end{prop}

\begin{proof}
These easily follow from ordinary properties of integrals.
However,
for our purpose,
we will need proofs based on \eqref{eqn:Gaussian} only.
It is clear that \eqref{eqn:Gaussian-Scaling} follows from \eqref{eqn:Gaussian}.
Now we prove \eqref{eqn:Gaussian-Translation}.
First let $n = 2m$,
then
\ben
&& e^{- \frac{ac^2}{2}} \sum_{j=0}^\infty \frac{1}{\sqrt{2\pi}} \int_\bR dx
(x+c)^{n} (-ac)^j \frac{x^j}{j!} \cdot e^{ -\frac{a}{2} x^2} \\
& = & e^{- \frac{ac^2}{2}} \sum_{j=0}^\infty \frac{1}{\sqrt{2\pi}} \int_\bR dx
\sum_{k=0}^{n} \binom{n}{k} x^{k} c^{n-k} \cdot (-ac)^j \frac{x^j}{j!} \cdot e^{ -\frac{a}{2} x^2} \\
& = & e^{- \frac{ac^2}{2}} \cdot \sum_{k=0}^{n}  \binom{n}{k} c^{n-k}
\cdot \frac{1}{\sqrt{2\pi}} \int_\bR dx
  \cdot \sum_{j=0}^\infty (-ac)^j \frac{x^{j+k}}{j!} \cdot e^{ -\frac{a}{2} x^2}.
\een
When $k=2l$,
\ben
&&  \frac{1}{\sqrt{2\pi}} \int_\bR dx
  \cdot \sum_{j=0}^\infty (-ac)^j \frac{x^{j+2l}}{j!} \cdot e^{ -\frac{a}{2} x^2} \\
& = & \sum_{j=0}^\infty \frac{(-ac)^{2j}}{(2j)!} \frac{(2j+2l-1)!!}{a^{j+l+1/2}} \\
& = & \frac{1}{a^{l+1/2}} \sum_{j=0}^\infty \frac{(ac^2)^{j}}{2^jj!} \frac{(2j+2l-1)!!}{(2j-1)!!} \\
& = & \frac{1}{a^{l+1/2}} \sum_{j=0}^l \frac{(2l)!}{(2l-2j)!j!2^j} (ac^2)^{l-j} \cdot e^{ac^2/2};
\een
when $k=2l+1$,
\ben
&&  \frac{1}{\sqrt{2\pi}} \int_\bR dx
  \cdot \sum_{j=0}^\infty (-ac)^j \frac{x^{j+2l+1}}{j!} \cdot e^{ -\frac{a}{2} x^2} \\
& = & \sum_{j=0}^\infty \frac{(-ac)^{2j+1}}{(2j+1)!} \frac{(2j+2l+1)!!}{a^{j+l+3/2}} \\
& = & -\frac{c}{a^{l+1/2}} \sum_{j=0}^\infty \frac{(ac^2)^{j}}{2^jj!} \frac{(2j+2l+1)!!}{(2j+1)!!} \\
& = & -\frac{c}{a^{l+1/2}} \sum_{j=0}^l \frac{(2l+1)!}{(2l+1-2j)!j!2^j} (ac^2)^{l-j} \cdot e^{ac^2/2}.
\een
In the above we have used the identities in Lemma \ref{lm:Summation}.
When  $n =2m$,
\ben
&& e^{- \frac{ac^2}{2}} \sum_{j=0}^\infty \frac{1}{\sqrt{2\pi}} \int_\bR dx
(x+c)^{n} (-ac)^j \frac{x^j}{j!} \cdot e^{ -\frac{a}{2} x^2} \\
& = & \sum_{l=0}^m \binom{2m}{2l} c^{2m-2l} \cdot
\frac{1}{a^{l+1/2}} \sum_{j=0}^l \frac{(2l)!}{(2l-2j)!j!2^j} (ac^2)^{l-j}\\
& - & \sum_{l=0}^{m-1} \binom{2m}{2l+1} c^{2m-2l-1} \cdot
\frac{c}{a^{l+1/2}} \sum_{j=0}^l \frac{(2l+1)!}{(2l+1-2j)!j!2^j} (ac^2)^{l-j} \\
& = & \sum_{j=0}^m \frac{c^{2m-2j}}{a^{j+1/2}}  (2j-1)!! \sum_{l=j}^m\binom{2m}{2l}   \cdot
 \frac{(2l)!}{(2l-2j)!(2j)!} \\
& - & \sum_{j=0}^{m-1}\frac{c^{2m-2j}}{a^{j+1/2}} (2j-1)!! \sum_{l=j}^{m-1} \binom{2m}{2l+1}  \cdot
\frac{(2l+1)!}{(2l+1-2j)!(2j)!}  \\
& = & \frac{(2m-1)!!}{a^{m+1/2}}
+ \sum_{j=0}^{m-1} \frac{c^{2m-2j}}{a^{j+1/2}}  (2j-1)!! \\
&& \cdot \biggl(\sum_{k=0}^{m-j} \binom{2m}{2m-2j-2k, 2j, 2k}
- \sum_{k=0}^{m-j-1} \binom{2m}{2m-2j-2k-1, 2k+1, 2j}  \biggr) \\
& = & \frac{(2m-1)!!}{a^{m+1/2}},
\een
in the last equality we have used the fact that
\ben
&& \sum_{k=0}^{m-j} \binom{2m}{2m-2j-2k, 2j, 2k}
- \sum_{k=0}^{m-j-1} \binom{2m}{2m-2j-2k-1, 2k+1, 2j} = 0
\een
because it is the coefficient of $z^{2j}$ in the expansion of$(1-1+z)^{2m}$.
Similarly,
when  $n =2m+1$,
\ben
&& e^{- \frac{ac^2}{2}} \sum_{j=0}^\infty \frac{1}{\sqrt{2\pi}} \int_\bR dx
(x+c)^{n} (-ac)^j \frac{x^j}{j!} \cdot e^{ -\frac{a}{2} x^2} \\
& = & \sum_{l=0}^m \binom{2m+1}{2l} c^{2m+1-2l} \cdot
\frac{1}{a^{l+1/2}} \sum_{j=0}^l \frac{(2l)!}{(2l-2j)!j!2^j} (ac^2)^{l-j}\\
& - & \sum_{l=0}^{m} \binom{2m+1}{2l+1} c^{2m-2l} \cdot
\frac{c}{a^{l+1/2}} \sum_{j=0}^l \frac{(2l+1)!}{(2l+1-2j)!j!2^j} (ac^2)^{l-j} \\
& = & \sum_{j=0}^m \frac{c^{2m+1-2j}}{a^{j+1/2}}  (2j-1)!! \sum_{l=j}^m\binom{2m+1}{2l}   \cdot
 \frac{(2l)!}{(2l-2j)!(2j)!} \\
& - & \sum_{j=0}^{m}\frac{c^{2m+1-2j}}{a^{j+1/2}} (2j-1)!! \sum_{l=j}^{m} \binom{2m+1}{2l+1}  \cdot
\frac{(2l+1)!}{(2l+1-2j)!(2j)!}  \\
& = &  \sum_{j=0}^{m} \frac{c^{2m-2j}}{a^{j+1/2}}  (2j-1)!! \\
&& \cdot \biggl(\sum_{k=0}^{m-j} \binom{2m+1}{2m-2j-2k+1, 2j, 2k}
- \sum_{k=0}^{m-j} \binom{2m+1}{2m-2j-2k, 2k+1, 2j}  \biggr) \\
& = & 0,
\een
in the last equality we have used the fact that
\ben
&& \sum_{k=0}^{m-j} \binom{2m+1}{2m+1-2j-2k, 2j, 2k}
- \sum_{k=0}^{m-j} \binom{2m}{2m-2j-2k, 2k+1, 2j} = 0
\een
because it is the coefficient of $z^{2j}$ in the expansion of$(1-1+z)^{2m+1}$.
This completes the proof of \eqref{eqn:Gaussian-Translation}.

When $n=2m$,
\ben
&& \frac{1}{\sqrt{2\pi}} \int_\bR dx \cdot x^{2m} e^{ -\frac{a+b}{2} x^2}
= \frac{(2m-1)!!}{(a+b)^{m+1/2}},\\
&&  \sum_{j \geq 0} \frac{b^j}{2^jj!} \frac{1}{\sqrt{2\pi}} \int_\bR dx \cdot x^{2m+2j} e^{ -\frac{a}{2} x^2}
= \sum_{j \geq 0} \frac{b^j}{2^jj!} \frac{(2m+2j-1)!!}{a^{m+j+1/2}},
\een
so one need to check that
\ben
\frac{a^{m+1/2}}{(a+b)^{m+1/2}}  = \sum_{j \geq 0} \frac{(2m+2j-1)!!}{2^jj!(2m-1)!!} \frac{b^j}{a^j},
\een
but this is just a special case of Taylor expansion.
When $n=2m+1$,
\ben
&& \frac{1}{\sqrt{2\pi}} \int_\bR dx \cdot x^{2m+1} e^{ -\frac{a+b}{2} x^2} = 0, \\
&&  \sum_{j \geq 0} \frac{b^j}{2^jj!} \frac{1}{\sqrt{2\pi}} \int_\bR dx \cdot x^{2m+1+2j} e^{ -\frac{a}{2} x^2}
= 0.
\een
Therefore, we have proved \eqref{eqn:Gaussian-Separation}.

When $n =2m$
\ben
&& \frac{1}{\sqrt{2\pi}} \int_\bR dx \cdot \pd_x \biggl(x^{2m} e^{ -\frac{a}{2} x^2}\biggr) \\
& = & \frac{1}{\sqrt{2\pi}} \int_\bR dx \cdot (2m x^{2m-1}-a x^{2m+1}) e^{ -\frac{a}{2} x^2} \frac{1}{\sqrt{2\pi}}
= 0;
\een
when $n = 2m+1$,
\ben
&& \frac{1}{\sqrt{2\pi}} \int_\bR dx \cdot \pd_x \biggl(x^{2m+1} e^{ -\frac{a}{2} x^2}\biggr) \\
& = & \frac{1}{\sqrt{2\pi}} \int_\bR dx \cdot ((2m+1) x^{2m-1}-ax^{2m+2}) e^{ -\frac{a}{2} x^2} \frac{1}{\sqrt{2\pi}} \\
& = & (2m+1) \cdot \frac{(2m-1)!!}{a^{m+1/2}} - a \cdot \frac{(2m+1)!!}{a^{m+1+1/2}} = 0.
\een
So \eqref{eqn:Gaussian-By-Parts} is proved.
\end{proof}

\begin{lem} \label{lm:Summation}
For $l \geq 0$,
\be
\sum_{j \geq 0} \frac{x^{j}}{j!2^j} \frac{(2j+2l-1)!!}{(2j-1)!!}
= \sum_{j=0}^l \frac{(2l)!}{(2l-2j)!j!2^j} x^{l-j} \cdot e^{x/2},
\ee
\be
\sum_{j \geq 0} \frac{x^{j}}{j!2^j} \frac{(2j+2l+1)!!}{(2j+1)!!}
= \sum_{j=0}^l \frac{(2l+1)!}{(2l+1-2j)!j!2^j} x^{l-j} \cdot e^{x/2}.
\ee
\end{lem}

\begin{proof}
Note,
\ben
&& \sum_{j \geq 0} \frac{x^{j}}{j!2^j} \frac{(2j+2k-1)!!}{(2j-1)!!}
= (2x\frac{d}{dx} + 2k-1) \cdots (2x\frac{d}{dx} +1) e^{x/2}.
\een
It follows that
\ben
&& \sum_{j \geq 0} \frac{x^{j}}{j!2^j}  \frac{(2j+2k-1)!!}{(2j-1)!!}
= p_k(x) e^{x/2}
\een
for some polynomial of degree $k$,
and $p_0(x) = 1$.
Furthermore,
from
\ben
p_k(x) e^{x^2/2} = (2x\frac{d}{dx} + 2k-1) (p_{k-1}(x) e^{x/2})
\een
one derives a recursion relation:
\ben
p_k(x) = 2xp_{k-1}'(x) + (2k-1) p_{k-1}(x) + x p_{k-1}(x).
\een
It is straightforward to check that
$$p_k(x) = \sum_{j=0}^k \frac{(2k)!}{(2k-2j)!j!2^j} x^{k-j}
$$
is the unique solution of this recursion relation with initial value.
This proves the first identity.
The second identity can be proved in the same way.
\end{proof}

\subsection{Polymer model}

Consider the formal Gaussian integral:
\be
Z = \frac{1}{\sqrt{2\pi}\lambda}
\int  dx \exp \frac{1}{\lambda^2} \biggl( -\half x^2 + \sum_{n \geq 1} t_{n-1} \frac{x^n}{n!}  \biggr).
\ee
This is defined by first expanding $\exp \frac{1}{\lambda^2} \biggl(\sum_{n \geq 1} t_{n-1} \frac{x^n}{n!}  \biggr)$
as a formal power series in $x$ then taking the Gaussian integrals term by term.
After a change of variables:
\be
Z = \frac{1}{\sqrt{2\pi}}
\int  dx \exp \biggl( -\half x^2 + \sum_{n \geq 1} \frac{t_{n-1}}{n!} \lambda^{n-2} x^n \biggr).
\ee
As is well-known,
the asymptotic expansion is given by the following summation over Feynman diagrams:
\be \label{eqn:Feynman-Z}
Z = \sum_{\Gamma \in \cG} \frac{1}{|\Aut(\Gamma)|} \prod_{v\in V(\Gamma)} \lambda^{\val(v)-2} t_{\val(v)-1},
\ee
where the sum is taken over the set $\cG$ all possible graphs,
with the following Feynman rules:
\bea
&& w(v) = \lambda^{\val(v)-2} t_{\val(v)-1}, \\
&& w(e) = 1.
\eea
The free energy $F = \log Z$ is given by
\be  \label{eqn:Feynman-F}
F
= \sum_{\Gamma \in \cG^c} \frac{1}{|\Aut(\Gamma)|} \prod_{v\in V(\Gamma)} \lambda^{\val(v)-2} t_{\val(v)-1},
\ee
where the sum is taken over the set $\cG^c$ all possible connected graphs.
Write
\be
F = \sum_{g \geq 0} \lambda^{2g-2} F_g.
\ee
Then
\be \label{eqn:F0=Trees}
F_0 = \sum_{\text{$\Gamma$ is a tree}} \frac{1}{|\Aut(\Gamma)|} \prod_{v\in V(\Gamma)} t_{val(v)-1}.
\ee
The first few terms of $Z$ are
\ben
Z & = & 1 + (\frac{\lambda^{-2}}{2} t_0^2 + \half t_1)
+ 3 \cdot (\frac{t_0^4}{24}\lambda^{-4} + \frac{t_0^2t_1}{4} \lambda^{-2}
+ \frac{t_1^2}{8} + \frac{t_0t_2}{6} + \frac{t_3}{24} \lambda^2) \\
& + & 15 \cdot (\frac{t_0^6}{720} \lambda^{-6}  + \frac{t_1t_0^4}{48} \lambda^{-4}
 + \frac{t_0^2t_1^2}{16} \lambda^{-2} + \frac{t_0^3t_2}{36}\lambda^{-2}
 + \frac{t_1^3}{48}  + \frac{t_0t_1t_2}{12} \\
& + & \frac{t_0^2t_3}{48} + \frac{t_0t_4}{120} \lambda^2
+ \frac{t_1t_3}{48} \lambda^2 + \frac{t_2^2}{72} \lambda^2 + \frac{t_5}{720} \lambda^4  ) \\
& + & 105( \frac{t_0^8}{40320}\lambda^{-8}+\frac{t_1t_0^6}{1440}\lambda^{-6}
+\frac{t_2t_0^5 }{720}\lambda^{-4} + \frac{t_1^2t_0^4}{192}\lambda^{-4} \\
& + & \frac{t_2t_1t_0^3 }{72}\lambda^{-2} 
+ \frac{t_3t_0^4 }{576}\lambda^{-2}
+ \frac{t_1^3t_0^2 }{96}\lambda^{-2} \\
& + & \frac{t_4t_0^3}{720} + \frac{t_2^2t_0^2}{144} + \frac{t_1^4}{384}
+ \frac{t_3t_1t_0^2}{96} + \frac{t_2t_1^2t_0}{48}   \\
& + & \frac{t_5t_0^2}{1440}\lambda^2 
+  \frac{t_3t_1^2}{192}\lambda^2 +  \frac{t_2^2t_1}{144} \lambda^2 
+ \frac{ t_4t_0t_1}{240} \lambda^2 + \frac{t_3t_0t_2}{144} \lambda^2 \\
& + & \frac{t_3^2}{1152}\lambda^4 
  + \frac{t_5t_1}{1440}\lambda^4
+  \frac{t_4t_2}{720}\lambda^4   
+ \frac{t_0t_6}{5040} \lambda^4   
+ \frac{t_7}{40320} \lambda^6  ) + \cdots,
\een
and the first few terms of the free energy are given by:
\ben
F & = & (\frac{1}{2}\lambda^{-2}t_0^2+\frac{1}{2}t_1)
+ (\frac{1}{2}t_0^2t_1\lambda^{-2} +\frac{1}{2}t_0t_2+\frac{1}{4}t_1^2 + \frac{1}{8} t_3\lambda^2) \\
& + & (\frac{1}{2}t_0^2t_1^2\lambda^{-2}+\frac{1}{6} t_0^3t_2\lambda^{-2}
+ \frac{1}{6}t_1^3 + \frac{1}{4}t_0^2t_3+t_0t_2t_1 \\
& + & \frac{5}{24}t_2^2\lambda^2
+ \frac{1}{8}t_0t_4\lambda^2+\frac{1}{4}t_1t_3\lambda^2 +\frac{1}{48}t_5\lambda^4) \\
& + & (\frac{1}{2}\lambda^{-2}t_1^3t_0^2 + \frac{1}{2}\lambda^{-2}t_1t_2t_0^3 + \frac{1}{24}\lambda^{-2}t_3t_0^4 \\
& + & \frac{1}{8} t_1^4  + \frac{3}{4} t_0^2t_1t_3 + \frac{3}{2} t_0 t_1^2t_2 + \frac{1}{12} t_0^3 t_4
+ \frac{1}{2} t_0^2t_2^2 \\
& + & \frac{5}{8} \lambda^2 t_2^2t_1 + \frac{3}{8} \lambda^2t_4t_0t_1 + \frac{3}{8} \lambda^2t_3t_1^2
+ \frac{1}{16} \lambda^2t_5t_0^2   + \frac{2}{3} \lambda^2t_3t_0t_2 \\
& + & \frac{1}{16}\lambda^4t_5t_1
+ \frac{1}{12} \lambda^4t_3^2 + \frac{1}{48} \lambda^4t_0t_6
+ \frac{7}{48} \lambda^4t_4t_2 + \frac{1}{384} t_7\lambda^6 ) + \cdots.
\een
 
Feynman diagrams with no loops are exactly the same as those for $I_{-1} - 1 - \half I_0^2$
and the Feynman rules are modified by a power of $\lambda$ (cf. Theorem \ref{thm:Trees}).
The following are some Feynman diagrams with loops:
$$
\xy
(0,0)*\xycircle(2.5,2.5){}; (-2.5, 0)*+{\bullet};   (-2,-8)*+{\half t_1};
(17.5,0)*\xycircle(2.5,2.5){};
(10,0); (15,0), **@{-};  (10, 0)*+{\bullet}; (15, 0)*+{\bullet};
(15,-8)*+{\half t_0t_2};
(30,0)*\xycircle(2.5,2.5){};
(32.5, 0)*+{\bullet}; (27.5, 0)*+{\bullet};
(30,-8)*+{\frac{1}{4} t_1^2};
(42.5,0)*\xycircle(2.5,2.5){}; (47.5,0)*\xycircle(2.5,2.5){};  (40, 0)*+{\bullet}; (45, 0)*+{\bullet};
(50, 0)*+{\bullet};
(45,-8)*+{\frac{1}{8} t_3};
\endxy
$$
If one regards the vertices as atoms and the edges as chemical bonds,
then a Feynman diagram corresponds to a polymer
(possibly with self bonds).
One can also consider the dual diagrams,
they are branching chains \cite{Nishigaki-Yoneya1}.

We notice that the free energy has the following structure:
\ben
F & = & \frac{1}{\lambda^2} ( \frac{1}{2} \frac{t_0^2}{1-t_1}
+\frac{1}{6} \frac{t_0^3t_2}{(1-t_1)^3}
+ \frac{1}{24}\frac{t_3t_0^4}{(1-t_1)^4} + \cdots ) \\
& + & \biggl( \half \log \frac{1}{1-t_1} +   \frac{1}{2} \frac{t_0t_2}{(1-t_1)^2}
+ \frac{1}{4} \frac{t_0^2t_3}{(1-t_1)^3} + \cdots \biggr) \\
& + & \lambda^2\biggl( \frac{1}{8} \frac{t_3}{(1-t_1)^2} + \frac{1}{8} \frac{t_0t_4}{(1-t_1)^3}
+ \frac{5}{24} \frac{t_2^2}{(1-t_1)^3} + \cdots\biggr) + \cdots
\een
In fact we have

\begin{thm} \label{thm:1-t1}
The free energy of topological 1D gravity can be rewritten in the following form:
\be
F = \half \log (1-t_1)
+ \sum_{\substack{g,n \geq 0\\ 2g-2+n > 0}} \sum_{\substack{a_2, \dots, a_n \neq 1\\ \sum a_j = 2g-2+n }}
\frac{\corr{\prod\limits_{j=1}^n \tau_{a_j}}_g}{(1-t_1)^{g-1+n}} \prod\limits_{j=1}^n t_{a_j}.
\ee
\end{thm}

This can be easily proved by performing a scaling of variable to get:
\be
Z = \frac{(1-t_1)^{1/2}}{\sqrt{2\pi}\lambda}
\int_\bR dx \exp \frac{1}{\lambda^2} \biggl( -\half x^2
+ \sum_{\substack{n \geq 1\\n \neq 2}} \frac{t_{n-1}}{(1-t_1)^{n/2}} \frac{x^n}{n!}  \biggr).
\ee

The free energy $F$ now is a summation over all connected graphs without
vertices of valence $2$,
and the propagator is changed to
$$\frac{w_e}{1-w_v} = \frac{1}{1-t_1}.$$
I.e.,
one can take all the original Feynman diagrams
and ignore all the vertices of valence $2$.
This will produce all Feynman diagrams for the new Feynman rules,
except for the following cases:
$$
\xy
(-17.5,0)*\xycircle(2.5,2.5){}; (-10,0)*+{=}; (-18,-8)*+{\half\log \frac{1}{1-t_1}};
(-2.5,0)*\xycircle(2.5,2.5){}; (-5, 0)*+{\bullet};   (-2,-8)*+{\half t_1}; (5,0)*+{+};
(12.5,0)*\xycircle(2.5,2.5){};
(10, 0)*+{\bullet}; (15, 0)*+{\bullet}; (20,0)*+{+};
(12,-8)*+{\frac{1}{4} t_1^2};
(27.5,0)*\xycircle(2.5,2.5){};
(25, 0)*+{\bullet}; (27.5, 2.5)*+{\bullet}; (30, 0)*+{\bullet};
(27,-8)*+{\frac{1}{12} t_1^3}; (35,0)*+{+};
(42.5,0)*\xycircle(2.5,2.5){};    (40, 0)*+{\bullet}; (45, 0)*+{\bullet};
(42.5, 2.5)*+{\bullet}; (42.5, -2.5)*+{\bullet};
(42,-8)*+{\frac{1}{48} t_1^4};  (50,0)*+{+}; (58,0)*+{\cdots\cdots};
\endxy
$$

\subsection{Correlators}

The coefficients of $F$ gives us the correlators defined by:
\be
\corr{\tau_{a_1} \cdots \tau_{a_n}}_g
=  \frac{\pd^n}{\pd t_{a_1} \cdots \pd t_{a_n}} F_g |_{\bt = 0}.
\ee
These are Taylor series coefficients of  $F$,
in other words,
\ben
&& F_g = \sum_{m_0, \dots, m_n \geq 0} \corr{\tau_0^{m_0} \cdots \tau_n^{m_n} }_g \prod_{j=0}^n \frac{t_j}{m_j!}, \\
&& F = \sum_{g\geq 0} \lambda^{2g-2} F_g.
\een
The following are some examples:
\ben
&& \corr{\tau_0^2}_0 = 1,  \corr{\tau_1}_1 = \frac{1}{2},
\corr{\tau_0^2\tau_1}_0 = 1,  \corr{\tau_0\tau_2}_1 = \frac{1}{2},
\corr{\tau_1^2}_1 = \frac{1}{2}, \\
&& \corr{\tau_3}_2 = \frac{1}{8}, \corr{\tau_0^2\tau_1^2}_0 = 2,
\corr{\tau_0^3\tau_2}_0 = 1,
\corr{\tau_1^3}_1 = 1, \corr{\tau_0^2\tau_3}_1 = \frac{1}{2}, \\
&& \corr{\tau_0\tau_1\tau_2}_1=1,
\corr{\tau_2^2}_2= \frac{5}{12},
\corr{\tau_0\tau_4}_2 = \frac{1}{8},
\corr{\tau_1\tau_3}_2 = \frac{1}{4}, \\
&&  \corr{\tau_5}_3 = \frac{1}{48},
\corr{\tau_1^3\tau_0^2}_0 = 6,
\corr{\tau_1\tau_2\tau_0^3}_0 = 3,
\corr{\tau_3\tau_0^4}_0 = 1, \\
&& \corr{\tau_1^4}_1 = 3,
\corr{\tau_0^2\tau_1\tau_3}_1 = \frac{3}{2},
\corr{\tau_0 \tau_1^2\tau_2}_1 = 3,
\corr{\tau_0^3 \tau_4}_1 = \frac{1}{2}, \\
&& \corr{t_0^2t_2^2}_1 = 2,
\corr{\tau_2^2\tau_1}_2 = \frac{5}{4},
\corr{\tau_4\tau_0\tau_1}_2 = \frac{3}{8},
\corr{\tau_3\tau_1^2}_2 = \frac{3}{4}, \\
&& \corr{\tau_5\tau_0^2}_2 = \frac{1}{8},
\corr{\tau_0\tau_2\tau_3}_2 = \frac{2}{3},
 \corr{\tau_5\tau_1}_3 = \frac{1}{16},
\corr{t_3^2}_3 = \frac{1}{6},  \\
&& \corr{\tau_0\tau_6}_3 = \frac{1}{48},
\corr{t_4t_2}_3 =  \frac{7}{48},
\corr{\tau_7}_4 = \frac{1}{384}.
\een
By comparing with \eqref{eqn:Feynman-F},
we get

\begin{prop}
The correlators of topological 1D gravity can be given by Feynman sum as follows:
\be
\corr{\tau_0^{m_0} \cdots \tau_n^{m_n} }_g = \sum_\Gamma \frac{1}{|\Aut(\Gamma)|} \cdot \prod_{j=0}^n m_j!,
\ee
where the summation is taken over connected graphs with $m_j$ vertices of valence $j+1$,
$j=0, \dots, n$.
\end{prop}

One can rewrite \eqref{eqn:Feynman-F} in terms of correlators as follows: 
\be
\begin{split}
& \sum \corr{\tau_0^{m_0} \cdots \tau_n^{m_n} }_g  \prod_{j=0}^n \frac{t_j^{m_j}}{m_j!}
=  \sum_{\Gamma \in \cG^c} \frac{1}{|\Aut(\Gamma)|} \prod_{v\in V(\Gamma)} \lambda^{\val(v)-2} t_{\val(v)-1}.
\end{split}
\ee
This is the analogue of Kontsevich's main identity \cite{Kontsevich}.

\subsection{Some special cases of $Z$ and $F$} \label{sec:Special Cases Z}

In the above we have expressed the partition function and free energy of topological 1D gravity as
summation over graphs.
By $Z(t_{a_1}, \dots, t_{a_n})$ we mean
taking all $t_i$' to be zero except for $t_{a_1}, \dots, t_{a_n}$.
We have:
\be \label{eqn:Z(t0)}
Z(t_0) = \exp \biggl(\frac{t_0^2}{2\lambda^2} \biggr).
\ee
In fact,
\ben
Z(t_0)
& = & \frac{1}{\sqrt{2\pi}} \int_\bR dx \exp\biggl(-\half x^2 + \frac{t_0}{\lambda^2} x\biggr) \\
& = & \exp \biggl(\frac{t_0^2}{2\lambda^2} \biggr) \cdot
\frac{1}{\sqrt{2\pi}} \int_\bR dx \exp\biggl(-\half (x - \frac{t_0}{2\lambda^2})^2 \biggr) \\
& = & \exp \biggl(\frac{t_0^2}{2\lambda^2} \biggr) \cdot
\frac{1}{\sqrt{2\pi}} \int_\bR dx \exp\biggl(-\half x^2 \biggr) \\
& = &  \exp \biggl(\frac{t_0^2}{2\lambda^2} \biggr).
\een

Similarly, we also have
\be \label{eqn:Z(t1)}
Z(t_0, t_1) =  \exp \biggl(\frac{t_0^2}{2\lambda^2(1-t_1)} + \half \log \frac{1}{1-t_1} \biggr),
\ee
and in particular,
\be
Z(t_1) =  \exp \biggl( \half \log \frac{1}{1-t_1} \biggr),
\ee
by the following computations:
\ben
Z(t_0,t_1)
& = & \frac{1}{\sqrt{2\pi}} \int_\bR dx \exp\biggl(-\half x^2 + \frac{t_0}{\lambda} x + t_1 \frac{x^2}{2} \biggr) \\
& = & \frac{1}{(1-t_1)^{1/2}} \cdot \frac{1}{\sqrt{2\pi}} \int_\bR d ((1-t_1)^{1/2}x) \\
&& \cdot \exp\biggl(-\half ((1-t_1)^{1/2}x)^2 - \frac{t_0}{\lambda(1-t_1)^{1/2}} (1-t_1)^{1/2}x \biggr) \\
& = &  \exp \biggl(\frac{t_0^2}{2\lambda^2(1-t_1)} + \half \log \frac{1}{1-t_1} \biggr).
\een

When $t_n = \delta_{n,2}$ one gets the integral:
\be
Z(t_2) = \frac{1}{\sqrt{2\pi}}
\int_\bR dx \exp \biggl( -\half x^2 +   t_2\lambda \frac{x^3}{3!}  \biggr)
\ee
Using \eqref{eqn:Gaussian},
one gets:
\be
Z(t_2) = \sum_{n=0}^\infty \frac{t_2^{2n}}{3!^{2n}}\lambda^{2n} \frac{(6n-1)!!}{(2n)!}.
\ee
The first few terms are
\ben
Z(t_2) & = & 1+ \frac{5}{24}(t_2^2\lambda^2)+\frac{385}{1152}(t_2^2\lambda^2)^2+\frac{85085}{82944} (t_2^2\lambda^2)^3
+ \frac{37182145}{7962624}(t_2^2\lambda^2)^4 \\
& + & \frac{5391411025}{191102976}(t_2^2\lambda^2)^5+\frac{5849680962125}{27518828544}(t_2^2\lambda^2)^6 \\
& + & \frac{1267709431363375}{660451885056}(t_2^2\lambda^2)^7
+ \frac{2562040760785380875}{126806761930752}(t_2^2\lambda^2)^8+\cdots.
\een
After taking logarithm:
\ben
F(t_2) & = & \frac{5}{24}t_2^2\lambda^2+ \frac{5}{16} (t_2^2\lambda^2)^2
+ \frac{1105}{1152} (t_2^2\lambda^2)^3+ \frac{565}{128} (t_2^2\lambda^2)^4
+ \frac{82825}{3072} (t_2^2\lambda^2)^5 \\
& + & \frac{19675}{96} (t_2^2\lambda^2)^6
+ \frac{1282031525}{688128} (t_2^2\lambda^2)^7+\frac{80727925}{4096}(t_2^2\lambda^2)^8 + \cdots.
\een
In particular, $F_g(t_2)$ has the following form:
\be
F_g(t_2) = b_g t_2^{2g-2}
\ee
for some constant $b_g$ for $g> 1$.

To compute
\ben
Z(t_0,t_2) & = &
 \frac{1}{\sqrt{2\pi}} \int_\bR dx \exp\biggl(-\half x^2 + \frac{t_0}{\lambda} x + \lambda t_2 \frac{x^3}{6} \biggr),
\een
we make a change of variables $x= y +a$:
\ben
Z(t_0,t_2) & = &
\frac{1}{\sqrt{2\pi}} \int_\bR dx \exp\biggl(-\half (y+a)^2 + \frac{t_0}{\lambda} (y+a) + t_2 \lambda \frac{(y+a)^3}{6} \biggr) \\
& = &  \frac{1}{\sqrt{2\pi}} \int_\bR dx \exp\biggl(-\half (1-at_2\lambda) y^2
+ (-a+\frac{t_0}{\lambda} + \frac{t_2}{2}\lambda a^2) y \\
&& + \frac{t_2}{6}\lambda y^3  \biggr) \cdot
\exp\biggl(-\half a^2 + \frac{t_0}{\lambda} a +\lambda t_2 \frac{a^3}{6}\biggr).
\een
We take
\be
a = \frac{1-\sqrt{1-2t_0t_2}}{t_2\lambda}
\ee
such that
$$
-a+\frac{t_0}{\lambda} + \frac{t_2}{2}\lambda a^2 = 0,
$$
and so we can reduce to the Airy integral and get:
\ben
Z(t_0,t_2) & = &  \exp\biggl( \frac{1}{3t_2^2\lambda^2} \biggl( (1-2t_0t_2)^{3/2} - (1-3t_0t_2)\biggr)
+ \frac{1}{4} \log(1-2t_0t_2) \biggr) \\
&& \cdot \sum_{n=0}^\infty  \biggl(\frac{t_2 \lambda}{3!(1-2t_0t_2)^{3/4}} \biggr)^{2n} \frac{(6n-1)!!}{(2n)!}.
\een
After taking logarithm,
one gets:
\ben
F(t_0,t_2)
& = & \frac{t_0^2}{\lambda^2} \frac{(1-2t_0t_2)^{3/2}-(1-3t_0t_2)}{3(t_0t_2)^2}
+ \frac{1}{4} \log \frac{1}{1-2t_0t_2} \\
& + & \frac{5}{24}  \frac{t_2^2\lambda^2 }{(1-2t_0t_2)^{3/2}}
+ \frac{5}{16}  \frac{t_2^4\lambda^4}{(1-2t_0t_2)^3}
+ \frac{1105}{1152}  \frac{t_2^6\lambda^6}{(1-2t_0t_2)^{9/2}} \\
& + & \frac{565}{128}  \frac{t_2^8\lambda^8}{(1-2t_0t_2)^6}
+ \frac{82825}{3072} \frac{t_2^{10}\lambda^{10}}{(1-2t_0t_2)^{15/2}}
+ \cdots.
\een
In particular,
when $g > 1$,
$F_g(t_0,t_2)$ has the following form:
\be
F_g(t_0,t_2) = b_g \frac{t_2^{2g-2}}{(1-2t_0t_2)^{(3g-3)/2}}.
\ee

In the same fashion as in the case of $Z(t_2)$,
\be
Z(t_{2k-1} ) = \sum_{n \geq 0} \biggl(\frac{t_{2k-1}\lambda^{2k-2}}{(2k)!} \biggr)^n
\frac{(2nk-1)!!}{n!},
\ee
and
\be
Z(t_{2k}) = \sum_{n \geq 0} \biggl( \frac{t_{2k}\lambda^{2k-1}}{(2k+1)!} \biggr)^{2n}
\frac{(2n(2k+1)-1)!!}{(2n)!}.
\ee

\subsection{General explicit expression for $Z$}
\label{sec:General expression for Z}

One can generalize the formulas in the preceding subsection for $Z(t_n)$ as follows:

\begin{prop}
The partition function $Z$ has the following closed expression:
\be \label{eqn:Formula-for-Z}
Z = \sum_{n \geq 0} \sum_{\sum_{j=1}^k m_j j =2n}
\frac{(2n-1)!!}{\prod_{j=1}^k (j!)^{m_j} m_j!}
\lambda^{2n-2\sum_{j=1}^k m_j} \cdot \prod_{j=1}^k t_{j-1}^{m_j}.
\ee
\end{prop}

\begin{proof}
\ben
Z & = & \frac{1}{\sqrt{2\pi}}
\int_\bR dx \exp \biggl( -\half x^2\biggr)  \sum_{l=0}^\infty \frac{1}{l!} (\frac{t_{n-1}}{n!} \lambda^{n-2} x^n)^l \\
& = & \frac{1}{\sqrt{2\pi}}
\int_\bR dx \exp \biggl( -\half x^2\biggr)  \sum_{m_1, \dots, m_n=0}^\infty
\prod_{j=1}^n \frac{1}{m_j!} \biggl(\frac{t_{j-1}}{j!} \lambda^{j-2} \biggr)^{m_j} x^{\sum_{j=1}^n m_jj} \\
& = & \sum_{n \geq 0} \sum_{\sum_{j=1}^k m_j j =2n}
\frac{(2n-1)!!}{\prod_{j=1}^k (j!)^{m_j} m_j!}
\lambda^{2n-2\sum_{j=1}^k m_j} \cdot \prod_{j=1}^k t_{j-1}^{m_j}.
\een
\end{proof}

By comparing \eqref{eqn:Formula-for-Z} with \eqref{eqn:Feynman-Z},
one gets

\begin{cor}
The following Feynman sum has a close form expression:
\ben
&& \sum_{\Gamma \in \cG} \frac{1}{|\Aut(\Gamma)|} \prod_{v\in V(\Gamma)} \lambda^{\val(v)-2} t_{\val(v)-1} \\
& = & \sum_{n \geq 0} \sum_{\sum_{j=1}^k m_j j =2n}
\frac{(2n-1)!!}{\prod_{j=1}^k (j!)^{m_j} m_j!}
\lambda^{2n-2\sum_{j=1}^k m_j} \cdot \prod_{j=1}^k t_{j-1}^{m_j}.
\een
\end{cor}

\begin{prop}
The partition function $Z$ has the following closed expression:
\be \label{eqn:Formula-for-Z2}
\begin{split}
Z = & \frac{1}{\sqrt{1-t_1}} \sum_{n \geq 0} \sum_{\substack{\sum\limits_{j=1}^k m_j j =2n\\m_2 = 0}}
\frac{(2n-1)!!}{\prod\limits_{j=1}^k (j!)^{m_j} m_j!}
\lambda^{2n-2\sum\limits_{j=1}^k m_j} \\
& \cdot \prod_{j=1}^k \biggl(\frac{t_{j-1}}{(1-t_1)^{j/2}}\biggr)^{m_j}.
\end{split}
\ee
\end{prop}

\begin{proof}
\ben
Z & = & \frac{1}{\sqrt{2\pi}}
\int_\bR dx \exp \biggl( -\half (1-t_1) x^2\biggr)  \exp \sum_{n\geq 1, \neq 2} \frac{t_{n-1}}{n!} \lambda^{n-2} x^n  \\
& = & \frac{1}{\sqrt{2\pi(1-t_1)}}
\int_\bR dx \exp \biggl( -\half x^2\biggr) \exp \sum_{n\geq 1, \neq 2} \frac{t_{n-1}\lambda^{n-2} x^n}{n!(1-t_1)^{n/2}} \\
& = & \frac{1}{\sqrt{1-t_1}}\sum_{n \geq 0} \sum_{\substack{\sum_{j=1}^k m_j j =2n\\m_2 = 0}}
\frac{(2n-1)!!}{\prod_{j=1}^k (j!)^{m_j} m_j!}
\lambda^{2n-2\sum_{j=1}^k m_j}\\
&& \cdot \prod_{j=1}^k \biggl(\frac{t_{j-1}}{(1-t_1)^{j/2}}\biggr)^{m_j}.
\een
\end{proof}

The following are the first few terms:
\ben
Z & = & 1 + (\half \lambda^{-2} t_0^2 + \half t_1) + \biggl( \frac{1}{8}\lambda^{-4}t_0^4
+ \frac{3}{4}\lambda^{-2}t_0^2t_1 + \frac{3}{8} t_1^2+ \frac{1}{2}t_0t_2 + \frac{1}{8}\lambda^2t_3 \biggr) \\
& + & \biggl( \frac{1}{48}\lambda^{-6}t_0^6 + \frac{5}{16}\lambda^{-4}t_0^4t_1
+ \frac{15}{16}\lambda^{-2}t_0^2t_1^2 + \frac{5}{16} t_1^3
+ \frac{5}{12}\lambda^{-2}t_0^3t_2 \\
&+ & \frac{5}{4} t_0t_1t_2 + \frac{5}{24} \lambda^2t_2^2  + \frac{5}{16}t_0^2t_3
+ \frac{5}{16} \lambda^2 t_1t_3 + \frac{1}{8}\lambda^2t_0t_4 + \frac{1}{48} \lambda^4t_5 \biggr) + \cdots
\een

\ben
Z & = & \frac{1}{\sqrt{1-t_1}} \biggl( 1 + \half \lambda^{-2} \frac{t_0^2}{1-t_1}
+ \frac{1}{(1-t_1)^2} \biggl( \frac{1}{8}\lambda^{-4} t_0^4
+ \frac{1}{2}t_0t_2  + \frac{1}{8}\lambda^2t_3 \biggr) \\
& + & \frac{1}{(1-t_1)^3} \biggl( \frac{1}{48}\lambda^{-6}t_0^6
+ \frac{5}{12}\lambda^{-2}t_0^3t_2 \\
& + & \frac{5}{24} \lambda^2t_2^2  + \frac{5}{16}t_0^2t_3
+ \frac{1}{8}\lambda^2t_0t_4 + \frac{1}{48} \lambda^4t_5 \biggr) + \cdots
 \biggr)
\een

\ben
F & = & (\frac{1}{2}\lambda^{-2}t_0^2+\frac{1}{2}t_1)
+ (\frac{1}{2}t_0^2t_1\lambda^{-2} +\frac{1}{2}t_0t_2+\frac{1}{4}t_1^2 + \frac{1}{8} t_3\lambda^2) \\
& + & (\frac{1}{2}t_0^2t_1^2\lambda^{-2}+\frac{1}{6} t_0^3t_2\lambda^{-2}
+ \frac{1}{6}t_1^3 + \frac{1}{4}t_0^2t_3+t_0t_2t_1 \\
& + & \frac{5}{24}t_2^2\lambda^2
+ \frac{1}{8}t_0t_4\lambda^2+\frac{1}{4}t_1t_3\lambda^2 +\frac{1}{48}t_5\lambda^4) \\
& + & (\frac{1}{2}\lambda^{-2}t_1^3t_0^2 + \frac{1}{2}\lambda^{-2}t_1t_2t_0^3 + \frac{1}{24}\lambda^{-2}t_3t_0^4 \\
& + & \frac{1}{8} t_1^4  + \frac{3}{4} t_0^2t_1t_3 + \frac{3}{2} t_0 t_1^2t_2 + \frac{1}{12} t_0^3 t_4
+ \frac{1}{2} t_0^2t_2^2 \\
& + & \frac{5}{8} \lambda^2 t_2^2t_1 + \frac{3}{8} \lambda^2t_4t_0t_1 + \frac{3}{8} \lambda^2t_3t_1^2
+ \frac{1}{16} \lambda^2t_5t_0^2   + \frac{2}{3} \lambda^2t_3t_0t_2 \\
& + & \frac{1}{16}\lambda^4t_5t_1
+ \frac{1}{12} \lambda^4t_3^2 + \frac{1}{48}t_0\lambda^4t_6
+ \frac{7}{48} \lambda^4t_4t_2 + \frac{1}{384} t_7\lambda^6 ) + \cdots,
\een
The coefficients of $F$ give us correlators.
For example,
\ben
&& \corr{\tau_3}_2= \frac{1}{8}, \;\;\; \corr{\tau_2^2}_2 = \frac{5}{12}, \\
&& \corr{\tau_5}_3 = \frac{1}{48}, \;\;\;
\corr{\tau_3^2}_3 = \frac{1}{6}, \;\;\;
\corr{\tau_2\tau_4}_3 = \frac{7}{48}, \;\;\; \corr{\tau_7}_4 = \frac{1}{384}.
\een
We also have:
\ben
F & = & \half \log \frac{1}{1-t_1} +  \frac{1}{2}\lambda^{-2} \frac{t_0^2}{1-t_1}
+ \frac{1}{(1-t_1)^2} ( \frac{1}{2}t_0t_2 + \frac{1}{8} t_3\lambda^2) \\
& + & \frac{1}{(1-t_1)^3} (  \frac{1}{6} t_0^3t_2\lambda^{-2}
  + \frac{1}{4}t_0^2t_3 +  \frac{5}{24}t_2^2\lambda^2
+ \frac{1}{8}t_0t_4\lambda^2 +\frac{1}{48}t_5\lambda^4) \\
& + & \frac{1}{(1-t_1)^4} ( \frac{1}{24}\lambda^{-2}t_3t_0^4  + \frac{1}{12} t_0^3 t_4
+ \frac{1}{2} t_0^2t_2^2 + \frac{1}{16} \lambda^2t_5t_0^2 \\
& + & \frac{2}{3} \lambda^2t_3t_0t_2  + \frac{1}{12} \lambda^4t_3^2 + \frac{1}{48}t_0\lambda^4t_6
+ \frac{7}{48} \lambda^4t_4t_2 + \frac{1}{384} t_7\lambda^6 ) + \cdots.
\een

\subsection{The selection rule}

A nonzero term in $Z$ is of the form
\be
t_{a_1} \cdots t_{a_n} \lambda^{2g-2},
\ee
up to coefficients.
The numbers $a_1+1, \dots, a_n+1$ gives a partition $\mu$ of length $n$,
so by \eqref{eqn:Formula-for-Z},
\ben
2g-2 = \sum_{j=1}^n (a_j +1) - 2n,
\een
so one must have
\be \label{eqn:Selection}
a_1+ \cdots + a_n = 2g-2+n.
\ee
This is the {\em selection rule} for nonvanishing terms in $Z$.
After taking logarithm,
one gets the same rule for $F$.

As an application of the selection rule,
we have

\begin{prop}
The genus zero part $F_0$ of the free energy satisfies the following initial condition:
\be
F_0|_{t_0 = 0} = 0.
\ee
\end{prop}

\begin{proof}
Just take $g=0$ in \eqref{eqn:Selection}.
\end{proof}

For $g \geq 1$,
we do not have $F_g|_{t_0=0} = 0$.

Another application of the selection rule is the following

\begin{prop}
The   free energy  $F$ restricted to the $t_0$-line is given by:
\be
F(t_0) = \frac{\lambda^{-2}}{2} t_0^2.
\ee
\end{prop}

\begin{proof}
Take $a_1 = \cdots = a_n = 0$ in \eqref{eqn:Selection} to get:
$$2g + n = 2.$$
Then one has $g=0$ and $n=2$.
\end{proof}

\subsection{Partition function and free energy in I-coordinates}
\label{sec:Partition-in-I-Coord}

Recall by Theorem \ref{thm:T-infinity},
\be
\begin{split}
S(x) = & \sum_{k=0}^\infty  \frac{(-1)^k}{(k+1)!} (I_k+\delta_{k,1}) I_0^{k+1} \\
- & \half (1 - I_1) (x-I_0)^2 + \sum_{n=3}^\infty \frac{I_{n-1}}{n!} (x-I_0)^n,
\end{split}
\ee
therefore,
by \eqref{eqn:Gaussian-Translation},
\be
\begin{split}
Z = & \exp \biggl( \frac{1}{\lambda^2}\sum_{k=0}^\infty  \frac{(-1)^k}{(k+1)!} (I_k+\delta_{k,1}) I_0^{k+1} \biggr) \\
& \cdot \frac{1}{\sqrt{\pi}\lambda}
\int  dx \exp \frac{1}{\lambda^2} \biggl( -\half (1-I_1) x^2 + \sum_{n \geq 3} I_{n-1} \frac{x^n}{n!}  \biggr).
\end{split}
\ee
After a scaling of the variable $x$,
\be \label{eqn:Z-Integral-I}
\begin{split}
Z = & \exp \biggl( \frac{1}{\lambda^2}\sum_{k=0}^\infty  \frac{(-1)^k}{(k+1)!} (I_k+\delta_{k,1}) I_0^{k+1}
+ \half \log \frac{1}{1-I_1} \biggr) \\
& \cdot \frac{1}{\sqrt{\pi}\lambda}
\int  dx \exp \frac{1}{\lambda^2} \biggl( -\half x^2 + \sum_{n \geq 3} \frac{I_{n-1}}{(1-I_1)^{n/2}} \frac{x^n}{n!}  \biggr).
\end{split}
\ee
Therefore, by \eqref{eqn:Formula-for-Z2},
we obtain the following:

\begin{thm}
The partition function of the topological 1D gravity can be expressed in I-coordinates as follows:
\be \label{eqn:Formula-for-Z-in-I}
\begin{split}
Z &=   \exp \biggl( \frac{1}{\lambda^2}\sum_{k=0}^\infty  \frac{(-1)^k}{(k+1)!} (I_k+\delta_{k,1}) I_0^{k+1}
+ \half \log \frac{1}{1-I_1} \biggr) \\
\cdot & \sum_{n \geq 0} \sum_{\sum\limits_{j=3}^k m_j j =2n}
\frac{(2n-1)!!}{\prod\limits_{j=3}^k (j!)^{m_j} m_j!}
\lambda^{2n-2\sum\limits_{j=1}^k m_j}
\cdot \prod_{j=3}^k \biggl(\frac{I_{j-1}}{(1-I_1)^{j/2}}\biggr)^{m_j}.
\end{split}
\ee
\end{thm}

For example,
\ben
Z & = & \exp \biggl( \frac{1}{\lambda^2}\sum_{k=0}^\infty  \frac{(-1)^k}{(k+1)!} (I_k+\delta_{k,1}) I_0^{k+1}
+ \half \log \frac{1}{1-I_1} \biggr) \\
&& \cdot¡¡\biggl( 1
+ \frac{1}{8} \frac{I_3}{(1-I_1)^2} \lambda^2
+  \frac{1}{(1-I_1)^3} \biggl( \frac{5}{24} I_2^2 \lambda^2
+ \frac{1}{48} I_5 \lambda^4 \biggr) \\
&& + \frac{1}{(1-I_1)^4}  \biggl( \frac{35}{384} I_3^2 \lambda^4 + \frac{7}{48} I_2I_4 \lambda^4
+ \frac{1}{384} I_7  \lambda^6 \biggr)
+ \cdots
 \biggr)
\een
so after taking logarithm:
\ben
F & = & \frac{1}{\lambda^2}\sum_{k=0}^\infty  \frac{(-1)^k}{(k+1)!} (I_k+\delta_{k,1}) I_0^{k+1}
+ \half \log \frac{1}{1-I_1} \\
& + & \frac{1}{8}  \frac{I_3}{(1-I_1)^2} \lambda^2
+ \frac{1}{(1-I_1)^3} (    \frac{5}{24}I_2^2\lambda^2 +\frac{1}{48}I_5\lambda^4) \\
& + & \frac{1}{(1-I_1)^4} (  \frac{1}{12} I_3^2  \lambda^4
+ \frac{7}{48} I_4I_2 \lambda^4 + \frac{1}{384} I_7\lambda^6 ) + \cdots.
\een
In particular,
\bea
&& F_0 = \sum_{k=0}^\infty  \frac{(-1)^k}{(k+1)!} (I_k+\delta_{k,1}) I_0^{k+1}, \label{eqn:F0-In-I-coord} \\
&& F_1 = \half \log \frac{1}{1-I_1},  \label{eqn:F1-In-I-coord}  \\
&& F_2 = \frac{1}{8}  \frac{I_3}{(1-I_1)^2} + \frac{5}{24} \frac{I_2^2}{(1-I_1)^3}.  \label{eqn:F2-In-I-coord}
\eea
In general, for $g \geq 2$,
\be \label{eqn:Fg-in-I}
F_g = \sum_{\sum_{j \geq 3} m_j(j-2) = 2g-2} \corr{\prod_{j\geq 3} \tau_{j-1}^{m_j}}_g
\cdot \prod_{j \geq 3} \frac{1}{m_j!} \biggl( \frac{I_{j-1}}{(1-I_1)^{j/2}} \biggr)^{m_j}.
\ee
We will rederive such formulas in later Sections by different methods.

\subsection{Feynman rules for free energy in $I$-coordinates}

By \eqref{eqn:Z-Integral-I} one gets the following Feynman rules for $F_g$ ($g \geq 2$):
\be
F_g \label{eqn:Feynman-Fg-in-I}
= \sum_{\Gamma \in \cG^c_g} \frac{1}{|\Aut(\Gamma)|}  \prod_{v\in V(\Gamma)}  I_{\val(v)-1}
\cdot \prod_{e\in E(\Gamma)} \frac{1}{1-I_1},
\ee
where the summation is taken over the set of $g$-loop connected graphs
whose vertices all have valences $\geq 3$.
For example,
in the case of $F_2$,
we have the following three diagrams:
$$
\xy
(0,0)*\xycircle(5,2.5){}; (-5, 0)*+{\bullet};   (5,0)*+{\bullet}; (-5,0); (5,0), **@{-};
(-1,-8)*+{\frac{1}{12} \frac{I_2^2}{(1-I_1)^3}};
(17.5,0)*\xycircle(2.5,2.5){};
(20,0); (25,0), **@{-};  (20, 0)*+{\bullet};
(23,-8)*+{\frac{1}{8} \frac{I_2^2}{(1-I_1)^3}};
(27.5,0)*\xycircle(2.5,2.5){};
(25, 0)*+{\bullet};
(42.5,0)*\xycircle(2.5,2.5){}; (47.5,0)*\xycircle(2.5,2.5){};  (45, 0)*+{\bullet};
(45,-8)*+{\frac{1}{8} \frac{I_3}{(1-I_1)^2}};
\endxy
$$

\subsection{Free energy in derivatives of $I_0$ and the corresponding Feynman rules}

We combine the identities \eqref{eqn:Fg-in-I} with \eqref{eqn:In-Pd-I0} to get for $g \geq 2$,
\be  
\begin{split}
& F_g = \sum_{\sum_{j \geq 3} m_j(j-2) = 2g-2} 
\corr{\prod_{j\geq 3} \tau_{j-1}^{m_j}}_g \\
\cdot & \prod_{n \geq 3} \frac{1}{m_n!} \biggl(-\sum_{\sum\limits_{j \geq 1} j m_j = n-1}
\frac{ (\sum_j (j+1)n_j)!}{\prod\limits_j ((j+1)!)^{n_j}n_j!}
\cdot  \frac{\prod_j \biggl(-\frac{\pd^{j+1} I_0}{\pd t_0^{j+1}}\biggr)^{n_j}}
{\biggl(\frac{\pd I_0}{\pd t_0}\biggr)^{\sum\limits_j (j/2+1)n_j}} \biggr)^{m_n}.
\end{split}
\ee
This expresses $F_g$ in terms of derivatives of $I_0$.
We will later show that 
$$I_0 = \frac{\pd F_0}{\pd t_0},$$
so we have expressed $F_g$ in terms of $\frac{\pd^2F_0}{\pd t_0^2}$ and its derivative.
This is analogous to similar results in 2D topological gravity where 
$F_g$ is expressed   in terms of $\frac{\pd^3F_0}{\pd t_0^3}$ and its derivative. 
One can also obtain such expressions by combining \eqref{eqn:Feynman-Fg-in-I} with \eqref{eqn:In-Pd-I0}.
It is an interesting problem to formulate the Feynman rules for such expressions.
The following are some examples:
\ben
&& F_2 = \frac{1}{8} \frac{\frac{\pd^3I_0}{\pd t_0^3}}{\biggl(\frac{\pd I_0}{\pd t_0}\biggr)^2}
- \frac{1}{12} \frac{\biggl(\frac{\pd^2I_0}{\pd t_0^2}\biggr)^2}{\biggl(\frac{\pd I_0}{\pd t_0}\biggr)^3}.
\een 
It is interesting to interpret it as a sum over two-loop diagrams:
$$
\xy
(0,0)*\xycircle(5,2.5){}; (-5, 0)*+{\bullet};   (5,0)*+{\bullet}; (-5,0); (5,0), **@{-};
(-1,-14)*+{- \frac{1}{12} \frac{(\pd^2_{t_0}I_0)^2}{(\pd_{t_0} I_0)^3}};
(17.5,0)*\xycircle(2.5,2.5){};
(20,0); (25,0), **@{-};  (20, 0)*+{\bullet};
(23,-14)*+{0 \cdot \frac{(\pd_{t_0}^2I_0)^2}{(\pd_{t_0} I_0)^3}};
(27.5,0)*\xycircle(2.5,2.5){};
(25, 0)*+{\bullet};
(42.5,0)*\xycircle(2.5,2.5){}; (47.5,0)*\xycircle(2.5,2.5){};  (45, 0)*+{\bullet};
(45,-14)*+{\frac{1}{8} \frac{\pd_{t_0}^3I_0}{ (\pd_{t_0} I_0)^2}};
\endxy
$$
The factor $\frac{1}{8}$ is just $\frac{1}{|\Aut(\Gamma)|}$ for the corresponding diagram $\Gamma$,
the factor $-\frac{1}{12}$ is $(-1)^{|V(\Gamma)-1}(|V(\Gamma)|)! \cdot \frac{1}{|\Aut(\Gamma)|}$. 
The factor is $0$ for the middle diagram seems to indicate only 1PI diagrams have contributions. 
Similarly,
\ben
F_3 & = &  \frac{1}{48} \frac{I_5}{(1-I_1)^3} + \frac{1}{12}  \frac{I_3^2}{(1-I_1)^4}
+ \frac{7}{48} \frac{I_4I_2}{(1-I_1)^4}
+ \frac{25}{48} \frac{I_2^2I_3}{(1-I_1)^5} + \frac{5}{16}\frac{I_2^4}{(1-I_1)^6} \\
& = & \frac{(\pd_{t_0}^5I_0)}{48(\pd_{t_0}I_0)^3}
-\frac{(\pd_{t_0}^3I_0)^2}{8(\pd_{t_0}I_0)^4} -\frac{(\pd_{t_0}^2I_0)(\pd_{t_0}^4I_0)}{6(\pd_{t_0}I_0)^4} 
+ \frac{3(\pd_{t_0}^2I_0)^2(\pd_{t_0}^3I_0)}{4(\pd_{t_0}I_0)^5}  
-\frac{1}{2} \frac{(\pd_{t_0} I_0)^4}{(\pd_{t_0}I_0)^6}.
\een
We have checked in this case only 1PI programs contributes when one interpret the formula in 
derivatives of $I_0$.

\section{Flow Equations, Polymer Equation, and Their Applications}

\label{sec:Flow-Polymer}

In this Section we present some applications of
the flow equation and polymer equation  \cite{Nishigaki-Yoneya1} of topological 1D gravity.
These include a proof of Theorem \ref{thm:Rooted trees},
some explicit expressions of $F_0$ and its derivatives in $t_0$ in $t$-coordinates and
a rederivation of the formula for $F_0$ in $I$-coordinates.

\subsection{Flow equations}

\begin{prop} (\cite{Nishigaki-Yoneya1} )
For each $n \geq 0$,
the following equation is satisfied by $Z$:
\be
\frac{\pd Z}{\pd t_n} = \frac{\lambda^{2n}}{(n+1)!} \frac{\pd^{n+1} Z}{\pd t_0^{n+1}}
\ee
\end{prop}

\begin{proof}
\ben
\frac{\pd Z}{\pd t_n} & =  &\frac{1}{\sqrt{2\pi}\lambda}
\int_\bR dx \frac{x^{n+1}}{(n+1)!\lambda^2} \exp \frac{1}{\lambda^2}
\biggl( -\half x^2 + \sum_{n \geq 1} t_{n-1} \frac{x^n}{n!}  \biggr) \\
& = & \frac{\lambda^{2n}}{(n+1)!} \frac{\pd^{n+1}}{\pd t_0^{n+1}}
\frac{1}{\sqrt{2\pi}\lambda}
\int_\bR dx  \exp \frac{1}{\lambda^2} \biggl( -\half x^2 + \sum_{n \geq 1} t_{n-1} \frac{x^n}{n!}  \biggr) \\
& = & \frac{\lambda^{2n}}{(n+1)!} \frac{\pd^{n+1} Z}{\pd t_0^{n+1}}
\een
\end{proof}

\subsection{Solution of the flow equations}

It is easy to see that the flow equations have the following solution:
\be \label{eqn:Z-Operator}
Z = \exp \biggl( \sum_{n=1}^\infty \frac{\lambda^{2n}}{(n+1)!} t_n \frac{\pd^{n+1}}{\pd t_0^{n+1}} \biggr)
\exp \frac{t_0^2}{2\lambda^2}.
\ee
This is a sort of free field realization of the topological 1D gravity.
One can write $Z$ as a summation over Feynman diagrams as follows.
First expand the two exponentials:
\be
Z = \sum_{m_1, \dots, m_n \geq 0} \prod_{j=1}^n\frac{1}{m_j!}
\biggl(\frac{\lambda^{2j}}{(j+1)!} t_j \frac{\pd^{j+1}}{\pd t_0^{j+1}} \biggr)^{m_j}
\sum_{m \geq 0}  \frac{1}{m!} \biggl(\frac{t_0^2}{2\lambda^2}\biggr)^m.
\ee
We understand each copy of $\frac{\pd^{j+1}}{\pd t_0^{j+1}}$ as a
vertex marked by $\bullet$ with $j+1$ edges,
and each copy of $t_0^2$ as a vertex with $\ast$ with two edges.
All the edges from a vertex with $\bullet$ must be connected with
a vertex connected a vertex marked with $\ast$,
but not vice versa.
If an edge from a vertex marked with $\ast$ is not connected to another vertex,
mark its open end by $\bullet$.
Then one gets some graphs with vertices marked by $\bullet$  or $\ast$.
Then
\be
Z = \sum_\Gamma \frac{1}{|\Aut(\Gamma)|} w_\Gamma,
\ee
where the Fenyman rule is
where the weight of $\Gamma$ is given by
\be
w_\Gamma = \prod_{v\in V(\Gamma)} w_v \cdot \prod_{e\in E(\Gamma)} w_e,
\ee
with $w_e$ and $w_v$ given by the following Feynman rule:
\bea
&& w(e) = 1, \\
&& w(v) = \begin{cases}
\lambda^{\val(v)}t_{\val(v)-1}, & \text{if $v$ is a vertex marked by $\bullet$}, \\
1, & \text{if $v$ is a vertex marked by $\ast$}.
\end{cases}
\eea
Now one can simply ignore the vertices marked by $\ast$.
This will not change $\Aut(\Gamma)$ or $\Gamma$,
so one gets exactly the Feynman diagrams and Feynman rules as in \eqref{eqn:Feynman}.

As an example,
we have
\ben
Z(t_0,t_2) & = & \sum_{m_2 \geq 0} \frac{1}{m_2!} \biggl( \frac{\lambda^4}{6}t_2
\frac{\pd^3}{\pd t_0^3} \biggr)^{m_2}
\sum_{m \geq 0} \frac{1}{m!} \biggl(\frac{t_0^2}{2\lambda^2} \biggr)^m \\
& = & \sum_{\substack{m_2, m \geq 0\\2m \geq 3m_2}}
\frac{1}{m_2!} \biggl( \frac{\lambda^4}{6}t_2 \biggr)^{m_2}
\frac{1}{m!} \biggl(\frac{1}{2\lambda^2} \biggr)^m \prod_{j=0}^{3m_2-1}
(2m-j) \cdot t_0^{2m-3m_2}.
\een
In particular,
\ben
Z(t_2)
& = & \sum_{n\geq 0}
\frac{1}{(2n)!} \biggl( \frac{\lambda^4}{6}t_2 \biggr)^{2n}
\frac{1}{(3n)!} \biggl(\frac{1}{2\lambda^2} \biggr)^{3n}  (6n)! \\
& = & \sum_{n \geq 0} \frac{(6n)!\lambda^{2n}}{288^n(2n)!(3n)!} t_2^{2n}.
\een

The formula \eqref{eqn:Z-Operator} is very elegant,
but it does not give us any information about the analytic properties of
the free energy.

\subsection{The polymer equation}

\begin{thm} (\cite{Nishigaki-Yoneya1} )
The partition function of the topological 1D gravity satisfies the following equation:
\be \label{eqn:Polymer}
\sum_{n \geq 0} \frac{t_n-\delta_{n,1} }{n!} \lambda^{2n}  \frac{\pd^n}{\pd t_0^n} Z = 0.
\ee
\end{thm}

\begin{proof}
\ben
&& \sum_{n \geq 0} \frac{t_n-\delta_{n,1} }{n!} \lambda^{2n} \frac{\pd^n}{\pd t_0^n}  Z \\
& = & \frac{1}{\sqrt{2\pi}\lambda}
\int_\bR dx \cdot (-x + \sum_{n \geq 0} t_n \frac{x^n}{n!})
 \exp \frac{1}{\lambda^2} \biggl( -\half x^2 + \sum_{n \geq 1} t_{n-1} \frac{x^n}{n!}  \biggr) \\
 & = & \frac{1}{\sqrt{2\pi}\lambda}
\int_\bR dx \cdot \frac{d}{dx}
 \exp \frac{1}{\lambda^2} \biggl( -\half x^2 + \sum_{n \geq 1} t_{n-1} \frac{x^n}{n!}  \biggr) \\
 & = & 0.
\een
\end{proof}

In  \cite{Nishigaki-Yoneya1, Nishigaki-Yoneya2},
\eqref{eqn:Polymer} has been called the {\em polymer equation} or the {\em chain equation}.

\subsection{Checking the  polymer equation}

Let $t_n = \delta_{n,0}$,
then \eqref{eqn:Polymer} becomes:
\be
\lambda^2 \frac{\pd}{\pd t_0} Z(t_0) = t_0 Z(t_0).
\ee
This matches with  \eqref{eqn:Z(t0)}.

Let $t_n = 0$ for $n > 1$, then \eqref{eqn:Polymer} becomes:
\be
t_0 Z(t_0, t_1) + (t_1-1) \lambda^2 \frac{\pd }{\pd t_0} Z(t_0, t_1) = 0.
\ee
It can be rewritten as:
\be
\frac{\pd}{\pd t_0} F(t_0, t_1) = \frac{t_0}{(1-t_1)\lambda^2}.
\ee
This completely determines $\frac{\pd}{\pd t_0}F(t_0,t_1)$.
After integration:
\be
F(t_0,t_1) = \frac{1}{2} \frac{t_0^2}{(1-t_1)\lambda^2} + F(t_1).
\ee
So the  polymer equation determines $F(t_0,t_1)$ up to the initial value $F(t_1)$.

Let $t_n = 0$ for $n=1$ or $n > 2$, then \eqref{eqn:Polymer} becomes:
\be
t_0 Z(t_0, t_2) - \lambda^2 \frac{\pd }{\pd t_0} Z(t_0, t_1)
+ \lambda^4 \frac{t_2}{2} \frac{\pd^2}{\pd t_0^2} Z(t_0,t_2)  = 0.
\ee
It can be rewritten as:
\be
\frac{\pd}{\pd t_0} F(t_0, t_2) = \frac{t_0}{\lambda^2}
+ \frac{\lambda^2}{2}t_2((\pd_{t_0}F(t_0,t_2))^2+ \pd_{t_0}^2F(t_0,t_2)).
\ee
After writing $F(t_0,t_2) = \sum_{g\geq 0} \lambda^{2g-2} F_g(t_0, t_2)$,
one gets
\ben
&& \pd_{t_0} F_0(t_0,t_2) = t_0 + \frac{t_2}{2} (\pd_{t_0}F(t_0,t_2))^2, \\
&& \pd_{t_0} F_g(t_0,t_2) = \frac{t_2}{2} \biggl( \sum_{h=0}^g
\pd_{t_0}F_h(t_0,t_2) \cdot \pd_{t_0}F_{g-h}(t_0,t_2) + \pd_{t_0}^2F_{g-1}(t_0,t_2)\biggr),
\een
for $g\geq 1$.
From the first equation one can get
\ben
\pd_{t_0} F_0(t_0,t_2) = \frac{1-(1-2t_0t_2)^{1/2}}{t_2},
\een
and one can rewrite the second equation as
\ben
&& \pd_{t_0} F_g(t_0,t_2) \\
& = & \frac{t_2}{2(1-t_2\pd_{t_0} F_0(t_0,t_2))}
\biggl( \sum_{h=1}^{g-1}
\pd_{t_0}F_h(t_0,t_2) \cdot \pd_{t_0}F_{g-h}(t_0,t_2)
+ \pd_{t_0}^2F_{g-1}(t_0,t_2)\biggr) \\
& = & \frac{t_2}{2(1-2t_0t_2)^{1/2}}
\biggl( \sum_{h=1}^{g-1}
\pd_{t_0}F_h(t_0,t_2) \cdot \pd_{t_0}F_{g-h}(t_0,t_2) + \pd_{t_0}^2F_{g-1}(t_0,t_2)\biggr)
\een
so it can be used to find $\pd_{t_0} F_g(t_0,t_2)$ for all $g \geq 1$ recursively.
For example,
\ben
\pd_{t_0} F_1(t_0,t_2)
& = & \frac{t_2}{2(1-2t_0t_2)^{1/2}} \pd_{t_0}^2 F_0(t_0,t_2)
= \frac{t_2}{2(1-2t_0t_2)}, \\
\pd_{t_0} F_2(t_0,t_2)
& = & \frac{t_2}{2(1-2t_0t_2)^{1/2}} \biggl(( \pd_{t_0}F_1(t_0,t_2) )^2
+ \pd_{t_0}^2F_1(t_0,t_2) \biggr) \\
& = & \frac{5t_2^3}{8(1-2t_0t_2)^{5/2}}.
\een
For $g > 1$, write
\be
\frac{\pd F_g}{\pd t_0}(t_0, t_2) = a_g \frac{t_2^{2g-1}}{(1-2t_0t_2)^{(3g-1)/2}},
\ee
then the above recursion relation yields:
\be
a_g = \half \biggl( \sum_{h=1}^{g-1} a_h a_{g-h}
+ (3g-4) a_{g-1} \biggr), \;\; g \geq 2.
\ee
Define the generating series:
\be
A(t) = \sum_{g \geq 1} a_g t^g
\ee
Then  the recursion relation is equivalent to
the following differential equation:
\be
\frac{3t^2}{2} A'(t)+ \frac{1}{2} A(t)^2-(1+\frac{t}{2})A(t) + \frac{t}{2} = 0.
\ee

Again the polymer equation only determines $\frac{\pd F}{\pd t_0}(t_0,t_2)$.
To determine $F(t_0, t_2)$ we need extra information about the initial value $F(t_2)$.
We have already shown that
\be
F_g(t_0,t_2) = b_g \frac{t_2^{2g-2}}{(1-2t_0t_2)^{(3g-3)/2}}.
\ee
for some constant $b_g$ when $g > 1$,
\be
F_g(t_0, t_2) = \frac{a_g}{3g-3} \frac{t_2^{2g-2}}{(1-2t_0t_2)^{(3g-3)/2}}.
\ee
These match with the formula for $F(t_0,t_2)$ in \S \ref{sec:Special Cases Z}.

\subsection{Comparison with the KdV}

Note
\be
\frac{\pd^n}{\pd t_0^n}Z
= \lambda^{2n} \biggl(\frac{\pd}{\pd t_0} + \frac{\pd F}{\pd t_0} \biggr)^n 1 \cdot Z.
\ee
Therefore,
one can rewrite the universal polymer equation as follows:
\be  \label{eqn:Polymer2}
\sum_{n \geq 0} \frac{t_n-\delta_{n,1} }{n!} \lambda^{2n} \biggl(\frac{\pd}{\pd t_0}
+ \frac{\pd F}{\pd t_0}\biggr)^n 1 = 0.
\ee
Similarly,
the flow equation can be written as:
\be \label{eqn:Flow for F}
\frac{\pd F}{\pd t_n}
= \frac{\lambda^{2n}}{(n+1)!} \biggl(\frac{\pd}{\pd t_0} + \frac{\pd F}{\pd t_0}\biggr)^{n+1} 1.
\ee
By taking $\frac{\pd}{\pd t_0}$ on both sides of \eqref{eqn:Polymer2}:
\be
1+ \sum_{n \geq 1} \frac{t_n-\delta_{n,1} }{n!}\lambda^{2n}  \frac{\pd}{\pd t_0}
\biggl(\frac{\pd}{\pd t_0} + \frac{\pd F}{\pd t_0}\biggr)^n 1 = 0.
\ee
This can be rewritten as:
\be
\frac{\pd }{\pd t_n} \frac{\pd F}{\pd t_0}= \frac{\lambda^{2n}}{(n+1)!} \frac{\pd}{\pd t_0}
\biggl(\frac{\pd}{\pd t_0} + \frac{\pd F}{\pd t_0}\biggr)^{n+1} 1.
\ee

By induction one can show that
\be
\frac{\pd}{\pd t_0} \biggl(\frac{\pd}{\pd t_0} + \frac{\pd F}{\pd t_0}\biggr)^{n} 1
= \biggl[\biggl(\frac{\pd}{\pd t_0} + \frac{\pd F}{\pd t_0}\biggr)^{n}, \frac{\pd F}{\pd t_0}\biggr] 1.
\ee
Therefore,
one has
\be
1+ \sum_{n \geq 1} \frac{t_n-\delta_{n,1} }{n!} \lambda^{2n}
\biggl[ \biggl(\frac{\pd}{\pd t_0} + \frac{\pd F}{\pd t_0} \biggr)^n, \frac{\pd F}{\pd t_0} \biggr] 1 = 0,
\ee
and
\be
\frac{\pd}{\pd t_n} \frac{\pd F}{\pd t_0}= \frac{\lambda^{2n}}{(n+1)!}
\biggl[\biggl(\frac{\pd}{\pd t_0} + \frac{\pd F}{\pd t_0}\biggr)^{n+1}, \frac{\pd F}{\pd t_0} \biggr] 1.
\ee
These formulas were derived in \cite{Nishigaki-Yoneya2}.

\subsection{Explicit formula for $\frac{\pd F_0}{\pd t_0}$}

By \eqref{eqn:Polymer2},
we have
\be \label{eqn:Polymer-Genus-0}
\frac{\pd F_0}{\pd t_0} = \sum_{n \geq 0} \frac{t_n}{n!} \biggl(\frac{\pd F_0}{\pd t_0}\biggr)^n.
\ee
This is exactly the equation \eqref{eqn:Critical} satisfied by $x_\infty = I_0$.
So we have

\begin{thm}
The genus zero part $F_0$ of the free energy $F$ satisfies the following differential equation:
\be \label{eqn:Diff-F0-t0}
\frac{\pd F_0}{\pd t_0} = I_0.
\ee
\end{thm}

As a corollary,
we have:

\begin{thm}
The derivative $\frac{\pd F_0}{\pd t_0}$ has the following explicit expression in $t$-coordinates:
\be \label{eqn:Pd F0}
\frac{\pd F_0}{\pd t_0} = I_0 = \sum_{k=1}^\infty \frac{1}{k}
\sum_{p_1 + \cdots + p_k = k-1} \frac{t_{p_1}}{p_1!} \cdots
\frac{t_{p_k}}{p_k!}.
\ee
\end{thm}

\begin{proof}
This  just follows from the formula \eqref{eqn:Xinfinity} for $I_0$.
\end{proof}

\begin{cor}
The derivative $\frac{\pd F_0}{\pd t_0}$ has the following explicit formula:
\be \label{eqn:Pd F0-1-t1}
\frac{\pd F_0}{\pd t_0} = \sum_{k=1}^\infty \frac{1}{k(1-t_1)^k}
\sum_{\substack{p_1 + \cdots + p_k = k-1\\p_j \neq 1,\;\; j=1,\dots, k}}
\frac{t_{p_1}}{p_1!} \cdots
\frac{t_{p_k}}{p_k!}.
\ee
\end{cor}

\begin{proof}
Rewrite the polymer equation in genus zero \eqref{eqn:Polymer-Genus-0} as
\be
\frac{\pd F_0}{\pd t_0}  = \frac{t_0}{1-t_1}
+ \sum_{n \geq 2} \frac{t_2}{1-t_1} \frac{(\frac{\pd F_0}{\pd t_0})^n}{n!}.
\ee
This is just \eqref{eqn:Polymer-Genus-0} with $t_i$ replaced by $\tilde{t_i}$,
where
\be
\tilde{t_i}
= \begin{cases}
\frac{t_i}{1-t_1}, & i \neq 1, \\
0, & i = 1.
\end{cases}
\ee
Therefore,
\be
\frac{\pd F_0}{\pd t_0}  = \sum_{k=1}^\infty \frac{1}{k(1-t_1)^k}
\sum_{\substack{p_1 + \cdots + p_k = k-1\\p_j \neq 1,\;\; j=1,\dots, k}}
\frac{t_{p_1}}{p_1!} \cdots
\frac{t_{p_k}}{p_k!}.
\ee
One can also get this directly from \eqref{eqn:Pd F0} as follows:
\ben
\frac{\pd F_0}{\pd t_0}   & = & \sum_{k=1}^\infty \frac{1}{k}
\sum_{m=0}^{k-1} \binom{k}{m} t_1^m
\sum_{\substack{p_1 + \cdots + p_{k-m} = k-m-1\\ p_1, \dots, p_{k-m} \neq 1} }
\frac{t_{p_1}}{p_1!} \cdots
\frac{t_{p_{k-m}}}{p_{k-m}!} \\
& = & \sum_{k=1}^\infty \sum_{m=0}^\infty \frac{1}{k+m}\binom{k+m}{m} t_1^m
\sum_{\substack{p_1 + \cdots + p_k = k-1\\ p_1, \dots, p_k \neq 1} }
\frac{t_{p_1}}{p_1!} \cdots
\frac{t_{p_k}}{p_k!} \\
& = & \sum_{k=1}^\infty \frac{1}{k(1-t_1)^k}
\sum_{\substack{p_1 + \cdots + p_k = k-1\\p_j \neq 1,\;\; j=1,\dots, k}}
\frac{t_{p_1}}{p_1!} \cdots
\frac{t_{p_k}}{p_k!}.
\een
\end{proof}

\subsection{Expressing $F_0$ in $I$-coordinates}

An application of \eqref{eqn:Diff-F0-t0} is that we can have another derivation of the expression of $F_0$
in terms of $I$-coordinates.
Recall by \eqref{eqn:diff-tk},
we have
\ben
\frac{\pd}{\pd t_0} = \frac{1}{1-I_1} \frac{\pd}{\pd I_0}
+ \sum_{l \geq 1} \frac{I_{l+1}}{1-I_1} \frac{\pd}{\pd I_l},
\een
and so by \eqref{eqn:Diff-F0-t0} we have:
\be
\frac{\pd}{\pd I_0} F_0 + \sum_{l \geq 1} I_{l+1} \frac{\pd}{\pd I_l} F_0  = (1-I_1) I_0.
\ee
Writing
\be
F_0 = \sum_{j=2}^\infty b_j(I_1, I_2, \dots) \frac{I_0^j}{j!},
\ee
one gets the following recursion relations:
\ben
&& b_2 = (1-I_1), \\
&& b_j = - \sum_{l \geq 1} I_{l+1} \frac{\pd}{\pd I_l} b_{j-1}, \;\; j \geq 3.
\een
We find the following solution:
\be
F_0 = \frac{1}{2!} (1-I_1)I_0^2
+ \sum_{j=3}^\infty (-1)^j I_{j-1} \frac{I_0^j}{j!}.
\ee
This is just \eqref{eqn:F0-In-I-coord} derived in \S \ref{sec:Partition-in-I-Coord}.

\subsection{Proof of Theorem \ref{thm:Rooted trees}}

Another application of \eqref{eqn:F0-In-I-coord} is that we can now have a proof of Theorem \ref{thm:Rooted trees}.
Recall formula \eqref{eqn:F0=Trees} expresses $F_0$ as a summation over trees:
\ben
F_0 = \sum_{\text{$\Gamma$ is a tree}} \frac{1}{|\Aut(\Gamma)|} \prod_{v\in V(\Gamma)} t_{val(v)-1}.
\een
Taking $\frac{\pd}{\pd t_0}$ on both sides proves Theorem \ref{thm:Rooted trees}.

\subsection{Explicit expression of $F_0$ in $t$-coordinates}

\begin{thm}
The following formulas hold:
\bea
F_0 & = & \sum_{k=1}^\infty \frac{1}{k(k+1)}
\sum_{ p_1 + \cdots + p_{k+1} = k-1 } \frac{t_{p_1}}{p_1!} \cdots
\frac{t_{p_{k+1}}}{p_{k+1}!} \\
& = & \sum_{k=1}^\infty \frac{1}{k(k+1)(1-t_1)^k}
\sum_{ \substack{p_1 + \cdots + p_{k+1} = k-1 \\ p_1, \dots, p_{k+1} \neq 1}}
\frac{t_{p_1}}{p_1!} \cdots
\frac{t_{p_{k+1}}}{p_{k+1}!}.
\eea
\end{thm}

\begin{proof}
We first rewrite \eqref{eqn:Pd F0} as follows:
\be
\frac{\pd F_0}{\pd t_0} = \sum_{k=1}^\infty \frac{1}{k} \sum_{m=1}^k \binom{k}{m} t_0^m
\sum_{\substack{p_1 + \cdots + p_{k-m} = k-1 \\p_1, \dots, p_{k-m} > 0}}
\frac{t_{p_1}}{p_1!} \cdots
\frac{t_{p_{k-m}}}{p_{k-m}!}.
\ee
Integrating once,
\ben
F_0
& = & \sum_{k=1}^\infty \frac{1}{k} \sum_{m=1}^k \binom{k}{m} \frac{t_0^{m+1}}{m+1}
\sum_{\substack{p_1 + \cdots + p_{k-m} = k-1 \\p_1, \dots, p_{k-m} > 0}}
\frac{t_{p_1}}{p_1!} \cdots
\frac{t_{p_{k-m}}}{p_{k-m}!} \\
& = & \sum_{k=1}^\infty \frac{1}{k(k+1)} \sum_{m=1}^k \binom{k+1}{m+1} t_0^{m+1}
\sum_{\substack{p_1 + \cdots + p_{k-m} = k-1 \\p_1, \dots, p_{k-m} > 0}}
\frac{t_{p_1}}{p_1!} \cdots
\frac{t_{p_{k-m}}}{p_{k-m}!} \\
& = & \sum_{k=1}^\infty \frac{1}{k(k+1)}
\sum_{ p_1 + \cdots + p_{k+1} = k-1 } \frac{t_{p_1}}{p_1!} \cdots
\frac{t_{p_{k+1}}}{p_{k+1}!}.
\een
Similarly,
rewrite \eqref{eqn:Pd F0-1-t1} as follows:
\be
\frac{\pd F_0}{\pd t_0} = \sum_{k=1}^\infty \frac{1}{k(1-t_1)^k}
\sum_{m=1}^{k} \binom{k}{m} t_0^m
\sum_{\substack{p_1 + \cdots + p_{k-m} = k-1\\p_j \neq 0, 1,\;\; j=1,\dots, k-m}}
\frac{t_{p_1}}{p_1!} \cdots
\frac{t_{p_k}}{p_k!}.
\ee
Integrating once:
\ben
 F_0
& = & \sum_{k=1}^\infty \frac{1}{k(1-t_1)^k} \sum_{m=1}^{k} \binom{k}{m} \frac{t_0^{m+1}}{m+1}
\sum_{\substack{p_1 + \cdots + p_{k-m} = k-1\\p_j \neq 0, 1,\;\; j=1,\dots, k-m}}
\frac{t_{p_1}}{p_1!} \cdots
\frac{t_{p_k}}{p_k!} \\
& = & \sum_{k=1}^\infty \frac{1}{k(k+1)(1-t_1)^k}
\sum_{ \substack{p_1 + \cdots + p_{k+1} = k-1 \\ p_1, \dots, p_{k+1} \neq 1}}
\frac{t_{p_1}}{p_1!} \cdots
\frac{t_{p_{k+1}}}{p_{k+1}!}.
\een
\end{proof}

\subsection{Explicit expression of higher derivatives of $F_0$ in $t$-coordinates}
By \eqref{eqn:Diff-F0-t0} and \eqref{eqn:Pd-l-I0-Pd-t0} we get: 
\be
\frac{\pd^{l+1} F_0}{\pd t_0^{l+1}}
= \sum_{k=1}^\infty (k+1)\cdots (k+l-1) \sum_{p_1 + \cdots + p_{k} = k+l-1 } \frac{t_{p_1}}{p_1!} \cdots
\frac{t_{p_{k}}}{p_{k}!}.
\ee
Similarly, one can differentiate
\be
\frac{\pd F_0}{\pd t_0} = \sum_{k=1}^\infty \frac{1}{k(1-t_1)^k}
\sum_{\substack{p_1 + \cdots + p_k = k-1\\p_j \neq 1,\;\; j=1,\dots, k}} \frac{t_{p_1}}{p_1!} \cdots
\frac{t_{p_k}}{p_k!}
\ee
repeated as follows.
First rewrite it in the following form:
\be
\frac{\pd F_0}{\pd t_0} = \sum_{k=1}^\infty \frac{1}{k(1-t_1)^k} \sum_{m=1}^{k} \binom{k}{m} t_0^m
\sum_{\substack{p_1 + \cdots + p_{k-m} = k-1\\p_j \neq 0, 1,\;\; j=1,\dots, k-m}} \frac{t_{p_1}}{p_1!} \cdots
\frac{t_{p_k}}{p_k!}.
\ee
Then one finds:
\ben
\frac{\pd^2 F_0}{\pd t_0^2} & = & \frac{1}{1-t_1} +  \sum_{k=2}^\infty \frac{1}{k(1-t_1)^k} \sum_{m=1}^{k} \binom{k}{m} m t_0^{m-1}
\sum_{\substack{p_1 + \cdots + p_{k-m} = k-1\\p_j \neq 0, 1,\;\; j=1,\dots, k-m}} \frac{t_{p_1}}{p_1!} \cdots
\frac{t_{p_{k-m}}}{p_{k-m}!} \\
 & = & \frac{1}{1-t_1} + \sum_{k=2}^\infty \frac{1}{(1-t_1)^k} \sum_{m=1}^{k} \binom{k-1}{m-1} t_0^{m-1}
\sum_{\substack{p_1 + \cdots + p_{k-m} = k-1\\p_j \neq 0, 1,\;\; j=1,\dots, k-m}} \frac{t_{p_1}}{p_1!} \cdots
\frac{t_{p_{k-m}}}{p_{k-m}!} \\
& = & \frac{1}{1-t_1} +  \sum_{k=2}^\infty \frac{1}{(1-t_1)^k}
\sum_{\substack{p_1 + \cdots + p_{k-1} = k-1\\p_j \neq 1,\;\; j=1,\dots, k-1}} \frac{t_{p_1}}{p_1!} \cdots
\frac{t_{p_{k-1}}}{p_{k-1}!} \\
& = &  \frac{1}{1-t_1} +  \sum_{k=1}^\infty \frac{1}{(1-t_1)^{k+1}}
\sum_{\substack{p_1 + \cdots + p_{k} = k\\p_j \neq 1,\;\; j=1,\dots, k}} \frac{t_{p_1}}{p_1!} \cdots
\frac{t_{p_{k}}}{p_{k}!} \\
& = &  \frac{1}{1-t_1} +  \sum_{k=2}^\infty \frac{1}{(1-t_1)^{k+1}}
\sum_{\substack{p_1 + \cdots + p_{k} = k\\p_j \neq 1,\;\; j=1,\dots, k}} \frac{t_{p_1}}{p_1!} \cdots
\frac{t_{p_{k}}}{p_{k}!}.
\een
One more time:
\ben
\frac{\pd^3 F_0}{\pd t_0^3}
& = & \sum_{k=2}^\infty \frac{k}{(1-t_1)^{k+1}}
\sum_{\substack{p_1 + \cdots + p_{k-1} = k\\p_j \neq 1,\;\; j=1,\dots, k-1}} \frac{t_{p_1}}{p_1!} \cdots
\frac{t_{p_{k-1}}}{p_{k-1}!} \\
& = & \frac{t_2}{(1-t_1)^{3}} + \sum_{k=2}^\infty \frac{k+1}{(1-t_1)^{k+2}}
\sum_{\substack{p_1 + \cdots + p_{k} = k+1 \\p_j \neq 1,\;\; j=1,\dots, k}} \frac{t_{p_1}}{p_1!} \cdots
\frac{t_{p_{k}}}{p_{k}!}.
\een
Inductively one gets:
\be
\frac{\pd^{l+1}F_0}{\pd t_0^{l+1}}
=  \sum_{k=1}^\infty \frac{(k+1)\cdots (k+l-1)}{(1-t_1)^{k+l}}
\sum_{\substack{p_1 + \cdots + p_{k} = k+l-1 \\p_j \neq 1,\;\; j=1,\dots, k}} \frac{t_{p_1}}{p_1!} \cdots
\frac{t_{p_{k}}}{p_{k}!}.
\ee

One reason that we derive such explicit expressions for derivatives of $F_0$ in $t_0$ is that
they appear in the Feynman rules for $n$-point functions derived in \S \ref{sec:Feynman for N-Point}.

\subsection{Integration of the genus zero polymer equation}

Notice that
$F_0|_{t_0=0} =0$,
so we have:
\ben
&& \int_0^{t_0} \biggl(\frac{\pd F_0}{\pd t_0}\biggr)^n dt_0 \\
& = & t_0 \biggl(\frac{\pd F_0}{\pd t_0}\biggr)^n
- n \int_0^{t_0} t_0 \biggl(\frac{\pd F_0}{\pd t_0}\biggr)^{n-1}
d \frac{\pd F_0}{\pd t_0} \\
& = & t_0 \biggl(\frac{\pd F_0}{\pd t_0}\biggr)^n
- n \int_0^{t_0} \biggl(\frac{\pd F_0}{\pd t_0}
- \sum_{m \geq 1} \frac{t_m}{m!} \biggl(\frac{\pd F_0}{\pd t_0}\biggr)^m \biggr)
\biggl(\frac{\pd F_0}{\pd t_0}\biggr)^{n-1} d\frac{\pd F_0}{\pd t_0} \\
& = & t_0 \biggl(\frac{\pd F_0}{\pd t_0}\biggr)^n
- \frac{n}{n+1} \biggl(\frac{\pd F_0}{\pd t_0}\biggr)^{n+1}
+ \sum_{m \geq 1} \frac{nt_m}{m!(m+n)}
\biggl(\frac{\pd F_0}{\pd t_0}\biggr)^{m+n}  \\
& = & t_0 I_0^n - \frac{n}{n+1} I_0^{n+1}
+ \sum_{m \geq 1} \frac{nt_m}{m!(m+n)} I_0^{m+n},
\een
We now integrate
the genus zero polymer equation \eqref{eqn:Polymer-Genus-0}:
\ben
&& \frac{\pd F_0}{\pd t_0}
= \sum_{n \geq 0} \frac{t_n}{n!} \biggl(\frac{\pd F_0}{\pd t_0}\biggr)^n
\een
to get
\ben
F_0  & = & \frac{t_0^2}{2}
+ \sum_{n=1}^\infty \frac{t_n}{n!} \int_0^{t_0} \biggl(\frac{\pd F_0}{\pd t_0}\biggr)^n dt_0 \\
& = &  \frac{t_0^2}{2}
+ \sum_{n=1}^\infty \frac{t_n}{n!} t_0 I_0^n
-  \sum_{n=1}^\infty \frac{t_n}{(n-1)!} \biggl( \frac{I_0^{n+1}}{n+1}
- \sum_{m \geq 1} \frac{t_m}{m!} \frac{I_0^{m+n}}{m+n} \biggr).
\een

We now show that this gives an alternative proof of the identity:
\be
 F_0  = t_0 I_0 + (t_1-1) \frac{I_0^2}{2!} + t_2 \frac{I_0^3}{3!} + \cdots
 = \sum_{n=0}^\infty \frac{t_{n}-\delta_{n,1}}{(n+1)!} I_0^{n+1}.
\ee
It suffices to show that
\ben
&& \sum_{n=1}^\infty \frac{t_n}{(n-1)!} \sum_{m \geq 1} \frac{t_m}{m!} \frac{I_0^{m+n}}{m+n} \\
& = & \sum_{n=0}^\infty \frac{t_{n}-\delta_{n,1}}{(n+1)!} I_0^{n+1}
 + \sum_{n=1}^\infty \frac{t_n}{(n-1)!}   \frac{I_0^{n+1}}{n+1} - \frac{t_0^2}{2}
- \sum_{n=1}^\infty \frac{t_n}{n!} t_0 I_0^n.
\een
The right-hand side can be rewritten as follows:
\ben
&&  \sum_{n=0}^\infty t_n\frac{I_0^{n+1}}{(n+1)!}  - \frac{I_0^2}{2}
 + \sum_{n=1}^\infty n t_n  \frac{I_0^{n+1}}{(n+1)!} - \frac{t_0^2}{2}
- t_0(I_0-t_0) \\
& = & \half (I_0-t_0)^2;
\een
the left-hand side can be dealt with as follows:
\ben
&& \sum_{n=1}^\infty \frac{t_n}{(n-1)!} \sum_{m \geq 1} \frac{t_m}{m!} \frac{I_0^{m+n}}{m+n}
= \sum_{m, n=1}^\infty \frac{t_n}{(n-1)!}  \frac{t_m}{m!} \frac{I_0^{m+n}}{m+n} \\
& = & \half \sum_{m, n=1}^\infty  \biggl(\frac{t_n}{(n-1)!}  \frac{t_m}{m!}
+ \frac{t_m}{(m-1)!}  \frac{t_n}{n!} \biggr) \frac{I_0^{m+n}}{m+n} \\
& = & \half \sum_{m, n=1}^\infty   \frac{t_n}{n!}  \frac{t_m}{m!} I_0^{m+n}
= \half \biggl( \sum_{n \geq 1} \frac{t_n}{n!} I_0^n \biggr)^2 = \half (I_0-t_0)^2.
\een
This finishes the proof.

\section{Virasoro Constraints for Topological 1D Gravity}

\label{sec:Virasoro}

In this Section we study the applications of Virasoro constraints for topological 1D gravity derived
in \cite{Nishigaki-Yoneya1}.
These include another derivation of the formulas for $F_g$ in $I$-coordinates.

\subsection{Virasoro constraints from flow equation and polymer equation}

By the polymer equation
\ben
&& \sum_{n \geq 0} \frac{t_n-\delta_{n,1} }{n!} \lambda^{2n}  \frac{\pd^n}{\pd t_0^n} Z = 0.
\een
and the flow equation:
\ben
&& \frac{\pd Z}{\pd t_n} = \frac{\lambda^{2n}}{(n+1)!} \frac{\pd^{n+1} Z}{\pd t_0^{n+1}}
\een
one gets the {\em puncture equation} for topological 1D gravity:
\be \label{eqn:Puncture}
(t_0 + \sum_{n \geq 1} ( t_n-\delta_{n,1} ) \lambda^{2}  \frac{\pd}{\pd t_{n-1}} ) Z = 0.
\ee
Take derivative in $t_0$ and apply the flow equation again:
\be \label{eqn:Dilaton}
\biggl(1 + \sum_{n \geq 0} (n+1) ( t_n-\delta_{n,1} )   \frac{\pd}{\pd t_{n}} \biggr) Z = 0.
\ee
This is the {\em dilaton equation} for topological 1D gravity.
Take the $m+1$-th derivative in $t_0$ by Leibnitz formula:
\ben
\biggl( (m+1) \frac{\pd^{m}}{\pd t_0^{m}} + \sum_{n \geq 0} \frac{t_n -\delta_{n,1}}{n!} \lambda^{2n}
\frac{\pd^{m+n+1}}{\pd t_0^{m+n+1}} \biggr) Z = 0,
\een
 and apply the flow equation:
\ben
&& \biggl( (m+1)! \lambda^{-2(m-1)} \frac{\pd}{\pd t_{m-1}} \\
& + & \sum_{n \geq 0} \frac{(m+n+1)!(t_n -\delta_{n,1})}{n!} \lambda^{2n- 2(m+n)}
\frac{\pd}{\pd t_{m+n}} \biggr) Z = 0.
\een
This is how the authors of \cite{Nishigaki-Yoneya1} derived the following:

\begin{thm}
The partition function $Z$ of topological 1D gravity satisfies the following equations for $m \geq -1$:
\be \label{eqn:Virasoro}
L_m Z = 0,
\ee
where
\bea
&& L_{-1} = \frac{t_0}{\lambda^2} + \sum_{m \geq 1} (t_{m}-\delta_{m,1}) \frac{\pd}{\pd t_{m-1}}, \\
&& L_0 = 1 + \sum_{m \geq 0} (t_{m}-\delta_{m,1}) (m+1) \frac{\pd}{\pd t_{m}}, \\
&& L_m = \lambda^2(m+1)! \frac{\pd}{\pd t_{m-1}}
+  \sum_{n \geq 1} (t_{n-1}-\delta_{n,2}) \frac{(m+n)!}{(n-1)!} \frac{\pd}{\pd t_{m+n-1}},
\eea
for $m \geq 1$.
\end{thm}

\begin{thm}
The operators $\{L_n\}_{n \geq -1}$ satisfies the Virasoro commutation relations:
\be
[L_m, L_n] = (m-n) L_{m+n}.
\ee
\end{thm}

\begin{proof}
This can be established by a straightforward calculation.
Another proof given in \cite{Nishigaki-Yoneya1} will be discussed in next Section.
\end{proof}

\subsection{An application of the dilaton equation}

The dilaton equation can be rewritten as
\be
\frac{\pd F}{\pd t_1} = \sum_{m \geq 0} \frac{m+1}{2} t_m \frac{\pd F}{\pd t_m} + \frac{1}{2}.
\ee
In terms of correlators,
\bea
&& \corr{\tau_1}_1 = \frac{1}{2}, \\
&& \corr{\tau_1 \prod_{j=1}^n \tau_{a_j}}_g
= \sum_{j=1}^n \frac{a_j+1}{2} \corr{\prod_{j=1}^n \tau_{a_j}}_g.
\eea
Therefore,
\be
\corr{\tau_1^m}_1 = \frac{1}{2} (m-1)!,
\ee
and for $a_2, \dots, a_n \neq 1$ which satisfies the selection rule \eqref{eqn:Selection},
$$a_1+ \cdots + a_n = 2g-2+n,$$
we have
\be
\corr{\tau_m \prod_{j=1}^n \tau_{a_j}}_g
= \prod_{k=0}^{m-1} (g-1+n +k) \cdot \corr{\prod_{j=1}^n \tau_{a_j}}_g.
\ee
It follows that we have

\begin{thm} \label{thm:1-t1}
The free energy of topological 1D gravity can be rewritten in the following form:
\be
F = \half \log (1-t_1)
+ \sum_{\substack{g,n \geq 0\\ 2g-2+n > 0}} \sum_{\substack{a_2, \dots, a_n \neq 1\\ a_1+ \cdots + a_n = 2g-2+n }}
\frac{\corr{\prod_{j=1}^n \tau_{a_j}}_g}{(1-t_1)^{g-1+n}}.
\ee
\end{thm}

\subsection{The operator $L_{-1}$ in I-coordinates}

By \eqref{eqn: diff I0} and \eqref{eqn:T-in-I},
we have
\be \label{eqn:L-1-in-I}
L_{-1} = - \frac{\pd}{\pd I_0} + \frac{1}{\lambda^2} \sum_{n=0}^\infty \frac{(-1)^nI_0^n}{n!}I_n.
\ee
This can also be checked by \eqref{eqn:diff-tk}:
\ben
\sum_{k\geq 0} t_{k+1} \frac{\pd}{\pd t_k}
& = & \sum_{k \geq 0} t_{k+1} \biggl( \frac{1}{1-I_1}\frac{I_0^k}{k!}  \frac{\pd}{\pd I_0} +
\frac{I_0^k}{k!} \sum_{l \geq 1} \frac{I_{l+1}}{1-I_1} \frac{\pd}{\pd I_l}
+ \sum_{1 \leq l \leq k} \frac{I_0^{k-l}}{(k-l)!} \frac{\pd}{\pd I_l} \biggr) \\
& = &  \frac{1}{1-I_1} \sum_{k\geq 0} t_{k+1} \frac{I_0^k}{k!} \cdot \frac{\pd}{\pd I_0} +
\sum_{k\geq 0} t_{k+1} \frac{I_0^k}{k!} \sum_{l \geq 1} \frac{I_{l+1}}{1-I_1} \frac{\pd}{\pd I_l} \\
& + & \sum_{k\geq 0} t_{k+1} \sum_{1 \leq l \leq k} \frac{I_0^{k-l}}{(k-l)!} \frac{\pd}{\pd I_l} \\
& = & \frac{I_1}{1-I_1}  \cdot \frac{\pd}{\pd I_0}
+ \sum_{l \geq 1} \frac{I_1I_{l+1}}{1-I_1} \frac{\pd}{\pd I_l}
+ \sum_{l=1}^\infty \sum_{k\geq l} t_{k+1} \frac{I_0^{k-l}}{(k-l)!} \frac{\pd}{\pd I_l} \\
& = & \frac{I_1}{1-I_1}  \cdot \frac{\pd}{\pd I_0}
+ \sum_{l \geq 1} \frac{I_1I_{l+1}}{1-I_1} \frac{\pd}{\pd I_l}
+ \sum_{l=1}^\infty I_{l+1} \frac{\pd}{\pd I_l} \\
& = & \frac{I_1}{1-I_1}  \cdot \frac{\pd}{\pd I_0}
+ \sum_{l \geq 1} \frac{I_{l+1}}{1-I_1} \frac{\pd}{\pd I_l}.
\een

It follows from \eqref{eqn:L-1-in-I} that
\bea
&& \frac{\pd F_0}{\pd I_0} = \sum_{n=0}^\infty \frac{(-1)^nI_0^n}{n!}I_n, \\
&& \frac{\pd F_g}{\pd I_0} = 0, \;\;\; g \geq 1.
\eea
Therefore,
we rederive the following result obtained in \S \ref{sec:Partition-in-I-Coord}.

\begin{thm}
The genus zero part of the free energy $F_0$ of topological 1D gravity is given in I-coordinates by:
\be
F_0= \frac{1}{2} I_0^2 + \sum_{n=0}^\infty \frac{(-1)^nI_0^{n+1}}{(n+1)!}I_n.
\ee
When $g \geq 1$,
$F_g$ is independent of $I_0$.
\end{thm}

\subsection{Dilaton operator in I-coordinates}

\begin{lem}
The dilaton operator $L_0$ is given in I-coordinates by:
\be
L_0 = - I_0 \frac{\pd}{\pd I_0} - 2 \frac{\pd}{\pd I_1} + \sum_{l \geq 1} (l+1) I_l \frac{\pd}{\pd I_l} + 1.
\ee
\end{lem}

\begin{proof}
By \eqref{eqn:diff-tk},
\ben
&& \sum_{k\geq 0} (k+1) t_{k} \frac{\pd}{\pd t_k}
= \sum_{k,l\geq 0} (k+1) t_{k} \biggl( \frac{1}{1-I_1}\frac{I_0^k}{k!}  \frac{\pd}{\pd I_0} \\
& + & \frac{I_0^k}{k!} \sum_{l \geq 1} \frac{I_{l+1}}{1-I_1} \frac{\pd}{\pd I_l}
+ \sum_{1 \leq l \leq k} \frac{I_0^{k-l}}{(k-l)!} \frac{\pd}{\pd I_l} \biggr) \\
& = & \frac{I_0I_1+I_0}{1-I_1} \frac{\pd}{\pd I_0}
+ \sum_{l\geq 1} \biggl(\frac{(I_0I_1+I_0)I_{l+1}}{1-I_1} + (I_0I_{l+1}+(l+1)I_l) \biggr)\frac{\pd}{\pd I_l} \\
& = & \frac{I_0I_1+I_0}{1-I_1} \frac{\pd}{\pd I_0}
+ \sum_{l\geq 1} \biggl(\frac{2I_0I_{l+1}}{1-I_1} + (l+1)I_l  \biggr)\frac{\pd}{\pd I_l}.
\een
And by \eqref{eqn:diff-tk} for $k=1$,
\ben
&& \frac{\pd}{\pd t_1}
=  \frac{I_0}{1-I_1} \frac{\pd}{\pd I_0} + \biggl(\frac{I_2I_0}{1-I_1} + 1\biggr) \frac{\pd}{\pd I_1}
+ \sum_{l \geq 2} \frac{I_{l+1}I_0}{1-I_1} \frac{\pd}{\pd I_l}.
\een
Therefore,
\ben
L_0 & = & -2 \frac{\pd}{\pd t_1} + \sum_{k \geq 0} (k+1) t_k \frac{\pd}{\pd t_{k}} +1 \\
& = & -2 \biggl( \frac{I_0}{1-I_1} \frac{\pd}{\pd I_0} + \biggl(\frac{I_2I_0}{1-I_1} + 1\biggr) \frac{\pd}{\pd I_1}
+ \sum_{l \geq 2} \frac{I_{l+1}I_0}{1-I_1} \frac{\pd}{\pd I_l} \biggr) \\
& + & \frac{I_0I_1+I_0}{1-I_1} \frac{\pd}{\pd I_0}
+ \sum_{l\geq 1} \biggl(\frac{2I_0I_{l+1}}{1-I_1} + (l+1)I_l  \biggr)\frac{\pd}{\pd I_l} \\
& = &  - I_0 \frac{\pd}{\pd I_0} - 2 \frac{\pd}{\pd I_1} + \sum_{l \geq 1} (l+1) I_l \frac{\pd}{\pd I_l} + 1.
\een
\end{proof}

From the dilaton equation,
one gets
\bea
&&  \frac{\pd F_0}{\pd I_1} = \sum_{l \geq 1} \frac{l+1}{2} I_l \frac{\pd F_0}{\pd I_l}
+  \frac{1}{2} \frac{\pd F_0}{\pd I_0}, \label{eqn:Dilaton-0} \\
&&  \frac{\pd F_1}{\pd I_1} = \sum_{l \geq 1} \frac{l+1}{2} I_l \frac{\pd F_1}{\pd I_l} +  \frac{1}{2},
\label{eqn:Dilaton-1} \\
&& \frac{\pd F_g}{\pd I_1} = \sum_{l \geq 1} \frac{l+1}{2} I_l \frac{\pd F_g}{\pd I_l}, \;\;\; g \geq 2.
\label{eqn:Dilaton-g}
\eea

\subsection{Solution in positive genera}

In this Subsection we will solve the dilaton equation for $g \geq 1$.
We rederive the following result obtained in \S \ref{sec:Partition-in-I-Coord}.

\begin{thm}
In $I$-coordinates we have
\be
F_1 = \frac{1}{2} \ln \frac{1}{1- I_1}
\ee
and for $g \geq 1$,
\be
F_g  =  \sum_{\sum_{j=2}^{2g-1}  \frac{j-1}{2} l_j = g-1}
 \corr{\tau_2^{l_2} \cdots \tau_{2g-1}^{l_{2g-1}}}_g
\prod_{j=2}^{2g-1} \frac{1}{l_j!}\biggl( \frac{I_j}{(1-I_1)^{(j+1)/2}}\biggr)^{l_j}.
\ee

\end{thm}

\begin{proof}
By the puncture equation we have already seen that $F_g$ does not depend on $I_0$ when $g \geq 1$.
Write $F_g$ as formal power series in $I_1$,
with coefficients a priori formal series in $I_2, I_3, \dots$,
\be \label{eqn:Fg-Unknown}
F_g = a_{0,g}(I_2, I_3, \dots) + a_{1,g}(I_2, I_3, \dots) I_1 + \cdots.
\ee
Write
\be
a_{0,g} = \sum \alpha_{l_2, \dots, l_m} \frac{I_2^{l_2}}{l_2!} \cdots \frac{I_m^{l_m}}{l_m!}
\ee
By comparing the coefficients of $\frac{t_2^{l_2}}{l_2!} \cdots \frac{t_m^{l_m}}{l_m!}$
on both sides of \eqref{eqn:Fg-Unknown},
 it is easy to see that:
\ben
\alpha_{l_2, \dots, l_m} = \frac{\pd^{l_2+ \cdots + l_m} F_g}{(\pd t_2)^{l_2} \cdots (\pd t_m)^{l_m}} \biggr|_{\bt = 0}
= \corr{\tau_2^{l_2} \cdots \tau_m^{l_m}}_g.
\een
This vanishes unless the following selection rule is satisfied:
\be
\sum_{j=2}^m jl_j = 2g-2  + \sum_{j=2}^m l_j.
\ee
Assume $l_m \geq 1$,
then
$$ m - 1 \leq \sum_{j=2}^m (j-1) l_j = 2g-2,$$
hence
$$m \leq 2g-1.$$
Therefore,
\be
a_{0,g} = \sum  \corr{\tau_2^{l_2} \cdots \tau_{2g-1}^{l_{2g-1}}}_g \frac{I_2^{l_2}}{l_2!} \cdots
\frac{I_m^{l_{2g-1}}}{l_{2g-1}!},
\ee
where the summation is taken over all nonnegative integers $l_2, \dots, l_{2g-1}$ such that
\be
\sum_{j=2}^{2g-1}  \frac{j+1}{2} l_j = g-1 + \sum_{j=2}^{2g-1} l_j.
\ee

Let
$$\tilde{E} = \sum_{l\geq 2} \frac{l+1}{2} I_l \frac{\pd}{\pd I_l}.$$
The the equation \eqref{eqn:Dilaton-1} gives us the following recursion relations:
\bea
&& a_{1,g} = \tilde{E} a_{0,g} + \frac{\delta_{g,1}}{2}, \\
&& n a_{n,g} = (n-1) a_{n-1, g} + \tilde{E} a_{n-1, g}, \;\;\; n \geq 2.
\eea
When $g=1$,
$a_{0, 1} = 0$.
One easily sees that
$a_{n,1} = \frac{1}{2n}$.
Therefore,
\ben
F_1 & = & \sum_{n \geq 1} \frac{1}{2n} I_1^n = \frac{1}{2} \ln \frac{1}{1- I_1}.
\een
When $g > 1$,
one finds
\ben
&& a_{n, g} =  \sum_{\sum_{j=2}^{2g-1}  \frac{j-1}{2} l_j = g-1}  \corr{\tau_2^{l_2} \cdots \tau_{2g-1}^{l_{2g-1}}}_g
(-1)^n \binom{-(g-1+\sum_{j=2}^{2g-1} l_j)}{n}
\frac{I_2^{l_2}}{l_2!} \cdots\frac{I_m^{l_{3g-2}}}{l_{3g-2}!}.
\een
This proves:
\ben
F_g  & = &  \sum_{\sum_{j=2}^{2g-1}  \frac{j-1}{2} l_j = g-1}
 \corr{\tau_2^{l_2} \cdots \tau_{2g-1}^{l_{2g-1}}}_g \frac{1}{(1-I_1)^{g-1+\sum_{j=2}^{2g-1} l_j}}
\prod_{j=2}^{2g-1} \frac{I_j^{l_j}}{l_j!} \\
& = &  \sum_{\sum_{j=2}^{2g-1}  \frac{j-1}{2} l_j = g-1}
 \corr{\tau_2^{l_2} \cdots \tau_{2g-1}^{l_{2g-1}}}_g
\prod_{j=2}^{2g-1} \frac{1}{l_j!}\biggl( \frac{I_j}{(1-I_1)^{(j+1)/2}}\biggr)^{l_j}.
\een
\end{proof}

For example,
\ben
F_2 & = & \half \corr{\tau_2^2}_2 \frac{I_2^2}{(1-I_1)^3}
+ \corr{\tau_3}_2 \frac{I_3}{(1-I_1)^2}, \\
F_3 & = & \corr{\tau_2^4}_3 \cdot \frac{1}{4!}\frac{I_2^4}{(1-I_1)^6}
+ \corr{\tau_2^2\tau_3}_3 \cdot \frac{1}{2!}\frac{I_2^2I_3}{(1-I_1)^5}
+ \corr{\tau_3^2}_3 \cdot \frac{1}{2!}\frac{I_3^2}{(1-I_1)^4} \\
&& + \corr{\tau_2\tau_4}_3 \cdot \frac{I_2I_4}{(1-I_1)^4}
+ \corr{\tau_5}_3 \frac{I_5}{(1-I_1)^3}, \\
F_4 & = & \corr{\tau_2^6}_4 \cdot \frac{1}{6!}\frac{I_2^6}{(1-I_1)^9}
+ \corr{\tau_2^4\tau_3}_4 \cdot \frac{1}{4!}\frac{I_2^4I_3}{(1-I_1)^8}
+ \corr{\tau_2^2\tau_3^2}_4 \cdot \frac{1}{2!2!}\frac{I_2^2I_3^2}{(1-I_1)^7} \\
& + & \corr{\tau_3^3}_4 \frac{1}{3!} \frac{I_3^3}{(1-I_1)^6} + \corr{\tau_2^3\tau_4}_3 \cdot \frac{1}{3!}\frac{I_2^3I_4}{(1-I_1)^7}
+ \corr{\tau_2\tau_3\tau_4}_4 \cdot \frac{I_2I_3I_4}{(1-I_1)^6}\\
& + & \corr{\tau_4^2}_4 \cdot \frac{1}{2!} \frac{I_4^2}{(1-I_1)^5}
+ \corr{\tau_2^2\tau_5}_4 \cdot \frac{1}{2!} \frac{I_2^2I_5}{(1-I_1)^6}
+ \corr{\tau_2\tau_6}_4 \cdot \frac{I_2I_6}{(1-I_1)^5} \\
& + & \corr{\tau_7}_4 \cdot \frac{I_7}{(1-I_1)^4}.
\een
For the relevant correlators,
see \S \ref{sec:General expression for Z}.

\section{Operator Algebra of Topological 1D Gravity and W-Constraints}

\label{sec:W-Constraints}

In this Section we study the operator algebra that leads to the Virasoro constraints
and W-constraints of topological 1D gravity.
We also present a different version of Virasoro constraints.\emph{}

\subsection{Virasoro constraints as Dyson-Schwinger equations}

For $n \geq -1$, by \eqref{eqn:Gaussian-By-Parts} one has:
\be
\frac{1}{\sqrt{2\pi}\lambda}
\int dx \cdot \frac{\pd}{\pd x} \biggl( \frac{x^{n+1}}{(n+1)!} \cdot
\exp \frac{1}{\lambda^2} \biggl( -\half x^2 + \sum_{n \geq 1} t_{n-1} \frac{x^n}{n!}  \biggr) \biggr) \\
= 0.
\ee
Rewrite the left-hand side as follows:
For $n \geq 1$,
\ben
&& \frac{1}{\sqrt{2\pi}\lambda}
\int dx \cdot  \biggl( \frac{x^{n}}{n!} + \frac{x^{n+1}}{(n+1)!\lambda^2}\cdot \sum_{m \geq 1}
(t_{m-1}-\delta_{m,2}) \frac{x^{m-1}}{(m-1)!} \biggr) \\
&& \cdot \exp \frac{1}{\lambda^2} \biggl( -\half x^2 + \sum_{k \geq 1} t_{n-1} \frac{x^k}{k!} \biggr) \\
& = & \biggl( \lambda^2 \frac{\pd}{\pd t_{n-1}}
+  \sum_{m \geq 1} (t_{m-1}-\delta_{m,2}) \frac{(m+n)!}{(n+1)!(m-1)!} \frac{\pd}{\pd t_{m+n-1}}\biggr) Z;
\een
there are two exceptional cases: for $n=-1$,
\ben
&& \frac{1}{\sqrt{2\pi}\lambda}
\int dx \cdot  \biggl( \frac{1}{\lambda^2}\cdot \sum_{m \geq 1}
(t_{m-1}-\delta_{m,2}) \frac{x^{m-1}}{(m-1)!} \biggr) \\
&& \cdot \exp \frac{1}{\lambda^2} \biggl( -\half x^2 + \sum_{k \geq 1} t_{k-1} \frac{x^k}{k!} \biggr) \\
& = & \biggl(\frac{t_0}{\lambda^2} + \sum_{m \geq 1} (t_{m}-\delta_{m,1}) \frac{\pd}{\pd t_{m-1}}\biggr) Z;
\een
and for $n=0$,
\ben
&& \frac{1}{\sqrt{2\pi}\lambda}
\int dx \cdot  \biggl( 1 + \frac{x}{\lambda^2}\cdot \sum_{m \geq 1}
(t_{m-1}-\delta_{m,2}) \frac{x^{m-1}}{(m-1)!} \biggr) \\
&& \cdot \exp \frac{1}{\lambda^2} \biggl( -\half x^2 + \sum_{k \geq 1} t_{n-1} \frac{x^k}{k!} \biggr) \\
& = & \biggl( 1
+  \sum_{m \geq 0} (t_{m}-\delta_{m,1}) (m+1) \frac{\pd}{\pd t_{m}}\biggr) Z.
\een

\subsection{Virasoro constraints from loop equation}

By \eqref{eqn:Gaussian-By-Parts} one has:
\be \label{eqn:Loop}
\frac{1}{\sqrt{2\pi}\lambda}
\int dx \cdot \frac{\pd}{\pd x} \biggl( \frac{1}{z-x}  \cdot
\exp \frac{1}{\lambda^2} \biggl( -\half x^2 + \sum_{n \geq 1} t_{n-1} \frac{x^n}{n!}  \biggr) \biggr) \\
= 0.
\ee
This is called the {\em loop equation} of the topological 1D gravity.
By the expansion
\be
\frac{1}{z-x} = \sum_{n=0}^\infty \frac{x^n}{z^{n+1}},
\ee
one can rederive the Virasoro constraints by consider the coefficients of $\frac{1}{z^{n+2}}$ for $n \geq -1$.

\subsection{W-constraints}

By \eqref{eqn:Gaussian-By-Parts} one has:
\be
\frac{1}{\sqrt{2\pi}\lambda}
\int dx \cdot \frac{\pd^k}{\pd x^k} \biggl( \frac{x^{n+1}}{(n+1)!} \cdot
\exp \frac{1}{\lambda^2} \biggl( -\half x^2 + \sum_{n \geq 1} t_{n-1} \frac{x^n}{n!}  \biggr) \biggr) \\
= 0.
\ee
Use Leibniz formula to rewrite the left-hand side as follows:
\ben
&& \frac{1}{\sqrt{2\pi}\lambda}
\int dx \cdot  \biggl( \sum_{j=0}^k \binom{k}{j} \frac{\pd^{k-j}}{\pd x^{k-j}} \biggl(\frac{x^{n+1}}{(n+1)!} \biggr) \\
&& \cdot
\frac{\pd^j}{\pd x^j} \exp \frac{1}{\lambda^2} \biggl( -\half x^2 + \sum_{n \geq 1} t_{n-1} \frac{x^n}{n!}  \biggr) \biggr) \\
& = & \frac{1}{\sqrt{2\pi}\lambda}
\int dx \cdot  \biggl( \sum_{j=0}^k \binom{k}{j} (k-j)! \binom{n+1-k+j}{k-j} \frac{x^{n+1-k+j}}{(n+1)!} \\
&& \cdot
\frac{\pd^j}{\pd x^j} \exp \frac{1}{\lambda^2} \biggl( -\half x^2 + \sum_{n \geq 1} t_{n-1} \frac{x^n}{n!}  \biggr) \biggr) \\
& = &  \frac{1}{\sqrt{2\pi}\lambda}
\int dx \cdot  \biggl( \sum_{j=0}^k \binom{k}{j} (k-j)! \binom{n+1-k+j}{k-j} \frac{x^{n+1-k+j}}{(n+1)!} \\
&& \cdot
\sum_{l \geq 0} A_{j,l}(\bt) \frac{x^l}{l!} \cdot \exp \frac{1}{\lambda^2} \biggl( -\half x^2 + \sum_{n \geq 1} t_{n-1} \frac{x^n}{n!}  \biggr) \biggr) ,
\een
where
\ben
&& \frac{\pd^j}{\pd x^j} \exp \frac{1}{\lambda^2} \biggl( -\half x^2 + \sum_{n \geq 1} t_{n-1} \frac{x^n}{n!}  \biggr) \\
& = & \sum_{l \geq 0} A_{j,l}(\bt) \frac{x^l}{l!} \cdot \exp \frac{1}{\lambda^2} \biggl( -\half x^2 + \sum_{n \geq 1} t_{n-1} \frac{x^n}{n!}  \biggr).
\een
Therefore,
one gets a constraint:
\be
\begin{split}
& \sum_{j=0}^k \binom{k}{j} (k-j)! \binom{n+1-k+j}{k-j} \frac{x^{n+1-k+j}}{(n+1)!} \\
 \cdot &
\sum_{l \geq 0} A_{j,l}(\bt) \frac{x^l}{l!} (n+1-k+j+l)! \frac{\pd}{\pd t_{n-k+j+l}} Z = 0.
\end{split}
\ee

\subsection{Operator algebra}

Denote by $\cA$ the algebra of differential operators with polynomial coefficients:
\be
\sum_{m, n \geq 0} a_{m,n} x^m \pd_x^n
\ee
where $a_{m,n} = 0$ for $m \gg 0$ or $n \gg 0$.
One can take either $\{x^m\pd_x^n\}_{m,n \geq 0}$ or $\{\pd_x^nx^m\}_{m,n \geq 0}$ as a basis.
By the Leibniz formula,
they are related as follows:
\be
\pd_x^n x^m = \sum_{j=0}^n j! \binom{n}{j} \binom{m}{j} x^{m-j} \pd_x^{n-j}.
\ee
By induction one can also show that:
\be
x^m \pd_x^n = \sum_{j=0}^{n-1} (-1)^j j! \binom{n-1}{j} \binom{m}{j} \pd_x^{n-j} x^{m-j}.
\ee
With these formulas,
one gets the structure constants of the operator algebra $\cA$:
\be
x^{m_1} \pd_x^{n_1} \cdot x^{m_2} \pd_x^{n_2}
= \sum_{j=0}^{n_1} j! \binom{n_1}{j} \binom{m_2}{j} x^{m_1+m_2-j} \pd_x^{n_1+n_2-j},
\ee
and
\be
\pd_x^{n_1} x^{m_1} \cdot \pd_x^{n_2} x^{m_2}
= \sum_{j=0}^{n_2-1} (-1)^j \binom{n_2-1}{j} \binom{m_1}{j} \pd^{n_1+n_2-j} x^{m_1+m_2-j}.
\ee
The operator algebra $\cA$ is associative,
but not commutative.
The commutators of elements in $\cA$ are given by the above formula for structure constants:
\be
[x^{m_1} \pd_x^{n_1}, x^{m_2} \pd_x^{n_2}]
= \sum_{j \geq 1} j! \biggl( \binom{n_1}{j} \binom{m_2}{j}
- \binom{n_2}{j} \binom{m_1}{j} \biggr)
x^{m_1+m_2-j} \pd_x^{n_1+n_2-j},
\ee
and
\be
\begin{split}
& [\pd_x^{n_1} x^{m_1}, \pd_x^{n_2} x^{m_2}] \\
= & \sum_{j \geq 1} (-1)^j  \biggl( \binom{n_2-1}{j} \binom{m_1}{j}
-   \binom{n_1-1}{j} \binom{m_2}{j} \biggr) \pd^{n_1+n_2-j} x^{m_1+m_2-j}.
\end{split}
\ee
In particular,
when $n_1=n_2 = 1$,
\be
[\pd_x x^{m_1}, \pd_x x^{m_2}] = (m_2-m_1) \pd_x x^{m_1+m_2-1}.
\ee

\subsection{Representation of $\cA$}

Let us consider the following natural representation of the operator algebra on the following space:
\be
\cV = \biggl \{ \sum_{j=0}^\infty a_j(\bt; \lambda) x^j \cdot \vac:  \;\; a_j \in \bC[[\bt; \lambda^{-2}, \lambda^2]] \biggr\}.
\ee
where
\be
\vac = \exp \frac{1}{\lambda^2}
\biggl( -\half x^2 + \sum_{n \geq 1} t_{n-1} \frac{x^n}{n!}  \biggr)
\ee
is understood as the vacuum.
On $\cV$,
one also has the actions by the operators $\frac{\pd}{\pd t_n}$ and multiplications by $t_n$,
so $\cV$ actually admits a representation of $\cA \otimes \cB$,
where $\cB$ is the space of differential operators in $t_0, t_1, \dots, t_n, \dots$,
i.e., operators of the form
\be
\sum \alpha_{i_1, \dots, i_n} \frac{\pd}{\pd t_{i_1}} \cdots \frac{\pd}{\pd t_{i_n}},
\ee
where the coefficients $\alpha_{i_1, \dots, i_n} \in \bC[t_0, t_1, \dots, \lambda^2, \lambda^{-2}]$.
With the actions of $\cB$ one sees that the vacuum is not an uninteresting object that contains nothing,
but instead it contains everything in the sense that one can get every vector in $\cV$
by applying an operator in $\cB$ on $\vac$.

For convenience of notations,
we set
\be
\frac{\pd}{\pd t_{-1}} = \frac{1}{\lambda^2} \cdot.
\ee

\subsection{The action of $\cA$ on $\vac$}

Let us now examine the action of operators in $\cA$ on $\vac$.
We will need the following notations introduced by Itzykson-Zuber \cite{Itz-Zub}
in their study of 2D topological gravity:
\be
J_k = \sum_{n \geq 0} t_{n+k} \frac{x^n}{n!}.
\ee
Define
\be
\tilde{J}_k = I_k - \delta_{k,0} x - \delta_{k,1}.
\ee
Let
\be
S = - \frac{x^2}{2} + \sum_{n \geq 1} t_{n-1} \frac{x^n}{n!}.
\ee
Then
\be
\tilde{J}_k = \frac{\pd^k}{\pd x^k} S.
\ee
In particular,
\be
\tilde{J}_{k+1} = \pd_x \tilde{I}_k.
\ee
By induction one can easily get the following:

\begin{lem}
For $n \geq 0$,
\be
\pd_x^n \vac = D_n(\tilde{I}_0, \dots, \tilde{I}_{n-1}) \vac,
\ee
where $D_n$ is a polynomial in $\tilde{J}_0, \dots, \tilde{J}_{n-1}$ recursively defined as follows:
\bea
&& D_0 = 1, \label{eqn:D-Initial} \\
&& D_{n+1} = P D_n, \;\;\; n \neq 0, \label{eqn:D-Rec}
\eea
where
\be
P = \sum_{j\geq 0} \tilde{J}_{j+1} \frac{\pd}{\pd \tilde{J}_j} + \frac{\tilde{J}_0}{\lambda^2}.
\ee
\end{lem}

\begin{rmk}
If we define
\be
\frac{\pd}{\pd \tilde{J}_{-1}} = \frac{1}{\lambda^2} \cdot,
\ee
then we can write $P$ more compactly as
\be
P = \sum_{j \geq 0} \tilde{J}_j \frac{\pd}{\pd \tilde{J}_{j-1}}.
\ee
\end{rmk}

For example,
\ben
&& D_1 = \frac{1}{\lambda^2} \tilde{J}_0, \\
&& D_2 = \frac{1}{\lambda^2} \tilde{J}_1 + \frac{1}{\lambda^4} \tilde{J}_0^2, \\
&& D_3 = \frac{1}{\lambda^2} \tilde{J}_2 + \frac{3}{\lambda^4} \tilde{J}_0\tilde{J}_1
+ \frac{1}{\lambda^6} \tilde{I}_0^3, \\
&& D_4 = \frac{1}{\lambda^2} \tilde{J}_3 + \frac{4}{\lambda^4} \tilde{J}_0\tilde{J}_2
+ \frac{3}{\lambda^4} \tilde{J}_1^2 + \frac{6}{\lambda^6} \tilde{J}_0^2\tilde{J}_1
+ \frac{1}{\lambda^8} \tilde{J}_0^4.
\een
The recursive procedure of finding $D_n$ is very similar to the procedure
of find the decomposition of $V ^{\otimes n}$ into irreducible representations by the Littlewood-Richardson rule.
Indeed,
if one let $p_n = \frac{1}{\lambda^2}\tilde{J}_{n-1}$,
then each $D_n$ is a polynomial in $p_1, \dots, p_n$:
\be
D_n = \sum_{|\mu|= n} a_\mu p_\mu,
\ee
where the summation is over all Young diagrams with $n$ boxes.
In other words,
we associate a partition of $n$ to each monomial in $D_n$,
and represent the partition by its Young diagram.
Then we have
\be
Pp_\mu = \sum_{|\nu|=n+1} \alpha_\nu p_\nu,
\ee
the right-hand side of which can be obtained as follows:
Suppose that the Young diagram of $\mu$ has $l$ rows,
write down all the $l+1$ possible ways to either add a box from the right on a row,
or add a new row with one box from the bottom.
It is possible that not all the result diagrams are Young diagrams,
when they are not, simply switch the rows to make it a Young diagram,
this gives rise to the coefficients $\alpha_\nu$.
For example,
$$
\yng(1,1) \times \yng(1) = \yng(2,1) +\yng(1,2) + \yng(1,1,1) = 2 \yng(2,1) + \yng(1,1,1).
$$

\begin{thm}
For $n \geq 1$, $D_n$ is given by the following explicit formula:
\be
D_n = n! \sum_{\substack{m_1, \dots, m_n \geq 0 \\ \sum_{j=1}^n m_jj = n }}
\prod_{j=1}^n \frac{\tilde{I}_{j-1}^{m_j}}{(j!)^{m_j} m_j!}.
\ee
\end{thm}

\begin{proof}
One can directly check that the right-hand side satisfies the recursion relations \eqref{eqn:D-Rec} and the initial
value \eqref{eqn:D-Initial}.
\end{proof}

\subsection{Converting to actions of operators in $\cB$ on $\vac$}

We begin with the action of $x^m$ on $\vac$.
From the definition of $\vac$,
it is easy to see that

\begin{lem}
For $m \geq 0$,
\be
\frac{x^m}{m!} \vac = \lambda^2 \frac{\pd}{\pd t_{m-1}} \vac.
\ee
\end{lem}

\begin{cor}
For $n_1, \dots, n_k \geq 0$,
\be
\lambda^2\frac{\pd}{\pd t_{n_1-1}} \cdots \lambda^2\frac{\pd}{\pd t_{n_k-1}} \vac
= \binom{n_1+ \cdots + n_k}{n_1, \dots, n_k} \lambda^2 \frac{\pd}{\pd t_{n_1+\cdots + n_k-1}} \vac.
\ee
\end{cor}

\begin{proof}

\ben
&& \lambda^2\frac{\pd}{\pd t_{n_1-1}} \cdots \lambda^2\frac{\pd}{\pd t_{n_k-1}} \vac
= \frac{x^{n_1}}{n_1!} \cdots \frac{x^{n_k}}{n_k!} \vac \\
& = & \binom{n_1+ \cdots + n_k}{n_1, \dots, n_k}
\frac{x^{n_1+ \cdots + n_k}}{(n_1+\cdots + n_k)!} \vac  \\
& = & \binom{n_1+ \cdots + n_k}{n_1, \dots, n_k} \lambda^2 \frac{\pd}{\pd t_{n_1+\cdots + n_k-1}} \vac.
\een
\end{proof}

Next we have
\begin{lem}
\be
\pd_x \vac = \frac{1}{\lambda^2} \sum_{n=0}^\infty (t_n - \delta_{n,1}) \frac{\pd}{\pd t_{n-1}} \vac,
\ee

\be
\pd_x\frac{x^m}{m!} \vac = \lambda^2 \frac{\pd}{\pd t_{m-2}} \vac
\ee

\end{lem}

\be
\corr{\frac{x^n}{n!}} = \lambda^2 \frac{\pd}{\pd t_{n-1}} Z
\ee

\begin{thm}
The partition function of the topological 1D gravity satisfies the following equation:
\be \label{eqn:Join}
 \lambda^2\frac{\pd}{\pd t_{n_1-1}} \cdots \lambda^2\frac{\pd}{\pd t_{n_k-1}} Z
=   \binom{n_1+ \cdots + n_k}{n_1, \dots, n_k} \lambda^2 \frac{\pd}{\pd t_{n_1+\cdots + n_k-1}} Z.
\ee
In particular, in genus zero,
\be
\frac{\pd F_0}{\pd t_{n_1-1}} \cdots \frac{\pd F_0}{\pd t_{n_k-1}}
= \binom{n_1+ \cdots + n_k}{n_1, \dots, n_k} \frac{\pd F_0}{\pd t_{n_1+\cdots + n_k-1}}.
\ee
\end{thm}

\begin{proof}

\ben
&& \lambda^2\frac{\pd}{\pd t_{n_1-1}} \cdots \lambda^2\frac{\pd}{\pd t_{n_k-1}} Z
= \corr{\frac{x^{n_1}}{n_1!} \cdots \frac{x^{n_k}}{n_k!}} \\
& = & \binom{n_1+ \cdots + n_k}{n_1, \dots, n_k}
\corr{\frac{x^{n_1+ \cdots + n_k}}{(n_1+\cdots + n_k)!} } \\
& = & \binom{n_1+ \cdots + n_k}{n_1, \dots, n_k} \lambda^2 \frac{\pd}{\pd t_{n_1+\cdots + n_k-1}} Z.
\een
\end{proof}

\subsection{Another version of Virasoro constraints}
Recall we have derived that $Z$ satisfies the Virasoro constraints with the folowing Virasoro operators:
\bea
&& L_{-1} = \frac{t_0}{\lambda^2} + \sum_{n \geq 1} (t_{n}-\delta_{n,1}) \frac{\pd}{\pd t_{n-1}}, \\
&& L_0 = 1 + \sum_{m \geq 0} (t_{m}-\delta_{n,1}) (n+1) \frac{\pd}{\pd t_{n}}, \\
&& L_m = \lambda^2(m+1)! \frac{\pd}{\pd t_{m-1}}
+  \sum_{n \geq 0} (t_{n}-\delta_{n,1}) \frac{(m+n+1)!}{n!} \frac{\pd}{\pd t_{m+n}},
\eea
for $m \geq 1$.
Now applying \eqref{eqn:Join},
we get:

\begin{thm}
The partition function $Z$ of topological 1D gravity satisfies the following equations for $m \geq -1$:
\be \label{eqn:Virasoro-New}
\tilde{L}_m Z = 0,
\ee
where
\bea
&& \tilde{L}_{-1} = \frac{t_0}{\lambda^2} + \sum_{m \geq 1} (t_{m}-\delta_{m,1}) \frac{\pd}{\pd t_{m-1}},
\label{eqn:Virasoro-New--1} \\
&& \tilde{L}_0 = 1 + \sum_{m \geq 0} (t_{m}-\delta_{m,1}) (m+1) \frac{\pd}{\pd t_{m}}, \\
&& \tilde{L}_1 = 2 \lambda^2 \frac{\pd}{\pd t_0}
+  \sum_{n \geq 0} (t_{n}-\delta_{n,1}) \frac{(n+2)!}{n!} \frac{\pd}{\pd t_{n+1}}, \\
&& \tilde{L}_m = 2 \lambda^2 m! \frac{\pd}{\pd t_{m-1}}
+ \lambda^4 \sum_{\substack{m_1+m_2=m\\m_1, m_2 \geq 1}}  m_1!\frac{\pd}{\pd t_{m_1-1}}
\cdot  m_2!\frac{\pd}{\pd t_{m_2-1}} \label{eqn:Virasoro-New-m} \\
&& \;\;\;\; +  \sum_{n \geq 0} (t_{n}-\delta_{n,1}) \frac{(m+n+1)!}{n!} \frac{\pd}{\pd t_{m+n}}, \nonumber
\eea
for $m \geq 2$.
Furthermore,
$\{\tilde{L}_m\}_{m \geq 1}$ satisfies the following commutation relations:
\be
[\tilde{L}_m, \tilde{L}_n ] = 0,
\ee
for $m, n \geq -1$.
\end{thm}

\begin{proof}
When $m\geq 2$, for $m_1 = 1, \dots, m-1$, $m_2 = m - m_1$,
\ben
m!\lambda^2\frac{\pd}{\pd t_{m-1}} Z= \lambda^2 m_1! \frac{\pd}{\pd t_{m_1-1}} \cdot \frac{\pd}{\pd t_{m_2-1}} Z,
\een
it follows that
\ben
&& (m+1)! \lambda^2\frac{\pd}{\pd t_{m-1}} Z= 2 \cdot m!\lambda^2\frac{\pd}{\pd t_{m-1}} Z \\
&& \;\;\;\;\;\; + \sum_{\substack{m_1+m_2=m\\m_1,m_2\geq1}}
\lambda^2 m_1! \frac{\pd}{\pd t_{m_1-1}} \cdot \frac{\pd}{\pd t_{m_2-1}} Z.
\een
Therefore one can derive \eqref{eqn:Virasoro-New} from \eqref{eqn:Virasoro}.
The commutation relations can be checked by a standard calculation.
\end{proof}

\begin{cor}
The genus zero free energy $F_0$ of the topological 1D gravity satisifies the following equations:
\bea
&& \frac{\pd F_0}{\pd t_0} = t_0 + \sum_{m \geq 1} t_{m}  \frac{\pd F_0}{\pd t_{m-1}}, \\
&&  (m+2)! \frac{\pd F_0}{\pd t_{m+1}} =
  \sum_{n \geq 0}  \frac{(m+n+1)!}{n!} t_{n} \frac{\pd F_0}{\pd t_{m+n}},
\eea
where $m \geq 0$.
\end{cor}

\section{$N$-Point Functions in Topological 1D Gravity}

\label{sec:N-point function}

We will compute $n$-point functions in this Section.
We also present recursion relations satisfied by them.
Our technical tools are the loop operators.

\subsection{Two kinds of $n$-point functions}

In this subsection,
we define two kinds of $n$-point functions
of the topological 1D gravity:
\ben
&& \hat{W}_{g,n}(w_1, \dots, w_n;\bt)  = \delta_{g,0} \delta_{n,1} + \sum_{m_1, \dots, m_n \geq 0} \frac{\pd^n F_g}{\pd t_{m_1} \cdots \pd t_{m_n}}
w_1^{m_1} \cdots w_n^{m_n}, \\
&& W_{g,n}(z_1, \dots, z_n; \bt) = \frac{\delta_{g,0} \delta_{n,1}}{z_1} + \sum_{m_1, \dots, m_n \geq 1}
\frac{\pd^n F_g}{\pd t_{m_1-1} \cdots \pd t_{m_n-1}}
\prod_{j=1}^n m_j!z_j^{-m_j-1}.
\een
They are related to each other by Laplace transform:
\be
W_{g,n}(z_1, \dots, z_n;\bt)
= \int_{\bR_+^n}  e^{-\sum_{j=1}^n w_jz_j}
\hat{W}_{g,n}(z_1, \dots, z_n;\bt) d w_1 \cdots d w_n,
\ee
where $\bR_+= [0, +\infty)$.

Define
\ben
&& \hat{W}_{n}(w_1, \dots, w_n;\bt)  = \sum_{g \geq 0} \lambda^{2g}  \hat{W}_{g,n}(w_1, \dots, w_n;\bt), \\
&& W_{n}(z_1, \dots, z_n;\bt) = \sum_{g \geq 0} \lambda^{2g}  W_{g,n}(z_1, \dots, z_n;\bt).
\een
It is clear that
\ben
&& \hat{W}_{n}(w_1, \dots, w_n;\bt)  = \delta_{g,0} \delta_{n,1}
+ \lambda^2 \sum_{m_1, \dots, m_n \geq 0} \frac{\pd^n F}{\pd t_{m_1} \cdots \pd t_{m_n}}
w_1^{m_1} \cdots w_n^{m_n}, \\
&& W_{n}(z_1, \dots, z_n;\bt) = \frac{\delta_{g,0} \delta_{n,1}}{z_1}
+ \lambda^2 \sum_{m_1, \dots, m_n \geq 1}
\frac{\pd^n F}{\pd t_{m_1-1} \cdots \pd t_{m_n-1}} \prod_{j=1}^n m_j!z_j^{-m_j-1}.
\een

\subsection{Loop operators}

Similarly, define two kinds of loop operators:
\ben
&& \hat{B}(w) = \sum_{n \geq 0} w^n \frac{\pd}{\pd t_n}, \\
&& B(z) = \sum_{m \geq 1} \frac{m!}{z^{m+1}} \frac{\pd}{\pd t_{m-1}}.
\een
They are related by a Laplace transform,
and it is clear that for $n \geq 1$,
\bea
&& \hat{W}_n(w_1, \dots, w_n;\bt) = \delta_{n,1} + \hat{B}(w) \hat{W}_{n-1}(w_2, \dots, w_n;\bt),
\label{eqn:Hat(W)=Hat(B)Hat(W)} \\
&& W_n(z_1, \dots, z_n;\bt) = \frac{\delta_{n,1}}{z_1} + B(z_1) W_{n-1}(z_2, \dots, z_n;\bt).
\label{eqn:W=BW}
\eea
Here
\be
\hat{W}_0(\bt) = W_0(\bt) = \lambda^2F.
\ee

\subsection{Genus zero one-point functions}
By \eqref{eqn:Flow for F}, in genus zero we have:
\be
\frac{\pd F_0}{\pd t_m} = \frac{1}{(m+1)!} \biggl(\frac{\pd F_0}{\pd t_0}\biggr)^{m+1}
= \frac{1}{(m+1)!} I_0^{m+1}.
\ee
It follows that
\be
\hat{W}_{0,1}(w_1;\bt)  =  e^{w_1\frac{\pd F_0}{\pd t_0}} = e^{w_1I_0},
\ee
and
\be
W_{0,1}(z_1;\bt)
=   \frac{1}{z_1-\frac{\pd F_0}{\pd t_0}} = \frac{1}{z_1 - I_0}.
\ee
In particular,
\be
\frac{\pd F_0}{\pd t_n}(t_0) = \frac{t_0^{n+1}}{(n+1)!},
\ee
and so we have
\be
\hat{W}_{0,1}(w_1; t_0) =  e^{t_0w_1}
\ee
and
\be
W_{0,1}(z_1;t_0)
= \frac{1}{z_1-t_0}.
\ee

\subsection{One-point function in arbitrary genera}

By \eqref{eqn:Flow for F},
\be
\frac{\pd F}{\pd t_n}
= \frac{\lambda^{2n}}{(n+1)!} \biggl(\frac{\pd}{\pd t_0} + \frac{\pd F}{\pd t_0}\biggr)^{n+1} 1.
\ee
one gets:
\be \label{eqn:Pd-F-Gen}
\hat{W}_1(w;\bt) 
= e^{\lambda^2 w(\frac{\pd}{\pd t_0} + \frac{\pd F}{\pd t_0})}1
\ee
and
\be
W_1(z;\bt) 
= \frac{1}{z- \lambda^2(\frac{\pd}{\pd t_0} + \frac{\pd F}{\pd t_0})}1.
\ee

When restricted to the $t_0$-line, i.e., take $t_n = 0 $ for $n \geq 1$,

\ben
&& \frac{\pd F}{\pd t_0}(t_0) = \frac{t_0}{\lambda^2}, \\
&& \frac{\pd F}{\pd t_1}(t_0) = \frac{1}{2!} \biggl(\frac{t_0^2}{\lambda^2} + 1 \biggr), \\
&& \frac{\pd F}{\pd t_2}(t_0) = \frac{1}{3!} \biggl( \frac{t_0^3}{\lambda^2} + 3 t_0 \biggr), \\
&& \frac{\pd F}{\pd t_3}(t_0) = \frac{1}{4!} \biggl( \frac{1}{\lambda^2}t_0^4+6t_0^2
+3\lambda^2\biggr), \\
&& \frac{\pd F}{\pd t_4}(t_0) = \frac{1}{5!} \biggl( \frac{1}{\lambda^2}t_0^5
+10t_0^3+15\lambda^2t_0 \biggr).
\een
These are essentially Hermite polynomials!

The coefficients turns out to give the triangle of Bessel numbers
\be
T(n,k)= \frac{n! }{(n-2k)!k!2^k},
\ee
they have the following exponential generating series:
\be \label{eqn:Bessel-Gen}
\sum_{n, k}  \frac{1}{n!} T_{n,k} t^n z^k = \exp (z+tz^2/2).
\ee
In fact,
we have

\begin{thm}
The following formula holds:
\be
\frac{\pd F}{\pd t_n}(t_0) = \frac{1}{(n+1)!}\frac{1}{\lambda^2}
\sum_{k=0}^{[n/2]} T(n+1,k) t_0^{n+1-2k}\lambda^{2k}.
\ee
\end{thm}

\begin{proof}
By \eqref{eqn:Pd-F-Gen} and $\frac{\pd F}{\pd t_0}(t_0) = \frac{t_0}{\lambda^2}$,
\ben
1+ \lambda^2 \sum_{n \geq 1}  w^n  \frac{\pd F}{\pd t_{n-1}}(t_0)
& = & e^{\lambda^2 w(\frac{\pd}{\pd t_0} + \frac{t_0}{\lambda^2})}1
= e^{\half w^2\lambda^2} e^{w  t_0 }e^{w \lambda^2 \frac{\pd}{\pd t_0}}1 \\
& = & e^{\half w^2\lambda^2 + w t_0}.
\een
In the above we have used the Campbell-Baker-Hausdorff formula:
\ben
e^Xe^Y = \exp (X + Y + \frac{1}{2} [X, Y ] + \frac{1}{12} [X, [X, Y ]]
- \frac{1}{12} [Y, [X, Y ]] + \cdots)
\een
for $X = \frac{t_0}{\lambda^2}$, $Y = w\frac{\pd}{\pd t_0}$.
The proof is completed by \eqref{eqn:Bessel-Gen}.
\end{proof}

\begin{cor}
When restricted to the $t_0$-line,
the one-point function of topological 1D gravity is given by:
\be \label{eqn:W-1-t0}
W_1(z;t_0) 
= \sum_{g=0}^\infty \frac{(2g-1)!!}{(z-t_0)^{2g+1}} \lambda^{2g}.
\ee
\end{cor}

Since the generating series of the double factorials can be written as a continued fraction:
\be
\sum_{n \geq 0} (2n-1)!! x^n
= 1/(1-x/(1-2x/(1-3x/(1-4x/(1-\cdots,
\ee
one can express the genus zero one-point function restricted to the $t_0$-line by continued fraction:
\be
\begin{split}
& W_1(z;t_0) \\
= & \frac{1}{z-t_0} \cdot 1/(1-\frac{\lambda^2}{(z-t_0)^2}/(1-2\frac{\lambda^2}{(z-t_0)^2}/
(1-3\frac{\lambda^2}{(z-t_0)^2}/(1-\cdots
\end{split}
\ee

As a strange coincidence,
the Bessel numbers also appeared in the author's study of topological 2D gravity.

\begin{lem}
The following identity holds:
\be
\exp(t \frac{\pd}{\pd x}) \exp (f(x)) = \exp \sum_{n \geq 0} \frac{t^n}{n!} \frac{\pd^nf(x)}{\pd x^n}.
\ee
\end{lem}

\begin{proof}
We first check that both sides of the identity satisfies the same differential equation:
\be
\frac{\pd}{\pd t} G(t,x) = \frac{\pd}{\pd x} G(t,x),
\ee
and they have the same initial value $G(0,x) = e^{f(x)}$.
Write
$$G(t, x) = \sum_{n \geq 0} G_n(x) \frac{t^n}{n!},$$
then above differential equation is equivalent to the following recursion relations:
\be
G_n(x) = \frac{\pd}{\pd x} G_{n-1}(x).
\ee
This completes the proof.
\end{proof}

\begin{thm}
For topological 1D gravity,
we have
\be \label{eqn:Hat(W)-1}
\hat{W}_1(w;\bt) 
= \exp \sum_{n=1}^\infty \frac{w^n \lambda^{2n}}{n!} \frac{\pd^n F}{\pd t_0^n}
\ee
and
\be \label{eqn:W-1}
W_1(z;\bt) 
= \frac{1}{z} \sum_{m_1, \dots, m_n \geq 0} \frac{(\sum_{j=1}^n jm_j)!}{\prod_{j=1}^n m_j!} \prod_{j=1}^n
\biggl(\frac{  \lambda^{2j}}{z^jj!} \frac{\pd^j F}{\pd t_0^j}\biggr)^{m_j} .
\ee
\end{thm}

\begin{proof}
Note
\ben
\frac{\pd F}{\pd t_{n-1}}
= \frac{1}{Z} \frac{\pd Z}{\pd t_{n-1}} = \frac{1}{Z} \cdot \frac{\lambda^{2n-2}}{n!} \frac{\pd^n Z}{\pd t_0^n},
\een
therefore,
\ben
\hat{W}_1(w) & = & 1 + \sum_{n \geq 1} w^n \lambda^2 \frac{\pd F}{\pd t_{n-1}}
= 1 + \frac{1}{Z} \sum_{n \geq 1} \frac{w^n\lambda^{2n}}{n!} \frac{\pd^n Z}{\pd t_0^n} \\
& = & \exp(- F) \cdot \exp (w\lambda^2\frac{\pd}{\pd t_0}) \exp (F) \\
& = & \exp \sum_{n=1}^\infty \frac{w^n\lambda^{2n}}{n!} \frac{\pd^nF}{\pd t_0^n}.
\een
This proves \eqref{eqn:Hat(W)-1}.
By expanding the right-hand side of \eqref{eqn:Hat(W)-1} as a power series in $w$ and take
the Laplace transform,
one gets \eqref{eqn:W-1}.
\end{proof}

\begin{rmk}
Note \eqref{eqn:Hat(W)-1} and \eqref{eqn:W-1} can be rewritten in terms of loop operators as follows:
\bea
&& 1 + \hat{B}(w) F =  \exp \sum_{n=1}^\infty \frac{w^n \lambda^{2n}}{n!} \frac{\pd^n F}{\pd t_0^n}\\
&& \frac{1}{z} + B(z) F = \frac{1}{z} \sum_{m_1, \dots, m_n \geq 0} \frac{(\sum_{j=1}^n jm_j)!}{\prod_{j=1}^n m_j!} \prod_{j=1}^n
\biggl(\frac{  \lambda^{2j}}{z^jj!} \frac{\pd^j F}{\pd t_0^j}\biggr)^{m_j} .
\eea
\end{rmk}

\subsection{Genus zero two-point functions}

When $n_1, n_2 \geq 0$,
\ben
\frac{\pd^2 F_0}{\pd t_{n_1}\pd t_{n_2}}
& = & \frac{\pd}{\pd t_{n_1}} \frac{1}{(n_2+1)!} \biggl( \frac{\pd F_0}{\pd t_{0}} \biggr)^{n_1+1} \\
& = & \frac{\pd}{\pd t_{n_1}} \frac{1}{n_2!} \biggl( \frac{\pd F_0}{\pd t_{0}} \biggr)^{n_1}
\cdot \frac{\pd}{\pd t_0} \frac{\pd F}{\pd t_{n_1}} \\
& = & \frac{1}{n_1!n_2!}
\biggl(\frac{\pd F_0}{\pd t_0}\biggr)^{n_1+n_2} \cdot \frac{\pd^2F_0}{\pd t_0^2} \\
& = & \frac{1}{n_1!n_2!(n_1+n_2+1)} \frac{\pd}{\pd t_0}
\biggl(\frac{\pd F_0}{\pd t_0}\biggr)^{n_1+n_2+1}.
\een
It follows that
\be
\sum_{n_1, n_2 \geq 0} z_1^{n_1}z_2^{n_2} \frac{\pd^2 F_0}{\pd t_{n_1}\pd t_{n_2} }
= \frac{\pd^2F_0}{\pd t_0^2} \cdot
e^{(z_1+z_2)\frac{\pd F_0}{\pd t_0}}
= \frac{\pd }{\pd t_0 }
\frac{e^{(z_1+z_2)\frac{\pd F_0}{\pd t_0}}}{z_1+z_2},
\ee
and
\ben
W_{0,2}(z_1,z_2)
& = & \sum_{n_1, n_2 \geq 0} \frac{n_1+1}{z^{n_1+2}} \frac{n_2+1}{z_2^{n_2+2}}
\biggl(\frac{\pd F_0}{\pd t_0}\biggr)^{n_1+n_2} \cdot \frac{\pd^2F_0}{\pd t_0^2} \\
& = & \frac{1}{(z_1-\frac{\pd F_0}{\pd t_0})^2}
\frac{1}{(z_2- \frac{\pd F_0}{\pd t_0})^2}
\frac{\pd^2 F_0}{\pd t_0^2} \\
& = & \pd_{t_0} W_{0,1}(z_1) \cdot \pd_{t_0} W_{0,1}(z_2) \cdot \frac{1}{\frac{\pd^2 F_0}{\pd t_0^2}}.
\een
In particular,
when restricted to the $t_0$-line,
\be
\frac{\pd^2 F_0}{\pd t_{n_1}\pd t_{n_2}}(t_0)
= \frac{t_0^{n_1+n_2}}{n_1!n_2!},
\ee
and so we have
\be
\hat{W}_{0,2}(w_1, w_2; t_0)
 = e^{t_0(z_1+z_2)},
\ee
and
\be
W_{0,2}(z_1,z_2;t_0)
=  \frac{1}{(z_1-t_0)^2(z_2-t_0)^2}.
\ee
We also have
\be
\sum_{n_1, n_2 \geq 0} \frac{n_1!}{z_1^{n_1+1}} \frac{n_2!}{z_2^{n_2+1} }
\frac{\pd^2 F_0}{\pd t_{n_1}\pd t_{n_2} }(t_0)
=  \frac{1}{(z_1-t_0)(z_2-t_0)}.
\ee

\subsection{Two-point function in arbitrary genera}

\begin{thm}
The two-point function $\hat{W}_2(w_1, w_2;\bt)$ of the topological 1D gravity is given by the following formula:
\be \label{eqn:Pd-F-Two-Point}
\begin{split}
& \lambda^2 \hat{W}_{0,2}(w_1,w_2;\bt)
=  \exp \biggl(\sum_{n=1}^\infty \frac{(w_1+w_2)^n \lambda^{2n}}{n!} \frac{\pd^n F}{\pd t_0^n} \biggr) \\
- & \exp \biggl(\sum_{n=1}^\infty \frac{w_2^n \lambda^{2n}}{n!} \frac{\pd^n F}{\pd t_0^n} \biggr)
\cdot
\exp  \biggl(\sum_{n=1}^\infty \frac{w_1^n \lambda^{2n}}{n!} \frac{\pd^n F}{\pd t_0^n} \biggr).
\end{split}
\ee
\end{thm}

\begin{proof}
This can be proved by a direct calculation as follows:
\ben
&& \lambda^2 \hat{W}(w_1,w_2;\bt)  \\
& = & \lambda^4 \sum_{n_1, n_2 \geq 1} w_1^{n_1} w_2^{n_2}  \frac{\pd^2 F}{\pd t_{n_1-1}\pd t_{n_2-1}} \\
& = &  \lambda^2 \sum_{n_1 \geq 1} w_1^{n_1}  \frac{\pd}{\pd t_{n_1-1}}
\exp \sum_{n=1}^\infty \frac{w_2^n \lambda^{2n}}{n!} \frac{\pd^n F}{\pd t_0^n} \\
& = &  \exp \biggl(\sum_{n=1}^\infty \frac{w_2^n \lambda^{2n}}{n!} \frac{\pd^n F}{\pd t_0^n} \biggr)
\cdot \lambda^2 \sum_{n_1 \geq 1} w_1^{n_1}  \frac{\pd}{\pd t_{n_1-1}}
\sum_{n=1}^\infty \frac{w_2^n \lambda^{2n}}{n!} \frac{\pd^n F}{\pd t_0^n} \\
& = & \exp \biggl(\sum_{n=1}^\infty \frac{w_2^n \lambda^{2n}}{n!} \frac{\pd^n F}{\pd t_0^n} \biggr)
\cdot \lambda^2 \sum_{n_1 \geq 1} w_1^{n_1}  \frac{\pd}{\pd t_{n_1-1}}
\sum_{n=1}^\infty \frac{w_2^n \lambda^{2n}}{n!} \frac{\pd^n F}{\pd t_0^n} \\
& = & \exp \biggl(\sum_{n=1}^\infty \frac{w_2^n \lambda^{2n}}{n!} \frac{\pd^n F}{\pd t_0^n} \biggr)
\cdot  \sum_{n \geq 1}
 \frac{w_2^n \lambda^{2n}}{n!} \frac{\pd^n }{\pd t_0^n}
 \biggl( \sum_{n_1 \geq 1} \frac{\lambda^2 w_1^{n_1}}{n_1!}  \frac{\pd F}{\pd t_{n_1-1}} \biggr) \\
& = & \exp \biggl(\sum_{n=1}^\infty \frac{w_2^n \lambda^{2n}}{n!} \frac{\pd^n F}{\pd t_0^n} \biggr)
\cdot  \sum_{n \geq 1}
 \frac{w_2^n \lambda^{2n}}{n!} \frac{\pd^n }{\pd t_0^n}
 \exp  \biggl(\sum_{n=1}^\infty \frac{w_1^n \lambda^{2n}}{n!} \frac{\pd^n F}{\pd t_0^n} \biggr) \\
& = & \exp \biggl(\sum_{n=1}^\infty \frac{w_2^n \lambda^{2n}}{n!} \frac{\pd^n F}{\pd t_0^n} \biggr)
\cdot
\exp  \biggl(\sum_{n=1}^\infty \frac{w_1^n \lambda^{2n}}{n!} \frac{\pd^n F}{\pd t_0^n} \biggr) \\
&& \cdot \biggl( \exp \biggl( \sum_{m,n=1}^\infty \frac{w_1^n \lambda^{2n}}{n!}
\frac{w_2^m \lambda^{2m}}{m!} \frac{\pd^{m+n} F}{\pd t_0^{m+n}} \biggr) - 1 \biggr) \\
& = & \exp \biggl(\sum_{n=1}^\infty \frac{(w_1+w_2)^n \lambda^{2n}}{n!} \frac{\pd^n F}{\pd t_0^n} \biggr) \\
& - & \exp \biggl(\sum_{n=1}^\infty \frac{w_2^n \lambda^{2n}}{n!} \frac{\pd^n F}{\pd t_0^n} \biggr)
\cdot
\exp  \biggl(\sum_{n=1}^\infty \frac{w_1^n \lambda^{2n}}{n!} \frac{\pd^n F}{\pd t_0^n} \biggr).
\een
\end{proof}

\begin{rmk}
A more conceptual way to rewrite \eqref{eqn:Pd-F-Two-Point} is as follows:
\be
\lambda^2 \hat{W}_{0,2}(w_1,w_2;\bt)
= \hat{W}_{0,1}(w_1+w_2;\bt) - \hat{W}_{0,1}(w_1;\bt) \cdot \hat{W}_{0,1}(w_2;\bt).
\ee
It will be useful to note this is equivalent to
\be \label{eqn:Hat(B)-Hat(W)}
\lambda^2 \hat{B}(w_1) \hat{W}_{0,1}(z_2;\bt)
= \hat{W}_{0,1}(w_1+w_2;\bt) - \hat{W}_{0,1}(w_1;\bt) \cdot \hat{W}_{0,1}(w_2;\bt).
\ee
\end{rmk}

\begin{cor}
When restricted to the $t_0$-line,
the two-point function  $\hat{W}_2(w_1,w_2,\bt)$ is given by the following formula:
\be \label{eqn:Pd-F-Two-Point-t0}
\begin{split}
\lambda^2 \hat{W}(w_1,w_2;t_0)
= & \exp \biggl((w_1+w_2)t_0 + \frac{ (w_1+w_2)^2 \lambda^{2}}{2!}  \biggr) \\
- & \exp \biggl(w_1 t_0 + \frac{w_1^2 \lambda^{2}}{2!}  \biggr)
\cdot
\exp  \biggl(w_2t_0 + \frac{w_2^2 \lambda^{2}}{2!}  \biggr).
\end{split}
\ee
\end{cor}

One can compute $W_2(z_1,z_2;\bt)$ by Laplace transform the formula for $\hat{W}_2(w_1, w_2;\bt)$.
Here we will use a more concptual method by applying the loop operator $B(z_1)$ on $W_1(z_2; \bt)$.

\begin{thm}
The two-point function $W_2(z_1, z_2;\bt)$ of the topological 1D gravity is given by the following formula:
\be \label{eqn:W-2Point}
\lambda^2 W_2(z_1, z_2) = \sum_{k\geq 1} \frac{(-1)^k\lambda^{2k}}{k!} \frac{\pd^k}{\pd z^k} W_1(z_2;\bt)
\cdot  \frac{\pd^k}{\pd t_0^k} W_1(z_1;\bt).
\ee
\end{thm}

\begin{proof}
Note $B(z_1)$ is a derivation,
therefore,
by \eqref{eqn:W-1},
\ben
&& W_2(z_1, z_2) = B(z_1) W_1(z_2;\bt)  \\
& = & B(z_1)  \frac{1}{z_2} \sum_{m_1, \dots, m_n \geq 0} \frac{(\sum_{j=1}^n jm_j)!}{\prod_{j=1}^n m_j!} \prod_{j=1}^n
\biggl(\frac{  \lambda^{2j}}{z_2^jj!} \frac{\pd^j F}{\pd t_0^j}\biggr)^{m_j} \\
& = &  \frac{1}{z_2} \sum_{m_1, \dots, m_n \geq 0} \frac{(\sum_{j=1}^n jm_j)!}{\prod_{j=1}^n m_j!}
\sum_{k=1}^n \prod_{j=1}^n \biggl(\frac{  \lambda^{2j}}{z_2^jj!} \frac{\pd^j F}{\pd t_0^j}\biggr)^{m_j-\delta_{j,k}} \\
&& \cdot m_k \cdot \frac{\pd^k}{\pd t_0^k} B(z_1) F \\
& = & \sum_{k\geq 1} \frac{(-1)^k\lambda^{2k}}{k!} \frac{\pd^k}{\pd z^k}
\biggl(  \frac{1}{z_2} \sum_{m_1, \dots, m_n \geq 0} \frac{(\sum_{j=1}^n jm_j)!}{\prod_{j=1}^n m_j!} \prod_{j=1}^n
\biggl(\frac{  \lambda^{2j}}{z_2^jj!} \frac{\pd^j F}{\pd t_0^j}\biggr)^{m_j} \biggr) \\
&& \cdot  \frac{\pd^k}{\pd t_0^k} B(z_1) F \\
& = & \lambda^{-2} \sum_{k\geq 1} \frac{(-1)^k\lambda^{2k}}{k!} \frac{\pd^k}{\pd z^k} W_1(z_2;\bt)
\cdot  \frac{\pd^k}{\pd t_0^k} W_1(z_1;\bt).
\een
\end{proof}

\begin{rmk}
Note \eqref{eqn:W-2Point} can be rewritten as follows:
\be \label{eqn:B-W}
B(z_1) W_1(z_2;\bt) = \sum_{k\geq 1} \frac{(-1)^k\lambda^{2k}}{k!} \frac{\pd^k}{\pd z^k} W_1(z_2;\bt)
\cdot  \frac{\pd^k}{\pd t_0^k} W_1(z_1;\bt).
\ee
One can also symmetrize this:
\be \label{eqn:B-W-Symm}
\begin{split}
B(z_1) W_1(z_2;\bt)
 =  \frac{1}{2} \sum_{k\geq 1} \frac{(-1)^k\lambda^{2k}}{k!}&
\biggl(\frac{\pd^k}{\pd z^k} W_1(z_1;\bt) \cdot  \frac{\pd^k}{\pd t_0^k} W_1(z_2;\bt) \\
+ & \frac{\pd^k}{\pd z^k} W_1(z_2;\bt) \cdot  \frac{\pd^k}{\pd t_0^k} W_1(z_1;\bt) \biggr).
\end{split}
\ee
\end{rmk}

\begin{cor}
When restricted to the $t_0$-line,
the two-point function  $W_2(z_1,z_2,\bt)$ is given by the following formula:
\be
\begin{split}
W(z_1,z_2;t_0)
= \sum_{k\geq 1} \frac{\lambda^{2k}}{k!}
& \sum_{g=0}^\infty \frac{(2g-1)!!\prod_{j=0}^{k-1}(2g+1+j)}{(z_1-t_0)^{2g+1+k}} \lambda^{2g} \\
\cdot & \sum_{g=0}^\infty \frac{(2g-1)!!\prod_{j=0}^{k-1}(2g+1+j)}{(z_2-t_0)^{2g+1+k}} \lambda^{2g}.
\end{split}
\ee
\end{cor}

\subsection{Genus zero $l$-point functions}

Inductively we find for $l \ geq 3$:
\ben
\frac{\pd^l F_0}{\pd t_{n_1}\cdots \pd t_{n_l}}
& = & \frac{1}{n_1!\cdots n_l!(n_1+\cdots + n_l +1)} \frac{\pd^{l-1}}{\pd t_0^{l-1}}
\biggl(\frac{\pd F_0}{\pd t_0}\biggr)^{n_1+\cdots + n_l+1} \\
& = & \frac{1}{n_1!\cdots n_l!} \frac{\pd^{l-2}}{\pd t_0^{l-2}} \biggl[
\biggl(\frac{\pd F_0}{\pd t_0}\biggr)^{n_1+\cdots + n_l}
\cdot \frac{\pd^2F_0}{\pd t_0^2} \biggr].
\een
It follows that
\be
\hat{W}_{0,l}(w_1,\dots, w_l;\bt)
= \frac{\pd^{l-1} }{\pd t_0^{l-1} }
\frac{e^{\frac{\pd F_0}{\pd t_0} \sum_{j=1}^l w_j}}{\sum_{j=1}^l w_j},
\ee
and
\be \label{eqn:W-l-Point}
W_{0,l}(z_1,\dots, z_l;\bt)
=\frac{\pd^{l-2}}{\pd t_0^{l-2}} \biggl( \frac{1}{(z_1-\frac{\pd F_0}{\pd t_0})^2}\cdots
\frac{1}{(z_l- \frac{\pd F_0}{\pd t_0})^2}
\frac{\pd^2 F_0}{\pd t_0^2} \biggr).
\ee
In particular,
when restricted to the $t_0$-line,
\be \label{eqn:Pd-F0-l-Point}
\frac{\pd^l F_0}{\pd t_{n_1} \cdots \pd t_{n_l}}(t_0)
= \frac{\pd^{l-2}}{\pd t_0^{l-2}}  \frac{t_0^{n_1+\cdots + n_l}}{n_1!\cdots n_l!},
\ee
and so
\be
\hat{W}_{0,l}(w_1, \dots, w_l; t_0)
=  e^{t_0 \sum_{j=1}^l wz_j}(\sum_{j=1}^l w_j)^{l-2}.
\ee
and
\be
W_{0,l}(z_1, \dots, z_l; t_0)
= \frac{\pd^{l-2}}{\pd t_0^{l-2}}  \frac{1}{\prod_{j=1}^l (z_j-t_0)^2}.
\ee
We also have
\be
\sum_{n_1, \dots n_l \geq 0} \prod_{j=1}^l \frac{n_j!}{z_j^{n_j+1}} \cdot
\frac{\pd^l F_0}{\pd t_{n_1}\dots \pd t_{n_l} }(t_0)
= \frac{\pd^{l-2}}{\pd t_0^{l-2}}  \frac{1}{\prod_{j=1}^l (z_j-t_0)}.
\ee

By \eqref{eqn:Pd-F0-l-Point}, one can derive the following formula:
\be
\sum_{n_1, \dots, n_l \geq 0} \corr{\tau_{n_1} \cdots \tau_{n_l}}_0 \cdot \prod_{j=1}^l x_j^{n_j}
= (x_1+ \cdots + x_l)^{l-2}.
\ee

\subsection{General $l$-point functions in arbitrary genera}

We now apply \eqref{eqn:Hat(W)=Hat(B)Hat(W)} and \eqref{eqn:Hat(B)-Hat(W)}
repeatedly to compute  $\hat{W}_l$.
Let us fix some notations.
By $[l]$ we mean the set of indices $\{1, \dots, l\}$.
By $I_1 \coprod \cdot \coprod I_k = [l]$
we mean a partition of $[l]$ into disjoint union of $k$ nonempty subsets
$I_1, \dots, I_k$.
Given such a partition,
for $j=1, \dots, k$, $|I_j|$ denotes the number of elements in $I_j$,
and we define
\ben
w_{I_j} = \sum_{i\in I_j} w_i.
\een

\begin{thm}
The $l$-point function $\hat{W}_l(w_1, \dots, w_l;\bt)$ of the topological 1D gravity is given by the following formula:
\be \label{eqn:W-l-Point}
\begin{split}
& \lambda^{2l-2} \hat{W}(w_1,\dots, w_l;\bt)
=  \sum_{I_1 \coprod \cdots \coprod I_k = [l]} (-1)^{k-1} k!
\prod_{j=1}^k W_{|I|}(w_{I_j};\bt).
\end{split}
\ee
\end{thm}

\begin{proof}
This can be easily proved by induction on $l$.
\end{proof}

Similarly,
one can apply \eqref{eqn:W=BW} and \eqref{eqn:B-W}
repeatedly to compute  $W_l$.
For example, from
\ben
\lambda^2 W_2(z_2, z_3;\bt)
& = &  \sum_{k\geq 1} \frac{(-1)^k\lambda^{2k}}{k!} \frac{\pd^k}{\pd z^k} W_1(z_2;\bt)
\cdot  \frac{\pd^k}{\pd t_0^k} W_1(z_3;\bt),
\een
one gets
\ben
&& \lambda^4 W_3(z_1, z_2, z_3; \bt) \\
& = &  \lambda^2 B(z_1) \sum_{k\geq 1} \frac{(-1)^k\lambda^{2k}}{k!} \frac{\pd^k}{\pd z^k} W_1(z_2;\bt)
\cdot  \frac{\pd^k}{\pd t_0^k} W_1(z_3;\bt) \\
& = & \sum_{k\geq 1} \frac{(-1)^k\lambda^{2k}}{k!} \frac{\pd^k}{\pd z^k} \lambda^2 B(z_1)W_1(z_2;\bt)
\cdot  \frac{\pd^k}{\pd t_0^k} W_1(z_3;\bt) \\
& + & \sum_{k\geq 1} \frac{(-1)^k\lambda^{2k}}{k!} \frac{\pd^k}{\pd z^k} W_1(z_2;\bt)
\cdot  \frac{\pd^k}{\pd t_0^k} \lambda^2 B(z_1) W_1(z_3;\bt) \\
& = & \sum_{k\geq 1} \frac{(-1)^k\lambda^{2k}}{k!} \frac{\pd^k}{\pd z^k} \lambda^2 W_2(z_1, z_2;\bt)
\cdot  \frac{\pd^k}{\pd t_0^k} W_1(z_3;\bt) \\
& + & \sum_{k\geq 1} \frac{(-1)^k\lambda^{2k}}{k!} \frac{\pd^k}{\pd z^k} W_1(z_2;\bt)
\cdot  \frac{\pd^k}{\pd t_0^k} \lambda^2 W_2(z_1,z_3;\bt).
\een
The result is not manifestly symmetric with respect to $z_1, z_2, z_3$,
one can partially symmetrize it with respect to $z_2, z_3$ to get:
\ben
&& \lambda^4 W_3(z_1, z_2, z_3; \bt) \\
& = & \frac{1}{2!} \sum_{I\coprod J = [3]_1}
\sum_{k\geq 1} \frac{(-1)^k\lambda^{2k}}{k!}
\biggl(\frac{\pd^k}{\pd z^k} \lambda^2 W_2(z_1, z_I;\bt)
\cdot  \frac{\pd^k}{\pd t_0^k} W_1(z_J;\bt) \\
& + & \frac{\pd^k}{\pd z^k} W_1(z_I;\bt)
\cdot  \frac{\pd^k}{\pd t_0^k} \lambda^2 W_2(z_1,z_J;\bt) \biggr).
\een
Here we use the following notation: For $1 \leq i \leq n$,
$[n]_i = \{1, \dots, n\} - \{i\}$.
One can also symmetrize with respect to all three variables to get:
\ben
&& \lambda^4 W_3(z_1, z_2, z_3; \bt) \\
& = & \frac{1}{3!} \sum_{I \coprod J =[3]} \sum_{k\geq 1} \frac{(-1)^k\lambda^{2k}}{k!}
\frac{\pd^k}{\pd z^k}|I|! W_{|I|}(z_I;\bt)
\cdot  \frac{\pd^k}{\pd t_0^k} |J|!W_{|J|}(z_J;\bt) .
\een
By induction we get the following:

\begin{thm}
The $l$-point functions $W_l$ of the topological 1D gravity satisfies the following
recursion relations:

\ben
&& \lambda^{2l-2} W_l(z_1, \dots, z_l; \bt) \\
& = & \frac{1}{l!} \sum_{I \coprod J =[l]} \sum_{k\geq 1} \frac{(-1)^k\lambda^{2k}}{k!}
\binom{l-2}{|I|-1} \\
&& \cdot \frac{\pd^k}{\pd z^k}|I|! W_{|I|}(z_I;\bt)
\cdot  \frac{\pd^k}{\pd t_0^k} |J|!W_{|J|}(z_J;\bt)
\een
and
\ben
&& \lambda^{2l-2} W_l(z_1, \dots, z_l; \bt) \\
& = & \frac{1}{(l-1)!} \sum_{I \coprod J =[l]_1}
\sum_{k\geq 1} \frac{(-1)^k\lambda^{2k}}{k!}
\binom{l-3}{|I|-1} \\
&& \cdot \biggl(
\frac{\pd^k}{\pd z^k}|I|! W_{|I|}(z_1, z_I;\bt)
\cdot  \frac{\pd^k}{\pd t_0^k} |J|!W_{|J|}(z_J;\bt) \\
&& + \frac{\pd^k}{\pd z^k}|I|! W_{|I|}(z_I;\bt)
\cdot  \frac{\pd^k}{\pd t_0^k} |J|!W_{|J|}(z_1, z_J;\bt) \biggr).
\een
\end{thm}

\section{Feynman Rules for $N$-Point Functions}

\label{sec:Feynman for N-Point}

In this Section we discuss the Feynman rules for $n$-point functions
in topological 1D gravity.

\subsection{Feynman rules for genus zero $n$-point functions}
We have already shown that
\be
W_{0,n}(z_1, \dots, z_n; \bt)
= \frac{\pd^{n-2}}{\pd t_0^{n-2}} \biggl[
\frac{1}{(z_1-\frac{\pd F_0}{\pd t_0})^2} \cdots  \frac{1}{(z_n-\frac{\pd F_0}{\pd t_0})^2}
\cdot \frac{\pd^2F_0}{\pd t_0^2} \biggr].
\ee
After taking $\frac{\pd^{n-2}}{\pd t_0^{n-2}}$ on the right-hand side,
we get a polynomial in
\ben
&& \frac{j!}{(z_i-\frac{\pd F_0}{\pd t_0})^{j+1}}, \;\;\; i=1, \dots, n, \;\; j \geq 2,  \\
&& \frac{\pd^kF_0}{\pd t_0^k}, \;\;\; k \geq 2.
\een
The first few examples are:
\ben
W_{0,1}(z_1;\bt) = \frac{1}{z_1-\frac{\pd F_0}{\pd t_0}},
\een
\ben
W_{0,2}(z_1,z_2;\bt) = \frac{1}{(z_1-\frac{\pd F_0}{\pd t_0})^2}
\cdot \frac{1}{(z_2-\frac{\pd F_0}{\pd t_0})^2} \cdot \frac{\pd^2F_0}{\pd t_0^2},
\een
\ben
W_{0,3}(z_1, z_2, z_3;\bt)
& = & \biggl(\frac{2}{(z_1-\frac{\pd F_0}{\pd t_0})^3} \cdot   \frac{1}{(z_2-\frac{\pd F_0}{\pd t_0})^2}
 \cdot   \frac{1}{(z_3-\frac{\pd F_0}{\pd t_0})^2} \\
& + & \frac{1}{(z_1-\frac{\pd F_0}{\pd t_0})^2} \cdot   \frac{2}{(z_2-\frac{\pd F_0}{\pd t_0})^3}
 \cdot   \frac{1}{(z_3-\frac{\pd F_0}{\pd t_0})^2} \\
& + & \frac{1}{(z_1-\frac{\pd F_0}{\pd t_0})^2} \cdot   \frac{1}{(z_2-\frac{\pd F_0}{\pd t_0})^2}
 \cdot   \frac{2}{(z_3-\frac{\pd F_0}{\pd t_0})^3} \biggr)
\cdot \biggl(\frac{\pd^2F_0}{\pd t_0^2} \biggr)^2 \\
& + & \frac{1}{(z_1-\frac{\pd F_0}{\pd t_0})^2} \cdot   \frac{1}{(z_2-\frac{\pd F_0}{\pd t_0})^2}
 \cdot   \frac{1}{(z_3-\frac{\pd F_0}{\pd t_0})^2}
\cdot \frac{\pd^3F_0}{\pd t_0^3}
\een
By examining these examples we observe that one can
interpret the terms  by some Feynman rules.
$W_{0,1}(z_1;\bt)$ corresponds to the following graph:
$$
\xy
(0,0); (10,0), **@{.}; (0, 0)*+{\bullet}; (10,0)*+{\circ}; (12,-2)*+{z_1};
\endxy
$$
$W_{0,2}(z_1,z_2;\bt)$ corresponds to the following graph:
$$
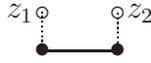
\xy
(0,5); (0,0), **@{.};  (10,0), **@{-}; (10,5), **@{.};
(0, 0)*+{\bullet}; (0,5)*+{\circ};  (10, 0)*+{\bullet}; (10, 5)*+{\circ};
(-3,5)*+{z_1};   (13,5)*+{z_2};
\endxy
$$
The four terms in $W_{0,3}(z_1,z_2,z_3;\bt)$ correspond to the following four diagrams:
$$
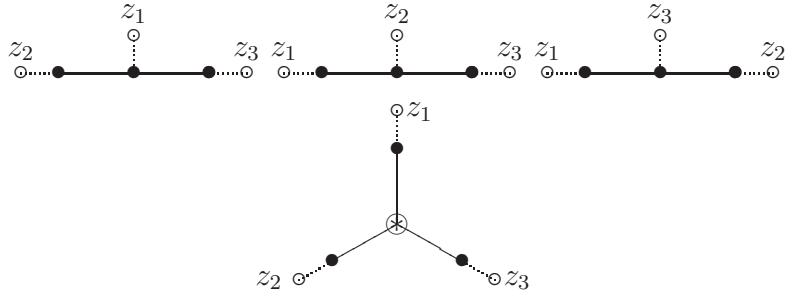
\xy
(0,0); (5,0), **@{.}; (15,0), **@{-}; (25,0), **@{-}; (30,0), **@{.};
(15, 0)*+{\bullet}; (25, 0)*+{\bullet}; (30,0)*+{\circ};
(15,0);     (15,5), **@{.};
(0,0)*+{\circ}; (5, 0)*+{\bullet};  (25, 0)*+{\bullet};    (15,5)*+{\circ};
(0,3)*+{z_2}; (15,8)*+{z_1}; (30, 3)*+{z_3};
(35,0); (40,0), **@{.}; (50,0), **@{-}; (60,0), **@{-}; (65,0), **@{.};
(50, 0)*+{\bullet}; (60, 0)*+{\bullet}; (65,0)*+{\circ};
(50,0);     (50,5), **@{.};
(35,0)*+{\circ}; (40, 0)*+{\bullet};  (60, 0)*+{\bullet};    (50,5)*+{\circ};
(35,3)*+{z_1}; (50,8)*+{z_2}; (65, 3)*+{z_3};
(70,0); (75,0), **@{.}; (85,0), **@{-}; (95,0), **@{-}; (100,0), **@{.};
(85, 0)*+{\bullet}; (95, 0)*+{\bullet}; (100,0)*+{\circ};
(85,0);     (85,5), **@{.};
(70,0)*+{\circ}; (75, 0)*+{\bullet};  (95, 0)*+{\bullet};    (85,5)*+{\circ};
(70,3)*+{z_1}; (85,8)*+{z_3}; (100, 3)*+{z_2};
(37,-27.5); (41.34,-25), **@{.}; (50,-20), **@{-}; (50, -10), **@{-}; (50,-5), **@{.};
(50,-20); (58.66,-25), **@{-}; (63,-27.5), **@{.};
(37,-27.5)*+{\circ}; (63,-27.5)*+{\circ}; (50,-5)*+{\circ};
(41.34,-25)*+{\bullet};  (50, -10)*+{\bullet};  (58.66,-25)*+{\bullet};
(50,-20)*+{\circledast};
(33,-27.5)*+{z_2}; (66,-27.5)*+{z_3}; (53,-5)*+{z_1};
\endxy
$$
Based on these examples we formulate the following:

\begin{thm} \label{thm:Feynman Rules-Genus Zero}
The genus zero $n$-point function $W_{0,n}(z_1, \dots, z_n; \bt)$ of topological 1D gravity is given by a summation
over marked trees $\Gamma$:
$$W_{0,n}(z_1, \dots, z_n; \bt)= \sum_\Gamma \frac{w_\Gamma}{|\Aut(\Gamma)|},$$
where $\Gamma$ satisfies the following conditions:
\begin{itemize}
\item[(1)] $\Gamma$ has exactly $n$ vertices marked by $\circ$, and they are all of valence one
and are marked by $z_1, \dots, z_n$ respectively.
\item[(2)] $\Gamma$ has exactly $n$ vertices marked by $\bullet$,
each of which is joined via an edge to a vertex marked by $\circ$.
\item[(3)] There are maybe some vertices of valence $\geq 3$ marked by $\circledast$,
and they can only be joined directly to vertices marked by $\bullet$.
\end{itemize}
We will refer to the above three kinds of vertices as $\circ$-vertices, $\bullet$-vertices
and $\circledast$-vertices respectively.
The weight $w_\Gamma$ is given as usual:
$$
w_\Gamma = \prod_{v\in V(\Gamma)} w_v \cdot \prod_{e \in E(\Gamma)} w_e,
$$
where $w_v$ is given by
$$
w_v = \begin{cases}
1, & \text{if $v$ is a $\circ$-vertex}, \\
\frac{(\val(v)-1)!}{(z_j-\frac{\pd F_0}{\pd t_0})^{\val(v)}}, &
\text{if $v$ is a $\bullet$-vertex joined to a $\circ$-vertex marked by $z_j$}, \\
\frac{\pd^{\val(v)} F_0}{\pd t_0^{\val(v)}}, & \text{if $v$ is a $\circledast$-vertex},
\end{cases}
$$
and $w_e$ is given by:
$$
w_e
= \begin{cases}
1, & \text{if $e$ is incident at a $\circ$-vertex or a $\circledast$-vertex}, \\
\frac{\pd^2F_0}{\pd t_0^2}, & \text{if $e$ joins two $\bullet$-vertices}.
\end{cases}
$$
\end{thm}

To further illustrate the Feynman rules, let us check them for $W_{0,4}$.
An application of Leibniz formula gives us:
\ben
W_{0,4}(z_1, \dots, z_4;\bt)
& = & \sum_{i=1}^4  \frac{3!}{(z_i-\frac{\pd F_0}{\pd t_0})^4} \cdot
\biggl(\frac{\pd^2F_0}{\pd t_0^2} \biggr)^2
  \cdot \prod_{\substack{1 \leq j \leq 4 \\j \neq i}}
\frac{1}{(z_j-\frac{\pd F_0}{\pd t_0})^2} \cdot  \frac{\pd^2F_0}{\pd t_0^2}  \\
& + & 2 \sum_{1\leq i < j \leq 4}  \frac{2!}{(z_i-\frac{\pd F_0}{\pd t_0})^3}
\cdot \frac{\pd^2F_0}{\pd t_0^2}
\cdot \frac{2!}{(z_j-\frac{\pd F_0}{\pd t_0})^3} \cdot \frac{\pd^2F_0}{\pd t_0^2}   \\
&& \cdot \prod_{\substack{1 \leq k \leq 4 \\k \neq i,j}}
\frac{1}{(z_j-\frac{\pd F_0}{\pd t_0})^2} \cdot  \frac{\pd^2F_0}{\pd t_0^2}  \\
& + & 3 \sum_{i=1}^4  \frac{2!}{(z_i-\frac{\pd F_0}{\pd t_0})^3} \cdot
 \frac{\pd^2F_0}{\pd t_0^2}
 \cdot \prod_{\substack{1 \leq j \leq 4 \\j \neq i}}
\frac{1}{(z_j-\frac{\pd F_0}{\pd t_0})^2} \cdot  \frac{\pd^3F_0}{\pd t_0^3}  \\
& + & \prod_{1 \leq j \leq 4 }
\frac{1}{(z_j-\frac{\pd F_0}{\pd t_0})^2} \cdot  \frac{\pd^4F_0}{\pd t_0^4}.
\een
The first summation on the right-hand side  corresponds to  diagrams of
the following form:
$$
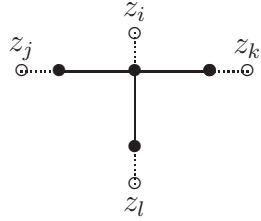
\xy
(0,0); (5,0), **@{.}; (15,0), **@{-}; (25,0), **@{-}; (30,0), **@{.};
(15, 0)*+{\bullet}; (25, 0)*+{\bullet}; (30,0)*+{\circ};
(15,5);     (15,0), **@{.}; (15,-10), **@{-}; (15,-15), **@{.};
(0,0)*+{\circ}; (5, 0)*+{\bullet};  (25, 0)*+{\bullet};    (15,5)*+{\circ};
(0,3)*+{z_j}; (15,8)*+{z_i}; (30, 3)*+{z_k}; (15,-18)*+{z_l};
 (15,-10)*+{\bullet}; (15,-15)*+{\circ};
\endxy
$$
The second summation on the right-hand side  corresponds to  diagrams of
the following form:
$$
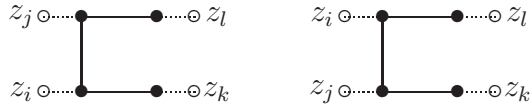
\xy
(0,0); (5,0), **@{.}; (15,0), **@{-}; (20,0), **@{.};
(0,10); (5,10), **@{.}; (15,10), **@{-}; (20,10), **@{.};
(0,0)*+{\circ}; (5, 0)*+{\bullet}; (15, 0)*+{\bullet}; (20,0)*+{\circ};
(0,10)*+{\circ}; (5, 10)*+{\bullet}; (15, 10)*+{\bullet}; (20,10)*+{\circ};
(-3,0)*+{z_i}; (-3,10)*+{z_j}; (23,0)*+{z_k}; (23,10)*+{z_l};
(5,0); (5,10), **@{-};
(40,0); (45,0), **@{.}; (55,0), **@{-}; (60,0), **@{.};
(40,10); (45,10), **@{.}; (55,10), **@{-}; (60,10), **@{.};
(40,0)*+{\circ}; (45, 0)*+{\bullet}; (55, 0)*+{\bullet}; (60,0)*+{\circ};
(40,10)*+{\circ}; (45, 10)*+{\bullet}; (55, 10)*+{\bullet}; (60,10)*+{\circ};
(45,0); (45,10), **@{-};
(37,10)*+{z_i}; (37,0)*+{z_j}; (63,0)*+{z_k}; (63,10)*+{z_l};
\endxy
$$
Note the appearance of the factor $2$ comes from the fact that the above two diagrams have the same
contributions.
The third summation on the right-hand side  corresponds to  diagrams of
the following form:
$$
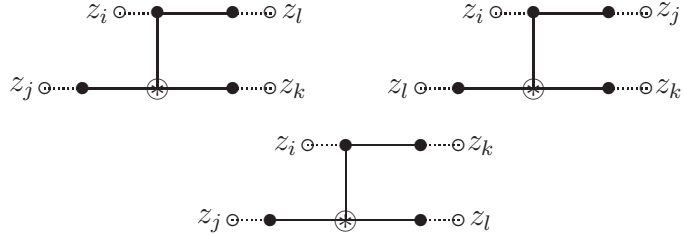
\xy
(-10,0); (-5,0), **@{.}; (5,0), **@{-}; (15,0), **@{-}; (20,0), **@{.};
(0,10); (5,10), **@{.}; (15,10), **@{-}; (20,10), **@{.};
(-10,0)*+{\circ}; (-5, 0)*+{\bullet}; (5, 0)*+{\circledast};(15, 0)*+{\bullet}; (20,0)*+{\circ};
(0,10)*+{\circ}; (5, 10)*+{\bullet}; (15, 10)*+{\bullet}; (20,10)*+{\circ};
(-13,0)*+{z_j}; (-3,10)*+{z_i}; (23,0)*+{z_k}; (23,10)*+{z_l};
(5,0); (5,10), **@{-};
(40,0); (45,0), **@{.}; (55,0), **@{-}; (65,0), **@{-}; (70,0), **@{.};
(50,10); (55,10), **@{.}; (65,10), **@{-}; (70,10), **@{.};
(40,0)*+{\circ}; (45, 0)*+{\bullet}; (55, 0)*+{\circledast};(65, 0)*+{\bullet}; (70,0)*+{\circ};
(50,10)*+{\circ}; (55, 10)*+{\bullet}; (65, 10)*+{\bullet}; (70,10)*+{\circ};
(55,0); (55,10), **@{-};
(47,10)*+{z_i}; (37,0)*+{z_l}; (73,0)*+{z_k}; (73,10)*+{z_j};
\endxy
$$
$$
\xy
(-10,0); (-5,0), **@{.}; (5,0), **@{-}; (15,0), **@{-}; (20,0), **@{.};
(0,10); (5,10), **@{.}; (15,10), **@{-}; (20,10), **@{.};
(-10,0)*+{\circ}; (-5, 0)*+{\bullet}; (5, 0)*+{\circledast};(15, 0)*+{\bullet}; (20,0)*+{\circ};
(0,10)*+{\circ}; (5, 10)*+{\bullet}; (15, 10)*+{\bullet}; (20,10)*+{\circ};
(-13,0)*+{z_j}; (-3,10)*+{z_i}; (23,0)*+{z_l}; (23,10)*+{z_k};
(5,0); (5,10), **@{-};
\endxy
$$
Note the appearance of the factor $3$ comes from the fact that the above three diagrams
all have the same contributions.
The last term on the right-hand side corresponds to the diagram:
$$
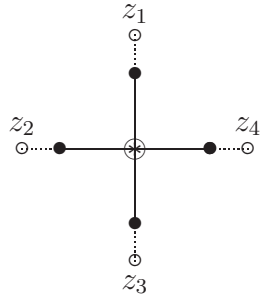
\xy
(0,0); (5,0), **@{.}; (15,0), **@{-}; (25,0), **@{-}; (30,0), **@{.};
(15, 10)*+{\bullet}; (25, 0)*+{\bullet}; (30,0)*+{\circ};
(15,15);     (15,10), **@{.}; (15,-10), **@{-}; (15,-15), **@{.};
(0,0)*+{\circ}; (5, 0)*+{\bullet};  (25, 0)*+{\bullet}; (15,0)*+{\circledast};   (15,15)*+{\circ};
(0,3)*+{z_2}; (15,18)*+{z_1}; (30, 3)*+{z_4}; (15,-18)*+{z_3};
 (15,-10)*+{\bullet}; (15,-15)*+{\circ};
\endxy
$$

{\em Proof of Theorem \ref{thm:Feynman Rules-Genus Zero}}.
We use induction on $n$.
We have seen above that the Theorem holds for $n=1$.
Suppose that it holds for some $n \geq 1$.
Then we apply the loop operator $B(z_{n+1})$ to get:
\ben
W_{0,n+1}(z_1, \dots, z_n; \bt)
= \sum_\Gamma \frac{1}{|\Aut(\Gamma)|} B(z_{n+1})w_\Gamma,
\een
where
$$
B(z_{n+1}) w_\Gamma = B(z_{n+1}) \prod_{v\in V(\Gamma)} w_v \cdot \prod_{e \in E(\Gamma)} w_e.
$$
We will write $B(z_{n+1}) w_\Gamma$ as a summation over Feynman diagrams
obtained by grafting new branch of the form
$
\xy
(0,0); (5,0), **@{.}; (15,0), **@{-};
(0,0)*+{\circ}; (5, 0)*+{\bullet};
(-6,0)*+{z_{n+1}};
\endxy
$
to $\Gamma$ and perhaps at the same time split a $\circledast$-vertex.
Because $B(z_{n+1})$ is a derivation,
we need to consider $B(z_{n+1})w_v$ for all $v\in V(\Gamma)$
and $B(z_{n+1})w_e$ for all $e \in E(\Gamma)$.

Let us consider $B(z_{n+1})w_e$ first.
If $e$ is incident at a $\circ$-vertex or a $\circledast$-vertex,
then $w_e= 1$ and so $B(z_{n+1})w_e = 0$.
This vanishing means there is no grafting on the interior of the edge $e$.

If $e$ joins two $\bullet$-vertices,
then $ w_e = \frac{\pd^2F_0}{\pd t_0^2}$,
and so
\ben
B(z_{n+1})w_e & = & B(z_{n+1}) \frac{\pd^2F_0}{\pd t_0^2}
=  \frac{\pd^2 }{\pd t_0^2}  B(z_{n+1})F_0
=  \frac{\pd^2 }{\pd t_0^2}  \frac{1}{z_{n+1} - \frac{\pd F_0}{\pd t_0}} \\
& = & \frac{1}{(z_{n+1} - \frac{\pd F_0}{\pd t_0})^2} \cdot \frac{\pd^3F_0}{\pd t_0^3}
+ \frac{2}{(z_{n+1} - \frac{\pd F_0}{\pd t_0})^3} \cdot \biggl(\frac{\pd^2F_0}{\pd t_0^2} \biggr)^2.
\een
This can graphically represented as follows:
$$
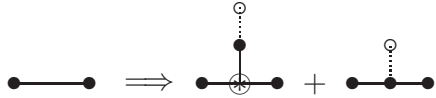
\xy
(0,0); (10,0), **@{-}; (0,0)*+{\bullet}; (10,0)*+{\bullet}; (18,0)*+{\Longrightarrow};
(25,0); (35, 0), **@{-}; (30, 10); (30,5), **@{.}; (30, 0), **@{-}; (40,0)*+{+};
(25,0)*+{\bullet}; (35, 0)*+{\bullet}; (30, 10)*+{\circ}; (30,5)*+{\bullet}; (30, 0)*+{\circledast};
(45,0); (55,0), **@{-}; (50, 0); (50,5), **@{.};
(45,0)*+{\bullet}; (55,0)*+{\bullet}; (50, 0)*+{\bullet}; (50,5)*+{\circ};
\endxy
$$
This means there are two ways to graft a branch of the form
$
\xy
(0,0); (5,0), **@{.}; (15,0), **@{-};
(0,0)*+{\circ}; (5, 0)*+{\bullet};
(-6,0)*+{z_{n+1}};
\endxy
$
to an edge connecting two $\bullet$-vertices.

If $v$ is a $\circ$-vertex of $\Gamma$, then $w_v = 1$
and $B(z_{n+1})w_v = 0$. This vanishing means that there is no grafting at this vertex.
If $v$ is a $\bullet$-vertex of $\Gamma$,
then
\ben
B(z_{n+1}) w_v & = & B(z_{n+1}) \frac{(\val(v)-1)!}{(z_j-\frac{\pd F_0}{\pd t_0})^{\val(v)}} \\
& = & \frac{\val(v)!}{(z_j-\frac{\pd F_0}{\pd t_0})^{\val(v)+1}}\cdot \frac{\pd}{\pd t_0} B(z_{n+1})F_0 \\
& = & \frac{\val(v)!}{(z_j-\frac{\pd F_0}{\pd t_0})^{\val(v)+1}}\cdot \frac{\pd}{\pd t_0} \frac{1}{z_{n+1} - \frac{\pd F_0}{\pd t_0}} \\
& = & \frac{\val(v)!}{(z_j-\frac{\pd F_0}{\pd t_0})^{\val(v)+1}}\cdot \frac{1}{(z_{n+1} - \frac{\pd F_0}{\pd t_0})^2}
\cdot  \frac{\pd^2F_0}{\pd t_0^2}.
\een
Pictorially,
this can represented as follows:
$$
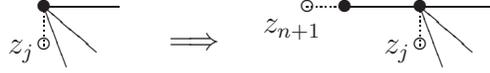
\xy
(5,-5); (5,0), **@{.}; (15,0), **@{-};
(5,0); (8,-8), **@{-}; (5,0); (12,-6), **@{-};
(5,-5)*+{\circ}; (5, 0)*+{\bullet};
(2,-6)*+{z_j};
(25,-5)*+{\Longrightarrow};
(40,0); (45,0), **@{.}; (55, 0), **@{-};
(55,-5); (55,0), **@{.}; (65,0), **@{-};
(55,0); (58,-8), **@{-}; (55,0); (62,-6), **@{-};
(55,-5)*+{\circ}; (55, 0)*+{\bullet};
(52,-6)*+{z_j}; (40,0)*+{\circ}; (45,0)*+{\bullet}; (38,-3)*+{z_{n+1}};
\endxy
$$
If $v$ is a $\circledast$-vertex of $\Gamma$,
then
\ben
B(z_{n+1})w_v & = & B(z_{n+1}) \frac{\pd^{\val(v)} F_0}{\pd t_0^{\val(v)}}
= \frac{\pd^{\val(v)}  }{\pd t_0^{\val(v)}} B(z_{n+1})  F_0 \\
& = & \frac{\pd^{\val(v)}  }{\pd t_0^{\val(v)}} \frac{1}{z_{n+1} - \frac{\pd F_0}{\pd t_0}}.
\een
We will write the derivative $\frac{\pd^m  }{\pd t_0^m} \frac{1}{z_{n+1} - \frac{\pd F_0}{\pd t_0}}$
as a sum over some diagrams.
This can be done inductively as follows.
For $m = 0$,
we associate to $\frac{1}{z_{n+1}-\frac{\pd F_0}{\pd t_0}}$ the diagram
$
\xy
(0,0); (5,0), **@{.};
(0,0)*+{\circ}; (5, 0)*+{\bullet};
(-6,0)*+{z_{n+1}};
\endxy
$;
note
\ben
\frac{\pd}{\pd t_0} \frac{1}{z_{n+1}-\frac{\pd F_0}{\pd t_0}}
= \frac{1}{(z_{n+1}-\frac{\pd F_0}{\pd t_0})^2} \cdot \frac{\pd^2F_0}{\pd t_0^2},
\een
we associate to it the diagram $
\xy
(0,0); (5,0), **@{.}; (15,0), **@{-};
(0,0)*+{\circ}; (5, 0)*+{\bullet};
(-6,0)*+{z_{n+1}},
\endxy
$;
i.e., we graft an edge at the $\bullet$-vertex of $
\xy
(0,0); (5,0), **@{.};
(0,0)*+{\circ}; (5, 0)*+{\bullet};
(-6,0)*+{z_{n+1}};
\endxy
$.
Next we note
\ben
&& \frac{\pd^2}{\pd t_0^2} \frac{1}{z_{n+1}-\frac{\pd F_0}{\pd t_0}}
=  \frac{2!}{(z_{n+1}-\frac{\pd F_0}{\pd t_0})^3} \cdot \biggl(\frac{\pd^2F_0}{\pd t_0^2}\biggr)^2
+  \frac{1}{(z_{n+1}-\frac{\pd F_0}{\pd t_0})^2} \cdot \frac{\pd^3F_0}{\pd t_0^3},
\een
this process can be graphically represented as follows:
$$
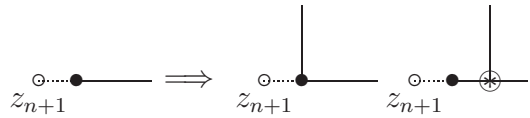
\xy
(0,0); (5,0), **@{.}; (15,0), **@{-};
(0,0)*+{\circ}; (5, 0)*+{\bullet};
(0,-3)*+{z_{n+1}}; (20,0)*+{\Longrightarrow};
(30,0); (35,0), **@{.}; (45,0), **@{-}; (35,0); (35,10), **@{-};
(30,0)*+{\circ}; (35, 0)*+{\bullet};
(30,-3)*+{z_{n+1}};
(50,0); (55,0), **@{.}; (65,0), **@{-}; (60,0); (60,10), **@{-};
(50,0)*+{\circ}; (55, 0)*+{\bullet}; (60,0)*+{\circledast};
(50,-3)*+{z_{n+1}};
\endxy
$$
After taking another derivative,
we get the following grafting rule for a $\circledast$-vertex of valence $3$:
$$
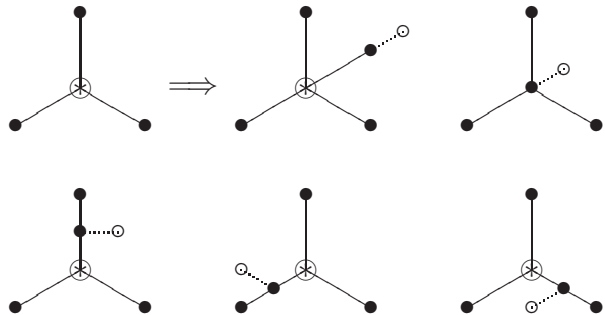
\xy
(0, 10); (0,0), **@{-}; (-8.6,-5), **@{-}; (0,0); (8.6,-5), **@{-};
(0, 10)*+{\bullet}; (0,0)*+{\circledast}; (-8.6,-5)*+{\bullet}; (8.6,-5)*+{\bullet};
(15,0)*+{\Longrightarrow};
(30, 10); (30,0), **@{-}; (21.4,-5), **@{-};   (38.6,-5); (30,0), **@{-};
(38.6,5), **@{-}; (42.9,7.5), **@{.}; (38.6,5)*+{\bullet}; (42.9,7.5)*+{\circ};
(30, 10)*+{\bullet}; (30,0)*+{\circledast}; (21.4,-5)*+{\bullet}; (38.6,-5)*+{\bullet};
(60, 10); (60,0), **@{-}; (51.4,-5), **@{-}; (60,0); (68.6,-5), **@{-};
(60, 10)*+{\bullet}; (60,0)*+{\bullet}; (51.4,-5)*+{\bullet}; (68.6,-5)*+{\bullet};
(60,0); (64.3,2.5), **@{.}; (64.3,2.5)*+{\circ};
\endxy$$
$$\xy
(0, 10); (0,0), **@{-}; (-8.6,-5), **@{-}; (0,0); (8.6,-5), **@{-};
(0, 10)*+{\bullet}; (0,0)*+{\circledast}; (-8.6,-5)*+{\bullet}; (8.6,-5)*+{\bullet};
(0,5); (5,5), **@{.}; (0,5)*+{\bullet}; (5,5)*+{\circ};
(30, 10); (30,0), **@{-}; (21.4,-5), **@{-};   (38.6,-5); (30,0), **@{-};
(25.7,-2.5); (21.4,0), **@{.}; (25.7,-2.5)*+{\bullet}; (21.4,0)*+{\circ};
(30, 10)*+{\bullet}; (30,0)*+{\circledast}; (21.4,-5)*+{\bullet}; (38.6,-5)*+{\bullet};
(60, 10); (60,0), **@{-}; (51.4,-5), **@{-}; (60,0); (68.6,-5), **@{-};
(60, 10)*+{\bullet}; (60,0)*+{\circledast}; (51.4,-5)*+{\bullet}; (68.6,-5)*+{\bullet};
(64.3,-2.5); (60,-5), **@{.}; (64.3,-2.5)*+{\bullet}; (60,-5)*+{\circ};
\endxy$$
In general,
by taking derivatives, one gets the following grafting rules at a $\circledast$-vertex.
(1) Grafting on the $\circledast$-vertex:
$$
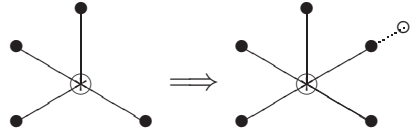
\xy
(0, 10); (0,0), **@{-}; (-8.6,-5), **@{-}; (-8.6,5); (8.6,-5), **@{-};
(0, 10)*+{\bullet}; (0,0)*+{\circledast}; (-8.6,-5)*+{\bullet}; (8.6,-5)*+{\bullet}; (-8.6, 5)*+{\bullet};
(15,0)*+{\Longrightarrow};
(30, 10); (30,0), **@{-}; (21.4,-5), **@{-};  (21.4,5); (38.6,-5), **@{-};  (38.6,-5); (30,0), **@{-};
(38.6,5), **@{-}; (42.9,7.5), **@{.}; (38.6,5)*+{\bullet}; (42.9,7.5)*+{\circ};
(30, 10)*+{\bullet}; (30,0)*+{\circledast}; (21.4,-5)*+{\bullet}; (38.6,-5)*+{\bullet}; (21.4,5)*+{\bullet};
\endxy$$
(2) Replacing $\circledast$-vertex by $\bullet$-vertex:
$$
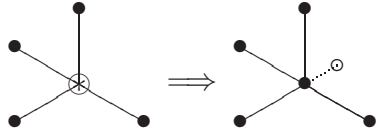
\xy
(0, 10); (0,0), **@{-}; (-8.6,-5), **@{-}; (-8.6,5); (8.6,-5), **@{-};
(0, 10)*+{\bullet}; (0,0)*+{\circledast}; (-8.6,-5)*+{\bullet}; (8.6,-5)*+{\bullet}; (-8.6, 5)*+{\bullet};
(15,0)*+{\Longrightarrow};
(30, 10); (30,0), **@{-}; (21.4,-5), **@{-};  (21.4,5); (38.6,-5), **@{-};  (38.6,-5); (30,0), **@{-};
(34.3,2.5), **@{.}; (34.3,2.5)*+{\circ};
(30, 10)*+{\bullet}; (30,0)*+{\bullet}; (21.4,-5)*+{\bullet}; (38.6,-5)*+{\bullet}; (21.4,5)*+{\bullet};
\endxy$$
(3) Grafting along an edge:
$$
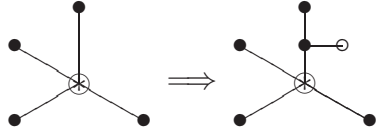
\xy
(0, 10); (0,0), **@{-}; (-8.6,-5), **@{-}; (-8.6,5); (8.6,-5), **@{-};
(0, 10)*+{\bullet}; (0,0)*+{\circledast}; (-8.6,-5)*+{\bullet}; (8.6,-5)*+{\bullet}; (-8.6, 5)*+{\bullet};
(15,0)*+{\Longrightarrow};
(30, 10); (30,0), **@{-}; (21.4,-5), **@{-};  (21.4,5); (38.6,-5), **@{-};  (38.6,-5); (30,0), **@{-};
(30,5); (35,5), **@{-};   (30,5)*+{\bullet}; (35,5)*+{\circ};
(30, 10)*+{\bullet}; (30,0)*+{\circledast}; (21.4,-5)*+{\bullet}; (38.6,-5)*+{\bullet}; (21.4,5)*+{\bullet};
\endxy$$
(4) Splitting of the $\circledast$-vertex: There are two kinds of splittings as indicated below:
$$
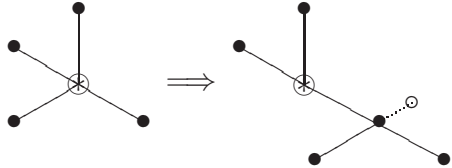
\xy
(0, 10); (0,0), **@{-}; (-8.6,-5), **@{-}; (-8.6,5); (8.6,-5), **@{-};
(0, 10)*+{\bullet}; (0,0)*+{\circledast}; (-8.6,-5)*+{\bullet}; (8.6,-5)*+{\bullet}; (-8.6, 5)*+{\bullet};
(15,0)*+{\Longrightarrow};
(30, 10); (30,0), **@{-}; (48.6,-10), **@{-}; (21.4,5); (30,0), **@{-};
(40,-5); (31.4,-10), **@{-};  (40,-5); (44.3,-2.5), **@{.};
(48.6,-10)*+{\bullet}; (44.3,-2.5)*+{\circ}; (31.4,-10)*+{\bullet};
(30, 10)*+{\bullet}; (30,0)*+{\circledast}; (21.4,5)*+{\bullet}; (40,-5)*+{\bullet};
\endxy$$
and
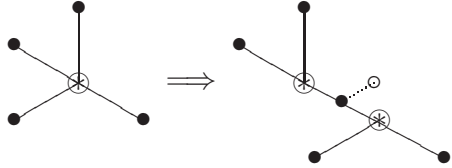
$$\xy
(0, 10); (0,0), **@{-}; (-8.6,-5), **@{-}; (-8.6,5); (8.6,-5), **@{-};
(0, 10)*+{\bullet}; (0,0)*+{\circledast}; (-8.6,-5)*+{\bullet}; (8.6,-5)*+{\bullet}; (-8.6, 5)*+{\bullet};
(15,0)*+{\Longrightarrow};
(30, 10); (30,0), **@{-}; (48.6,-10), **@{-}; (21.4,5); (30,0), **@{-};
(40,-5); (31.4,-10), **@{-}; (35,-2.5); (39.3,0), **@{.};
(48.6,-10)*+{\bullet}; (39.3,0)*+{\circ}; (31.4,-10)*+{\bullet}; (35,-2.5)*+{\bullet};
(30, 10)*+{\bullet}; (30,0)*+{\circledast};  (40,-5)*+{\circledast}; (21.4,5)*+{\bullet};
\endxy$$

In summary,
by induction $W_{0,n+1}(z_1, \dots, z_n; \bt)$ can be written as a summation over Feynman diagrams.
One also has to take care of the issue of the automorphism groups. That is left to the interested reader.

\subsection{Feynman rules for genus one $n$-point functions}
We have shown that
\be
F_1 = \half \log \frac{1}{1-I_1} = \half \log \frac{\pd^2F_0}{\pd t_0^2}.
\ee
Applying the loop operator $B(z_1)$:
\ben
W_{1,1}(z_1; \bt) & = & B(z_1) F_1
= \half \biggl( \frac{\pd^2F_0}{\pd t_0^2} \biggr)^{-1} \cdot \frac{\pd^2}{\pd t_0^2}
\frac{1}{z_1- \frac{\pd F_0}{\pd t_0} } \\
& =  & \half
\frac{1}{(z_1-\frac{\pd F_0}{\pd t_0})^2}
\biggl( \frac{\pd^2F_0}{\pd t_0^2} \biggr)^{-1}  \frac{\pd^3F_0}{\pd t_0^3}
+ \half\frac{2!}{(z_1-\frac{\pd F_0}{\pd t_0})^3} \frac{\pd^2F_0}{\pd t_0^2}.
\een
The two terms on the right-hand side correspond to the following two diagrams:
$$
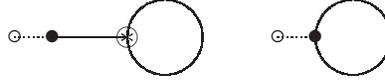
\xy
(-20,0); (-15,0), **@{.}; (-5,0), **@{-};
(-0,0)*\xycircle(5,5){}; (-5, 0)*+{\circledast};  (-20,0)*+{\circ}; (-15,0)*+{\bullet};
(15,0); (20,0), **@{.};   (25,0)*\xycircle(5,5){};
(15, 0)*+{\circ};  (20,0)*+{\bullet};
\endxy
$$
After applying the loop operator $B(z_2)$ to $W_{1,1}(z_1;\bt)$ we get
\ben
W_{1,2}(z_1,z_2;\bt)
& = & 2 \cdot \half
\frac{2!}{(z_1-\frac{\pd F_0}{\pd t_0})^3} \frac{1}{(z_2-\frac{\pd F_0}{\pd t_0})^2}
 \frac{\pd^3F_0}{\pd t_0^3} \\
& + & 2 \cdot \half
\frac{2!}{(z_2-\frac{\pd F_0}{\pd t_0})^3} \frac{1}{(z_1-\frac{\pd F_0}{\pd t_0})^2}
 \frac{\pd^3F_0}{\pd t_0^3} \\
& + & \half \frac{1}{(z_1-\frac{\pd F_0}{\pd t_0})^2}
\frac{3!}{(z_2-\frac{\pd F_0}{\pd t_0})^4}  \biggl(\frac{\pd^2F_0}{\pd t_0^2}\biggr)^2  \\
& + & \half \frac{1}{(z_2-\frac{\pd F_0}{\pd t_0})^2}
\frac{3!}{(z_1-\frac{\pd F_0}{\pd t_0})^4}  \biggl(\frac{\pd^2F_0}{\pd t_0^2}\biggr)^2  \\
& - & \half \frac{1}{(z_1-\frac{\pd F_0}{\pd t_0})^2}
\frac{1}{(z_2-\frac{\pd F_0}{\pd t_0})^2}  \biggl(\frac{\pd^3F_0}{\pd t_0^3}\biggr)^2.
\een
They correspond to the following diagrams:
$$
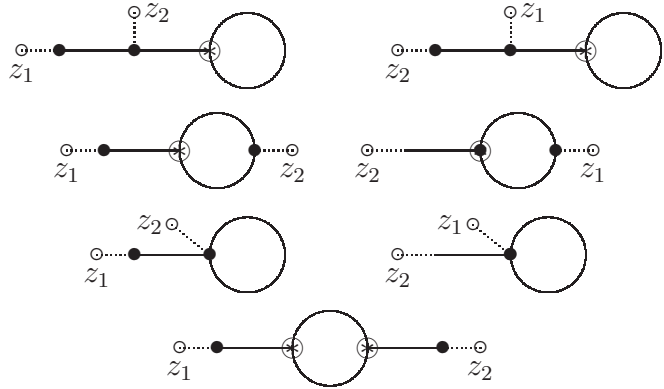
\xy
(-30,0); (-25,0), **@{.}; (-15,0), **@{-}; (-5,0), **@{-}; (-15,0),; (-15,5), **@{.};
(0,0)*\xycircle(5,5){}; (-5, 0)*+{\circledast};  (-30,0)*+{\circ};(-25,0)*+{\bullet};
(-15,0)*+{\bullet}; (-15,5)*+{\circ}; (-12,5)*+{z_2}; (-30,-3)*+{z_1};
(20,0); (25,0), **@{.}; (35,0), **@{-}; (45,0), **@{-}; (35,0),; (35,5), **@{.};
(50,0)*\xycircle(5,5){}; (45, 0)*+{\circledast};  (20,0)*+{\circ};(25,0)*+{\bullet};
(35,0)*+{\bullet}; (35,5)*+{\circ}; (38,5)*+{z_1}; (20,-3)*+{z_2};
\endxy
$$

$$
\xy
(-20,0); (-15,0), **@{.}; (-5,0), **@{-};
(0,0)*\xycircle(5,5){}; (-5, 0)*+{\circledast};  (-20,0)*+{\circ}; (-15,0)*+{\bullet};
(5,0); (10,0), **@{.}; (5,0)*+{\bullet}; (10,0)*+{\circ};
(-20,-3)*+{z_1}; (10, -3)*+{z_2};
(20,0); (25,0), **@{.}; (35,0), **@{-};
(40,0)*\xycircle(5,5){}; (35, 0)*+{\circledast};  (20,0)*+{\circ}; (35,0)*+{\bullet};
(45,0); (50,0), **@{.}; (45,0)*+{\bullet}; (50,0)*+{\circ};
(20,-3)*+{z_2}; (50, -3)*+{z_1};
\endxy
$$

$$
\xy
(-20,0); (-15,0), **@{.}; (-5,0), **@{-}; (-5,0); (-10,4), **@{.};
(0,0)*\xycircle(5,5){}; (-5, 0)*+{\bullet};  (-20,0)*+{\circ}; (-15,0)*+{\bullet}; (-10,4)*{\circ};
(-20,-3)*+{z_1}; (-13, 4)*+{z_2};
(20,0); (25,0), **@{.}; (35,0), **@{-}; (35,0); (30,4), **@{.};
(40,0)*\xycircle(5,5){}; (35, 0)*+{\bullet};  (20,0)*+{\circ}; (35,0)*+{\bullet}; (30,4)*{\circ};
(20,-3)*+{z_2}; (27, 4)*+{z_1};
\endxy
$$

$$
\xy
(-20,0); (-15,0), **@{.}; (-5,0), **@{-};
(0,0)*\xycircle(5,5){}; (-5, 0)*+{\circledast};  (-20,0)*+{\circ}; (-15,0)*+{\bullet};
(5,0); (15,0), **@{-}; (20,0), **@{.}; (5,0)*+{\circledast}; (15,0)*+{\bullet}; (20,0)*+{\circ};
(-20,-3)*+{z_1}; (20, -3)*+{z_2};
\endxy
$$
By repeatedly applying the loop operators one can write down the Feynman rules for
$W_{1, n}(z_1, \dots, z_n; \bt)$.

One can also use \eqref{eqn:F2-In-I-coord}, \eqref{eqn:F3-In-I-coord} etc. to derive Feynman rules
for $W_{g, n}(z_1, \dots, z_n; \bt)$ ($g \geq 2$).

\section{Spectral Curve and Its Special Deformation for Topological 1D Gravity}

\label{sec:Spectral Curve}

In this Section we show that the genus zero one-point function
combined with the gradient of the action function gives rise to the spectral curve
and its special deformation for topological 1D gravity.
We also establish a uniqueness result for the special deformation.

\subsection{Spectral curve and its special deformation}

We define the spectral curve of the topological 1D gravity  by
\be
y = \frac{1}{\sqrt{2}} \frac{\pd S(z, \bt)}{\pd z} + \sqrt{2} W_{0,1}(z; \bt),
\ee
or more concretely:
\be
y =  \frac{1}{\sqrt{2}} \sum_{n=0}^\infty \frac{t_n -\delta_{n,1}}{n!} z^n
+ \frac{\sqrt{2}}{z} + \sqrt{2}
\sum_{n = 1}^\infty \frac{n!}{z^{n+1}} \frac{\pd F_0}{\pd t_{n-1}}.
\ee
The reason for the artificial factor of $\sqrt{2}$ above will be made clear in the next Section.
Our computations for $W_{0,1}$ above yields:
\be
y = \frac{1}{\sqrt{2}} \sum_{n=0}^\infty \frac{t_n -1}{n!} z^n
+ \frac{\sqrt{2}}{z - \frac{\pd F_0}{\pd t_0}},
\ee
and when restricted to the line $t_n =0$ for $n \geq 1$,
\be
y = - \frac{1}{\sqrt{2}}(z- t_0) + \frac{\sqrt{2}}{z-t_0}.
\ee
This is a deformation of the following curve:
\be
y = - \frac{1}{\sqrt{2}} z + \frac{\sqrt{2}}{z},
\ee
which we call the {\em signed Catalan curve}.
Note
\be
\frac{z}{\sqrt{2}} = \frac{-y+\sqrt{y^2+4}}{2} = \sum_{n=0}^\infty \frac{(-1)^n}{n+1}\binom{2n}{n} y^{-2n-1}.
\ee
The coefficients of the series on the right-hand side are Catalan numbers $\frac{1}{n+1}\binom{2n}{n}$
with signs $(-1)^n$.

\begin{thm} \label{thm:Existence}
Consider the following series:
\be
y = \frac{1}{\sqrt{2}} \sum_{n=0}^\infty \frac{t_n -\delta_{n,1}}{n!} z^n
+ \frac{\sqrt{2}}{z} + \sqrt{2}
\sum_{n = 1}^\infty \frac{n!}{z^{n+1}} \frac{\pd F_0}{\pd t_{n-1}}.
\ee
Then one has:
\be
\half (y^2)_- =  W_{0,1}(z; \bt)^2.
\ee
Here for a formal series $\sum_{n \in \bZ} a_n f^n$,
\be
(\sum_{n \in \bZ} a_n f^n)_+ = \sum_{n \geq 0} a_n f^n, \;\;\;\;
(\sum_{n \in \bZ} a_n f^n)_ - = \sum_{n < 0} a_n f^n.
\ee
\end{thm}

\begin{proof}
This is actually equivalent to the Virasoro constraints for $F_0$.
Indeed,
\ben
\frac{y^2}{2} & = & \biggl( \half \sum_{n=0}^\infty \frac{t_n -\delta_{n,1}}{n!} z^n
+ \frac{1}{z} +
\sum_{n = 1}^\infty \frac{n!}{z^{n+1}} \frac{\pd F_0}{\pd t_{n-1}} \biggr)^2  \\
& = & \frac{1}{4} \biggl( \sum_{n=0}^\infty \frac{t_n -\delta_{n,1}}{n!} z^n \biggr)^2
+ \sum_{n=0}^\infty \frac{t_n -\delta_{n,1}}{n!} z^n
\biggl( \frac{1}{z} + \sum_{n = 1}^\infty \frac{n!}{z^{n+1}} \frac{\pd F_0}{\pd t_{n-1}} \biggr) \\
& + & \biggl( \frac{1}{z} + \sum_{n = 1}^\infty \frac{n!}{z^{n+1}} \frac{\pd F_0}{\pd t_{n-1}} \biggr)^2.
\een
It follows that
\ben
\half (y^2)_- & = & \frac{1}{z} +  \sum_{n=0}^\infty \frac{t_n -\delta_{n,1}}{n!}  \frac{\pd F_0}{\pd t_{n-1}} \\
& + & \sum_{m \geq 1}\sum_{n=0}^\infty  \frac{(n+m)!}{n!z^{m+1}} (t_n -\delta_{n,1})\frac{\pd F_0}{\pd t_{n+m-1}}   \\
& + & \biggl( \frac{1}{z} + \sum_{n = 1}^\infty \frac{n!}{z^{n+1}} \frac{\pd F_0}{\pd t_{n-1}} \biggr)^2.
\een
The proof is complted by Virasoro constraints for $F_0$.
\end{proof}

\subsection{Uniqueness of special deformation of the Airy curve}

Let us first prove a simple combinatorial result.

\begin{thm} \label{thm:Uniqueness}
There exists a unique series
\be
y = \frac{1}{\sqrt{2}} \sum_{n \geq 0} (v_n-\delta_{n,1}) z^{n} + \frac{\sqrt{2}}{z}
+ \sqrt{2} \sum_{n \geq 0} w_n z^{-n-2}
\ee
such that
each $w_n \in \bC[[v_0, v_1, \dots]]$ and
\be \label{eqn:y2-}
(y^2)_- =  \biggl(\frac{1}{z} +  \sum_{n \geq 0} w_n z^{-n-2}\biggr)^2.
\ee
\end{thm}

\begin{proof}
We begin by rewriting \eqref{eqn:y2-} as a sequence of equations:
\bea
&& w_0 =  v_0 + v_1 w_0 + v_2 w_1 + v_3 w_2 + \cdots, \label{eqn:Rec-w0} \\
&& w_1 = \;\;\;\; \;\;\;\;\;\; v_0 w_0 + v_1 w_1 + v_2w_2 + \cdots, \label{eqn:Rec-w1} \\
&& w_2 = \;\;\;\; \;\;\;\;\;\;\;\;\;\;\;\;\;\;\;\;\;\; v_0 w_1 + v_1 w_2 + \cdots, \label{eqn:Rec-w2} \\
&& w_3 = \;\;\;\; \;\;\;\;\;\;\;\;\;\;\;\;\;\;\;\;\;\; \;\;\;\;\;\;\;\;\;\;\; v_0 w_2  + \cdots, \label{eqn:Rec-w3} \\
&& \cdots\cdots\cdots\cdots \cdots
\eea
Write
$$
w_n = w_n^{(0)} + w_n^{(1)} + \cdots,
$$
where each $w^{(k)}_n$ consists of monomials in $v_0, v_1, \dots$ of  degree $k$.
Using such decompositions,
one can deduce by induction from the above system of equations:
\be
w_n^{(j)} = 0, \;\;\; n \geq 0, \;\; j =0, \dots, n,
\ee
and furthermore,
\ben
&& w_0^{(1)} = v_0, \\
&& w_m^{(n)} = \sum_{j=0}^\infty v_j w_{j+m-1}^{n-1}, \; m \geq 0, \; n \geq m+2.
\een
It follows that one can recursively determine all $w_m^{(n)}$ from the initial value $w_0^{(1)} =  v_0$.
\end{proof}

By combining Theorem \ref{thm:Existence} with Theorem \ref{thm:Uniqueness},
we then get:

\begin{thm}
For a series of the form
\be
y =  \frac{1}{\sqrt{2}} \frac{\pd S(z, \bt)}{\pd z} +  \sum_{n \geq 0} w_n z^{-n-2},
\ee
where each $w_n \in \bC[[t_0, t_1, \dots]]$,
the equation
\be
(y^2)_- = \biggl(\frac{1}{z} +  \sum_{n \geq 0} w_n z^{-n-2}\biggr)^2
\ee
 has a unique solution given by:
\ben
&& y = \frac{1}{\sqrt{2}} \sum_{n=0}^\infty \frac{t_n -\delta_{n,1}}{n!} z^n
+ \frac{\sqrt{2}}{z} + \sqrt{2}
\sum_{n = 1}^\infty \frac{n!}{z^{n+1}} \frac{\pd F_0}{\pd t_{n-1}}.
\een
where $F_0(\bu)$ is the free energy of the 1D topological gravity in genus zero.
\end{thm}

\section{Quantum Deformation Theory of the Spectral Curve for Topological 1D Gravity}

\label{sec:Quantum-Deformation-Theory}

We have already shown that the free energy in genus zero of 1D topological quantum gravity
can be used to produce special deformation of its spectral curve.
In this Section we will see that this deformation lead to a quantization of the spectral curve that
can be used to recover the free energy in all genera.

\subsection{Symplectic reformulation of the special deformation}

One can formally understand $y$ as a field on the spectral curve.
Consider the space consisting of  elements  of the form:
\be
\frac{1}{\sqrt{2}} \sum_{n =0}^\infty \frac{\tilde{u}_n}{n!} z^{n} + \frac{\sqrt{2}}{z}
+ \sqrt{2} \sum_{n =1}^\infty \tilde{v}_{n-1} \frac{n!}{z^{n+1}}.
\ee
We regard $\{ \tilde{u}_n, \tilde{v}_n\}$ as linear coordinates on $V$,
and introduce the following symplectic structure on $V$:
\be
\omega = \sum_{n =0}^\infty   d\tilde{u}_n \wedge d \tilde{v}_n.
\ee
It follows that
\be
\tilde{v}_n = \frac{\pd F_0}{\pd u_n}(\bu)
\ee
defines a Lagrangian submanifold in $V$,
and so does
\be
\tilde{v}_n = \frac{\pd (\lambda^2F)}{\pd u_n}(\bu).
\ee
In other words,
free energy in all genera produces a deformation of a Lagrangian sumanifold.

\subsection{Canonical quantization of the special deformation of spectral curve}

Take the natural polarization that $\{q_n = \tilde{u}_n\}$ and $\{p_n = \tilde{v}_n\}$,
one can consider the canonical quantization:
\be
\hat{\tilde{u}}_n = \tilde{u}_n \cdot, \;\;\; \hat{\tilde{v}}_n = \frac{\pd}{\pd \tilde{u}_n}.
\ee
Corresponding to the field $y$,
consider the following fields of operators on the spectral curve:
\be \label{Def:hat(y)}
\hat{y} = \sum_{n=0}^\infty  \beta_{-n-1} z^{n}
+ \frac{\beta_0}{z} +  \sum_{n=0}^\infty  \beta_{n+1} z^{-n-2} ,
\ee
where the operators $\beta_m$ are defined by:
\be
\beta_{-(k+1)} = \lambda^{-2} \frac{1}{\sqrt{2}}\frac{\tilde{t}_k}{k!}  \cdot, \;\;\;\;
\beta_{k+1} = \lambda^2 \sqrt{2}(k+1)!\frac{\pd}{\pd t_k},
\;\;\;\; \beta_0 = \sqrt{2}.
\ee

\subsection{The bosonic Fock space}

As usual,
the operators $\{ \beta_{n+1}\}_{n \geq 0}$ are the annihilators
while the operators $\{\beta_{-(n+1)}\}_{n\geq 0}$ are the creators.
Given $\beta_{n_1+1}, \dots, \beta_{n_k+1}$,
their normally ordered products  are defined:
\be
:\beta_{n_1+1}, \dots, \beta_{n_k+1}:
= \beta_{n_1'+1} \cdots \beta_{n_k'+1},
\ee
where $n_1' \geq \cdots \geq n_k'$ is a rearrangement of $n_1, \dots, n_k$.
Denote $\vac$ the vector $1$.
The bosonic Fock space $\Lambda$ is the space spanned by elements of form
$\beta_{-(n_1+1)} \cdots \beta_{-(n_k+1)} \vac$,
where $n_1, \dots, n_k \geq 0$.
On this space one can define a Hermitian product by setting
\bea
&& \langle0 | 0 \rangle = 1, \\
&& \beta_{n+1}^* = \beta_{-(n+1)}.
\eea
For a linear operator $A: \Lambda  \to \Lambda $,
its vacuum expectation value is defined by:
\be
\corr{A} = \lvac A \vac.
\ee

\subsection{Regularized products of two fields}

We now study the product of the fields $\hat{y}(z)$ with $\hat{y}(z)$.
This cannot be defined directly,
because for example,
\be
\lvac \hat{y}(z) \hat{y}(z) \vac
= \frac{2}{z^2} + \frac{1}{z^2} \sum_{n \geq 0} (n+1).
\ee
To fix this problem,
we follow the common practice in the physics literature by using the normally ordered products of fields
and regularization of the singular terms as follows.
First note
\ben
\hat{y}(z) \cdot \hat{y}(w)
& = & :\hat{y}(z)\hat{y}(w):
+ \sum_{n=0}^\infty (n+1) z^{-n-2}w^{n} \\
& = & :\hat{y}(z)\hat{y}(w): + \frac{1}{(z-w)^2}.
\een
It follows that
\be
\corr{ \hat{y}(z) \cdot \hat{y}(w) }
= \frac{1}{(z-w)^2},
\ee
hence
\be
\hat{y}(z) \cdot \hat{y}(w) = :\hat{y}(z) \cdot \hat{y}(w): + \corr{\hat{y}(z) \cdot \hat{y}(w)}.
\ee
Now we have
\ben
\hat{y}(z+\epsilon) \cdot \hat{y}(z)
& = & :\hat{y}(z+\epsilon)\hat{y}(z): + \frac{1}{\epsilon^2}.
\een
We define the regularized product of $\hat{y}(z)$ with itself by
\be \label{eqn:y(z)odot2}
\hat{y}(z) \odot \hat{y}(z) = \hat{y}(z)^{\odot 2}:
= \lim_{\epsilon \to 0} (\hat{y}(z+\epsilon) \hat{y}(z) - \frac{1}{\epsilon^2})
= :\hat{y}(z)\hat{y}(z): .
\ee
In other words,
we simply remove the term that goes to infinity as $\epsilon \to 0$,
and then take the limit.

\subsection{Virasoro constraints and mirror symmetry for 1D topological gravity}

The following result establishes the mirror symmetry of the theory of 1D topological gravity
and the quantum deformation theory of its spectral curve:

\begin{thm}
The partition function $Z$ of the topological 1D gravity is uniquely specified by the following equation:
\be \label{eqn:Virasoro-Operator-Field}
(\hat{y}(z)^{\odot 2})_- Z = 0.
\ee
\end{thm}

\begin{proof}
By the definition of $\hat{y}(z)$ and \eqref{eqn:y(z)odot2},
one gets:
\ben
\half (\hat{y}(z)^{\odot 2})_-
& = & (\beta_0 \beta_{-1} + \sum_{n=1}^\infty \beta_{-n-1}\beta_{n} ) z^{-1}
+  (\sum_{n =0}^\infty  \beta_{-n-1}\beta_{n+1}+ \frac{\beta_0^2}{2} ) z^{-2} \\
& + & \sum_{m \geq 1} (  \sum_{n=0}^\infty \beta_{-(n+1)} \beta_{n+m+1}
+ \half \sum_{\substack{j+k=m \\j, k \geq 0}} \beta_j\beta_k ) z^{-m-2}.
\een
It is then straightforward to see that \eqref{eqn:Virasoro-Operator-Field} is equivalent to
the Virasoro constraints \eqref{eqn:Virasoro-New--1}-\eqref{eqn:Virasoro-New-m}.
\end{proof}

\section{Concluding Remarks}

\label{sec:Conclusion}

In this paper we have focused mainly on the problems
of computing free energy, partition function and $n$-point functions
of topological 1D gravity.
Besides the flow equation and polymer equation that appeared long ago in the literature,
we have developed the techniques of changing coordinates to the $I$-coordinates.

In the process of computing $n$-point functions,
the importance of the role played by the loop operator has been made clear.
Furthermore,
in the study of spectral curve, its special deformation and its quantum deformation theory,
it has become clear that the gradient of the effective action function of topological
1D gravity and the loop operator can be combined into a free boson field,
and the partition function can be identified with a vector in the bosonic Fock space
uniquely specified by the Virasoro constraints,
determined themselves by the special deformation of the spectral curve.
This can be compared with an earlier work \cite{Zhou} in which we have done similar things
for topological 2D gravity.
Such phenomena can be regarded as examples of holography principle 
applied to the spectral curve.
The special deformation of the spectral curve can be detected by taking 
residues at infinity,
and the whole theory produces an element in the Fock space associated 
to the infinity.

This work provides the foundation for further developments of the theory of topological 1D gravity,
its generalizations to include topological matters, and generalizations and comparisons with the theory
of topological 2D gravity, etc.
We will report on these developments in forthcoming work.

\vspace{.1in}

{\em Acknowledgements}.
This research is partially supported by NSFC grant 11171174.


\begin{thebibliography}{999}

\bibitem{ADKMV}
M. Aganagic, R. Dijkgraaf, A. Klemm, M. Mari\~no, C. Vafa,
{\em Topological strings and integrable hierarchies},
Comm. Math. Phys.  261  (2006),  no. 2, 451-516.

\bibitem{Anderson-Myers-Periwal}
A. Anderson, R.C. Myers  and V. Periwal,
{\em Branched polymers  from  a double-scaling  limit of matrix model},
Nucl. Phys. B  360(1991) 463-479.

\bibitem{Bennett-Cochran-Safnuk-Woskoff}
J. Bennett, D. Cochran, B. Safnuk, K. Woskoff,
{\em Topological recursion for symplectic volumes of moduli spaces of curves},
Mich. Math. J. 61(2), 331-358 (2012).
arXiv:1010.1747.

\bibitem{Brezin-Kazakov}
 E. Br\'ezin  and V.A.  Kazakov,
 {\em Exactly solvable field theories of closed strings},
 Phys.  Lett. B  236 (1990), no. 2, 144-150.

\bibitem{Douglas-Shenker}
M.R. Douglas, S.H.  Shenker,
{\em Strings in less than one dimension},
Nucl. Phys.  B  335  (1990), no. 3, 635-654.

\bibitem{Gross-Migdal}
D.J.  Gross and  A.A. Migdal,
{\em A nonperturbative treatment of  two-diemensional quantum gravity},
Nucl. Phys.  B 340 (1990),  333-365..

\bibitem{Kontsevich} M. Kontsevich, {\em
Intersection theory on the moduli space of curves and the matrix
Airy function}. Comm. Math. Phys. {\bf 147} (1992), no. 1, 1--23.

\bibitem{Nishigaki-Yoneya1}
S.  Nishigaki,T. Yoneya,
{\em A nonperturbative theroy of polymer branching chains},
Nucl.  Phys.  B  348 (  1991  ), 787-807.

\bibitem{Nishigaki-Yoneya2}
S.  Nishigaki,T. Yoneya,
{\em The  double-scaling limit  of $O(N)$ vector  models
and the  KP  hierarchy},
Physics  Letters B 268 (1991), 35-39.

\bibitem{Di Vecchia-Kato-Ohta}
P.  Di  Vecchia, M.  Kato, N.  Ohta,
{\em Double scaling limit in $O(N)$ vector models},
Nucl. Phys.  B 357(1991), 495-520.

\bibitem{Dijkgraaf-Witten}
 R. Dijkgraaf,  E. Witten,
{\em Mean field theory, topological field theory, and multimatrix models},
 Nucl. Phys. B 342 (1990), 486-522.

\bibitem{Eguchi-Yamada-Yang}
T. Eguchi, Y. Yamada, S.-K. Yang,
{\em On the genus expansion in the topological string theory},
Rev. Math. Phys. 7 (1995) 279-309.
arXiv:hep-th/9405106.

\bibitem{Eynard}
B. Eynard,
{\em Recursion between Mumford volumes of moduli spaces}.
Ann. Henri Poincar\'e 12(8), 1431-1447 (2011).

\bibitem{Itzykson-Zuber}
C.Itzykson, J.-B.Zuber,
{\em Combinatorics of the modular group II: the Kontsevich integrals},
Int.J.Mod.Phys. A7 (1992) 5661-5705.

\bibitem{Witten} E. Witten,
{\it Two-dimensional gravity and intersection theory on moduli
space}, Surveys in Differential Geometry, vol.1, (1991) 243--310.

\bibitem{Zhou1}
J. Zhou,
{\em Topological recursions of Eynard-Orantin type for intersection numbers on moduli spaces of curves}.
Lett. Math. Phys.  103  (2013),  no. 11, 1191-1206.

\bibitem{Zhou2}
J. Zhou,
{\em Explicit formula for Witten-Kontsevich tau-function},
arXiv:1306.5429.


\bibitem{Zhou}
J. Zhou,
{\em Quantum deformation theory of the Airy curve and mirror symmetry of a point},
arXiv:1405.5296.




\end{thebibliography}
\end{document}